\documentclass{article}

\usepackage[utf8]{inputenc}
\usepackage{authblk}
\usepackage[T1]{fontenc}
\usepackage{mathtools}
\usepackage{paralist}
\usepackage{bbm}
\usepackage{array}
\newcolumntype{C}[1]{>{\centering\let\newline\\\arraybackslash\hspace{0pt}}m{#1}}
\newcolumntype{L}[1]{>{\raggedright\let\newline\\\arraybackslash\hspace{0pt}}m{#1}}

\usepackage{graphicx}
\usepackage{subcaption}

\usepackage{amsthm}
\usepackage{amssymb}
\usepackage{algorithm}  
\usepackage{algpseudocode}  

\usepackage{float}
\newfloat{algorithm}{t}{lop}

\iffalse
  \newcommand{\rnote}[1]{\marginpar{\small\color{red}#1}}
  \newcommand{\infloatrnote}[1]{{\color{red}[#1]}}
\else
  \newcommand{\rnote}[1]{}
  \newcommand{\infloatrnote}[1]{}
\fi

\let\emptyset\varnothing

\usepackage{multirow}
\usepackage{forloop}

\usepackage[pdftex,pdftitle={overlay-indexes and DB-trees}]{hyperref}
\usepackage[inline]{enumitem}
\usepackage{tabto}

\usepackage{listings}
\graphicspath{{figures/}}

\newtheorem{theorem}{Theorem}

\newtheorem{lemma}{Lemma}
\newtheorem{property}{Property}

\newcommand{\hash}[1]{\mathrm{hash}(#1)}

\newcommand{\removedtodo}[2][]{}

\newlength{\mywidth}
\algblockdefx{UponReceiving}{EndUponReceiving}[1]{\textbf{upon receiving} #1 \textbf{containing}}{}

\algblockdefx{Do}{EndDo}{\textbf{do}}{}

\algblockdefx{Indent}{EndIndent}{}{}
\algblockdefx{Status}{EndStatus}{\textbf{state}}{}

\algnewcommand{\LComment}[1]{\(\triangleright\) #1}

\renewcommand{\emph}[1]{\textit{#1}}

\begin{document}

\title{Overlay Indexes: Efficiently Supporting Aggregate Range Queries and Authenticated Data Structures in Off-the-Shelf Databases}

\author[a]{Diego Pennino}
\author[b]{Maurizio Pizzonia}
\author[c]{Alessio Papi}

\affil[abc]{Universit\`a degli Studi Roma Tre, Dipartimento di Ingegneria, Sezione Informatica e Automazione\\Via della Vasca Navale 79, 00146, Roma (Italy) }
\affil[a]{pennino@ing.uniroma3.it} 
\affil[b]{pizzonia@ing.uniroma3.it}
\affil[c]{alessio.papi@live.it}
\date{}

\maketitle

\begin{abstract}

Commercial off-the-shelf DataBase Management Systems (DBMSes) are highly optimized to
process a wide range of queries by means of carefully designed indexing and
query planning. However, many aggregate range queries are usually performed by
DBMSes using sequential scans, and certain needs, like storing 
Authenticated Data Structures (ADS), are not supported at all. Theoretically,
these needs could be efficiently fulfilled adopting specific kinds of indexing,
which however are normally ruled-out in DBMSes design.

We introduce the concept of \emph{overlay index}: an index that is meant to be
stored in a standard database, alongside regular data and managed by regular
software, to complement DBMS capabilities. We show a data structure, that we call
\emph{DB-tree}, that realizes an overlay index to support a wide range of
custom aggregate range queries as well as ADSes, efficiently. All DB-trees operations can be
performed by executing a small number of queries to the DBMS, that can be issued
in parallel in one or two \emph{query rounds}, and involves a logarithmic amount
of data. We experimentally evaluate the efficiency of DB-trees showing
\infloatrnote{C1.1} that our approach is effective, especially if data updates
are limited.

\end{abstract}


\section{Introduction }\label{sec:intro}

In relational and NoSQL databases, the ability to obtain aggregate information
from a (possibly large) set of ``records'' (tuples, documents, etc.) has always
been an important feature. Usually, the set on which to apply an
\emph{aggregation function} (e.g., COUNT or SUM in
SQL~\cite{eisenberg1999sql}) is identified by some form of record selection. In
the simplest form, records that have one of their fields in a given
range are selected and  the aggregation function is applied on them. This kind of
queries are usually referred to as \emph{aggregate range queries}. Applications of
aggregate range queries can be easily found in many data analytics activities
related to business intelligence, market analysis, user profiling, IoT, etc.
Further, this feature is one of the fundamental elements of a good system for big-data
analytics. In many applications, the speed at
which queries are fulfilled is often critical, possibly marking the distinction
between a system that meets the user needs and one that does not. Typical
examples are interactive visual systems, where users expect the system 
to respond in less than a second.

The vast majority of DataBase Management Systems (DBMSes) supports aggregate
range queries to some extent. 
When data are large, aggregation performed by a complete scan of the
selected data can be too costly to be viable. In theory, adopting proper data
structures for indexing~\cite{cormen2009introduction}, it is possible to answer
 aggregate range queries in $O(\log n)$ time, where $n$ is the amount of
selected records to be aggregated, for any selection range and for a quite
large class of aggregation functions\footnote{Essentially, for indexing techniques to be applicable, aggregation functions have to be
	\emph{associative}. In Section~\ref{ssec:db-tree-formal}, we provide a formal
	definition of  a class of aggregation functions that theoretically can be
	computed efficiently by proper indexing techniques.}. The general idea behind
those data structures is to store partial pre-computed aggregate values in each
node of the index. In this way, it is possible to answer aggregate range
queries on any range without actually scanning the data. Update operations on
these data structures also take logarithmic time, allowing them to be used even
when data is subject to updates.

However, in practice, indexes realized by many DBMSes are designed to support
regular (non aggregated) queries. This is reflected in the limited advantage
that regular indexes can provide to aggregate range queries. In particular, most
\mbox{DBMSes} can exploit regular indexes only for aggregation based on MIN and MAX
functions. Other typically DBMS-supported aggregation functions, like 
SUM, AVG, etc., usually require sequential scans. 
To speed up these
sorts of queries, a typical trick is to keep data in main memory, which is
costly and usually not possible for big-data applications. 

In this paper, we introduce the concept of \emph{overlay index}, which is a data
structure that plays the role of an index but is explicitly stored in a database
along with regular data. The logic to use an overlay index is not frozen in the
DBMS but can be programmed and customized at the same level of the application
logic, obtaining great flexibility.
However, designing an overlay index rises specific challenges.
In principle, any memory oriented data structure could be easily
represented in a database. Nonetheless, operations on these data structures
typically require traversing
a logarithmic number of elements, each of them pointing to the next one, which is
an inherently sequential task.
While this is fast and acceptable in main memory, sequentially performing a
logarithmic number of queries to a DBMS is extremely slow. In fact, each query
encompasses client/server communication, usually involving the network and the
operating system, which introduces a large cost that cannot be mitigated by the DBMS
query planner. Further,  
performing queries sequentially prevents the exploitation of the capability of RAID arrays and DBMS
clusters to fulfill many requests at the same time and of 
disk schedulers to
optimize the order of disk access~\cite{thomasian2011survey}.

Our main contribution is a new data structure, called \emph{DB-tree}, for the
realization of an overlay index. DB-trees are meant to be stored and used in
standard DBMSes to support custom aggregate range queries and possibly other
needs, like, for example, data authentication by Authenticated Data
Structures~\cite{tamassia2003authenticated} (ADS). DB-trees are a sort of search
trees whose balancing is obtained through randomization, as for skip lists~\cite{pugh1998skip}. In a
DB-tree, query operations require only a constant number of range selections,
that can be executed in parallel in a single  \emph{round} and that return a
logarithmic amount of data. Updates, insertions and deletions, also 
involve a logarithmic amount of data and can be executed in at most two rounds.
We formally describe all algorithms and prove their correctness
and efficiency.

Additionally, we present experimental evidence of the efficiency of our
approach, on two off-the-shelf DBMSes, by comparing the queries 
running times using DB-trees against the same operations performed with the only support of the DBMS.\rnote{C1.1} Experimental results
show that the adoption of DB-trees brings a large gain for range queries, but they may introduce a non negligible overhead when data changes.
We also discuss several applicative and architectural
aspects as well as some variations of DB-trees.

This paper is structured as follows. In Section~\ref{sec:soa}, we show the state
of the art. In Section~\ref{sec:overlay-indexes}, we introduce the concept of
overlay index and discuss several architectural aspects. DB-trees are introduced
in Section~\ref{sec:db-tree} along with some formal theoretical results.
Algorithms to query a DB-tree and perform insertions, deletions and updates are
shown in Section~\ref{sec:algorithms}.
In Section~\ref{sec:exepriments}, we show an experimental comparison between
DB-trees and plain DBMS. In Section~\ref{sec:groupBy}, we show how it is
possible to use DB-trees to perform, efficiently, aggregate range queries
grouped by values of a certain column. We also show an experimental comparison
with other approaches. In Section~\ref{sec:ADS}, we show how DB-trees can be
adapted to realize persistent ADSes. In Section~\ref{sec:arch_discussion}, we discuss some architecture generalizations. In
Section~\ref{sec:conclusions}, we draw the conclusions.



\section{State of the Art}\label{sec:soa}

In this section, we review the state of the art of technology and research about
aggregate range queries optimization. We also review the state of the art about
persistent representation of ADSes, which is a relevant application of the
results of this paper.

Optimization of query processing in DBMSes is a very classic subject in database
research, which historically also dealt with  the optimization of aggregate
queries (see for example~\cite{von1987translating}). \emph{Indexes} are primary
tools for the optimization of query execution. They are data structures internally used by DBMSes
to speed up searches. They are persisted on disk in a DBMS proprietary format.
Typical indexes realize some form of balanced search tree or hash
table~\cite{cormen2009introduction,ramakrishnan2000database}. \rnote{C2.4}Specific indexing 
techniques for uncertain data are also known, see for example~\cite{cheng2004efficient,angiulli2012indexing,chen2015efficient,chen2017indexing}. 
Some research effort was also dedicated to the creation of a general architecture for indexing, see, for example,~\cite{VLDB95*562,kornacker1997concurrency}.
However, these results have been adopted only by specific DBMSes (see, for example, \cite{borodin2017optimization}).
The query
optimizer of DBMSes can take into account the presence of indexes in the
planning phase if this is deemed favorable. In a typical DBMS, the creation of
an index is asked by the database administrator, or the application developer,
using proper constructs (available for example in SQL). The decisions about
indexes creation is usually based on the foreseen frequencies of
the queries, the involved data size, and execution time constraints. Several works deal with the possibility for a
DBMS to self tune and to choose right indexes (see, for
example,~\cite{chaudhuri2007self,kraska2019sagedb}). Usually, there is no way for the user of
the DBMS to access an index directly or to create indexes that are different
from the ones the DBMS is designed to support. \rnote{C2.4}A notable exception to this
is the PostgreSQL DBMS that provides some flexibility~\cite{postgresql12doc}.

Concerning aggregate range queries, in common \mbox{DBMSes}, regular indexing is
usually only effective for some aggregate functions, like MIN and MAX.
Sometimes, certain index usages can provide a great speed-up for aggregate range
queries due to the fact that putting certain data in the index may avoid random
disk access and/or most processing could occur
in-memory~\cite{stackoverflow1439016-oracle}. 
The current DBMS technology and the SQL standard do not allow the user to
specify custom indexes to obtain logarithmic time execution of aggregate range
queries, even for SQL-supported aggregation functions for which this would be theoretically possible.

A wide variety of techniques was proposed to optimize the execution of aggregate
range queries in DBMSes. A
whole class of proposals deal with \emph{materialized views} (see for
example~\cite{halevy2001answering,
	goldstein2001optimizing,gupta1993maintaining,srivastava1996answering,muller2013workload}). These
techniques require the DBMS to keep the results of certain queries stored and
up-to-date. These can be used to simplify the execution of certain aggregate
queries and are especially effective when data is not frequently updated.

A largely investigated approach is called \emph{Approximate Query Processing}
(AQP), which aims at gaining efficiency while reducing the precision of query results.
A survey of the achievements in this area is provided
in~\cite{li2018approximate}. This approach is relevant especially when data are
large. An approximate approach targeted to big-data environments is proposed
in~\cite{yun2014fastraq}. A method to perform approximate queries on granulated
data summaries of large datasets is shown in~\cite{slkezak2018new}. Approximated
techniques are now available on some widely used
systems~\cite{su2016approximate,chandramouli2013scalable}. The
VerdictDB~\cite{park2018verdictdb} system provides a handy way to support AQP
for SQL-based systems by adding a query/result-rewriting middle layer between
the client and the DBMS.  

A context in which aggregation speed is very relevant is in On-Line Analytical
Processing (OLAP) systems (see, for
example,~\cite{ho1997range,gupta1995aggregate}). Several works deal with fast
methods to obtain approximated
results~\cite{shanmugasundaram1999compressed,cuzzocrea2005providing,acharya1999aqua} for OLAP systems. In this field, specific indexing techniques may be adopted~\cite{sarawagi1997indexing}. Specific data structures were proposed to support aggregated range queries, like for example aR-trees~\cite{papadias2001efficient} for spatial OLAP systems.
The work in~\cite{lopez2005spatiotemporal} surveys several aggregation
techniques targeted to the storage of spatiotemporal data. 

To support strict time bounds, specific techniques for in-memory databases
exist~\cite{chavan2018accelerating}.

As will be clear in the following, one of the applications of our results is to
support \emph{Authenticated Data Structures} (\emph{ADS}) (see
Section~\ref{sec:ADS}). ADSes are useful when we need a cryptographic
proof that the result of a query is coherent with a certain version of the
dataset and that version is succinctly known to the client by a cryptographic hash that the client trusts. The first ADS was proposed by
Merkle~\cite{merkle1987digital}. Merkle trees are balanced search trees in which
each node contains a cryptographic hash of its children and, recursively, of the
whole subtree. 
The work
\rnote{C2.4} in~\cite{pennino2019pipeline} shows how to arbitrarily scale the throughput of 
an ADS-based system in a cloud setting with an arbitrarily large number of clients.
ADSes are especially desirable in the context of \emph{outsourced
	databases}. 
The work in~\cite{li2006dynamic} studies data structures to realize
authenticated versions of B-trees to authenticate queries. This proposal is
meant to be used in DBMSes as an internal indexing structure. Other works tried
to represent ADSes in the database itself. Several general-purpose 
techniques to represent a tree are presented in~\cite{celko2012joe}. The problem of
representing an \emph{authenticated skip list}~\cite{tamassia2003authenticated}
in a relational table was investigated in~\cite{di2007authenticated}. They
propose the use of \emph{nested sets} to perform queries in one query round. An
efficient use of this approach to compute authenticated replies for a wide class
of SQL queries is provided in~\cite{palazzi2010query}. 
Nested sets represents nodes
of trees as intervals bounded by integers and a parent-child relation is
represented as a containment relation between two intervals. Unfortunately,
nested sets cannot be updated efficiently. In~\cite{tropashko2005nested},
several variations of nested sets are described, varying the way in which
intervals bounds are represented. These methods are based on numerical
representations of rational numbers and are limited by precision problems.

\section{Overlay Indexes}\label{sec:overlay-indexes}

In this section, we discuss the rationale for introducing overlay indexes, the
applicative contexts where overlay indexes can be fruitfully applied and
discuss some architectural aspects. This discussion is largely independent from the
DB-tree data structure, which we propose as one form of realization of overlay indexes, described in Section~\ref{sec:db-tree}.

In the following, when we refer to a \emph{query}, we may intend 
either a proper data-reading query or a generic statement, which can also change the data. The distinction should be clear form the context.

\begin{figure}
	\centering
	\includegraphics[width=\linewidth]{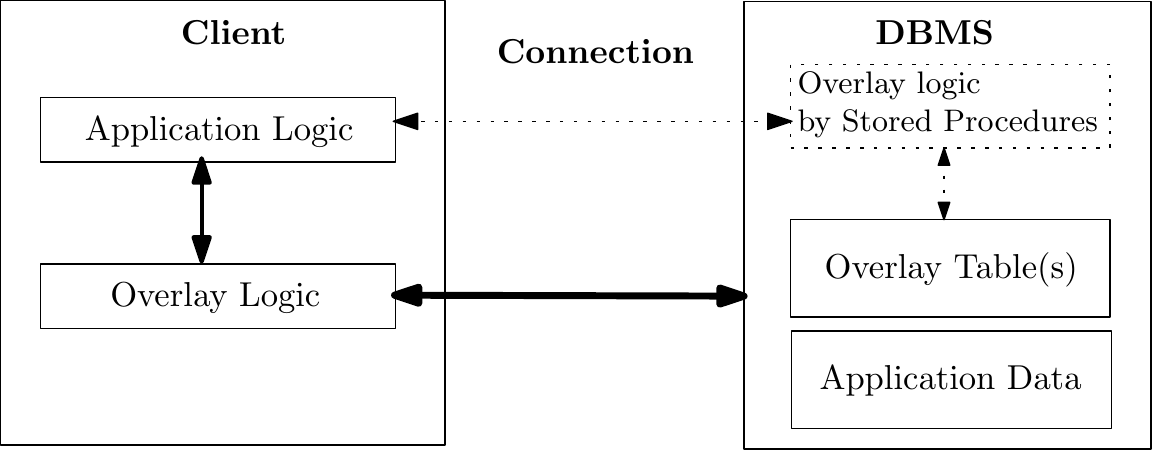}
	\caption{A schematic architecture for the adoption of overlay indexes. Thick arrows 
	shows the data path when overlay logic is realized client-side. Dotted arrows shows 
    data path when overlay logic is realized by stored procedures in the DBMS.}
	\label{fig:architecture}
\end{figure}

In Figure~\ref{fig:architecture},  we show a schematic architecture for the adoption of overlay indexes. We suppose to
have a \emph{client} and a \emph{DBMS} connected by a \emph{connection}. In our
scheme, the client contains the \emph{application logic} and performs queries to the
DBMS through the connection. Note that, we use this scheme in the following
discussion but we do not intend to restrict the use of overlay indexes to this
simple architecture. For example, the client may be the middle tier
server that performs DBMS queries on behalf of user clients, the storage of the
DBMS may be located in a cloud, the client may be on the same machine of the
DBMS and the connection may be just inter-process communication, and
so on. The protocol used to perform queries to the DBMS is a standard one (e.g.,
SQL or one of the many NoSQL query languages). Additional requirements on features supported by the DBMS
depend on  the way overlay indexes are realized. We provide these details for DB-trees in
Section~\ref{ssec:db-tree-formal}.

\subsection{Applicative contexts} 

The applicative contexts where we believe
overlay indexes may turn out to be of great help are those where
\begin{inparaenum}[(1)]
\item the application performs both updates (possibly comprising insertion and deletion) and reads on the database, %
\item some read operations have strong efficiency requirements, like, for example,
in interactive applications, %
\item the amount of data is so large that in-memory solutions cannot be applied and, 
hence, the efficiency of those read
operations can only be obtained by adopting indexes, and %
\item  the DBMS does not support the right indexing for those read operations.
\end{inparaenum}
For example, an application might need to provide the sum of a column for those
tuples having the value of their key contained in a certain range, which depends
on (unpredictable) user requests, or might need to provide a
cryptographic hash of the whole table to the client for security purposes. The
first case, is often inefficiently supported and often efficiency is obtained
not by proper indexing but by keeping data in main memory. The second case is usually not
supported by DBMSes. Clearly, in the above described conditions, approximate query processing
could be adopted. However, in certain situations approximation is not desirable or
is ruled-out by the nature of the aggregation function, as for the cryptographic hash case.

\subsection{Rationale} 

Clearly, the introduction of an overlay index has some
performance overhead when data are updated, as any indexing approach has, but
can be the only viable approach for certain non-supported aggregation functions
or can dramatically reduce the time taken to reply to certain queries for which
specific indexing is not available from the DBMS. For example, in the cases
where the DBMS cannot use indexes, an aggregated query may run in $O(r)$
where $r$ is the number of elements in the range selected by the query. This is
essentially due to the fact that the DBMS has no better strategy than
sequentially scan the result of the selection. For an overlay index realized as
a DB-tree (see Section~\ref{sec:db-tree}), both queries and updates take $O(\log
r)$. The speed-up obtained for the aggregate range queries may be huge, 
if $r$ is even moderately large, while paying a logarithmic slow-down on the 
update side is usually affordable, especially if data are rarely updated
 (see also Section~\ref{sec:exepriments}).

\subsection{Architectural Aspects}\label{ssec:architectural-aspects}

We refer to Figure~\ref{fig:architecture}.
Realizing an overlay index requires to introduce
\begin{inparaenum}[(i)]
\item one or more specific table(s) into the database, which we call \emph{overlay table(s)},
alongside the regular data, with the purpose to store the overlay index, and
\item a (possibly only conceptual) middle layer, which we call \emph{overlay logic},
that is in charge of keeping overlay tables up-to-date with the data and to fulfill
the specific queries the overlay index was introduced for.
\end{inparaenum}

We observe that the overlay logic can be naturally designed as a real middle
layer, which should 
\begin{inparaenum}[(1)]
\item take an \emph{original} query performed by the client, for example
expressed in plain or augmented SQL language, %
\item produce appropriate \emph{actual} queries
that act upon regular tables and/or overlay tables and submit them to the DBMS, %
\item get their results from the DBMS, and %
\item compose and interpret the results to reply to the original query of the client.
\end{inparaenum}
This is the approach taken by VerdictDB~\cite{park2018verdictdb}
for approximate processing of aggregated range queries. Since we aim at proving
the soundness and the practicality of the overlay index approach, the
construction of such a middle layer for overlay indexes is out of the scope
of this paper.

A simpler approach is to have a library that is only in charge to change or
query overlay tables. In this case, it is responsibility of the application to call this library
to change an overlay table every time the corresponding regular data table is
changed. Special care should be taken to keep consistency between the two.
For example, the whole update (of data and overlay tables) should be performed
within a transaction. 
In certain situations, it may be convenient to keep all the
data within the index itself without having a distinct regular data table. This
is the approach that we adopted for the experiments described in
Section~\ref{sec:exepriments}.

When designing a data structure for an overlay index, the time complexity of its operations
is clearly a major concern. However, we should also take into account its 
efficiency in terms of data transferred for each operation, and 
also how this transfer is performed. In fact, to perform a query on an overlay
index, it is likely that several actual queries to the DBMS should be performed.
It is important that these queries could be done in parallel. A 
\emph{round} is a set of actual queries that can be performed in parallel to the
DBMS. We assume actual queries to be submitted to the DBMS at the same instant.
The \emph{duration} of a round is the time elapsed from queries submission to the end of the longest one. Round duration is composed of
the execution time (on the DBMS) of the longest query and the communication latencies in both directions. A poor implementation may require the original query to take
several rounds to be accomplished, possibly consisting of only one actual query each. For example,
a plain porting of any balanced data structure (like, for example, AVL-trees,
skip lists, or B-trees) to use a DBMS as storage would take $O(\log n)$ rounds for their operations, where $n$ 
is the number of elements in the data structure.
This means that in the overall time to execute the original query, we should sum up the time spent
by $O(\log n)$ actual queries for both communication and execution on the DBMS.
As already mentioned, highly parallel systems (like
RAID arrays) may greatly reduce the overall response time if the \emph{tasks}
they have to perform (i.e., sectors to be read or written) are all known in
advance. In fact, they can usually execute many tasks concurrently making a much
better use of resources and obtaining a much smaller execution time, overall.
Even when a single hard drive is used, it is useful to know all the tasks in
advance since disk schedulers reorder all the tasks they know so that they are
fulfilled in a single sweep of the head of the hard
drive~\cite{thomasian2011survey} to reduce the time spent for seeking the
correct tracks. Some studies show how disk schedulers are relevant also when
solid state drives or virtualization is
adopted~\cite{boutcher2010does,kim2009disk,wang2013novel}.


To improve performances, the overlay logic may realize some form of caching by 
keeping part of the overlay table in the memory of the client. While this may
speed up some queries, it introduces a cache consistency problem, if more clients
are present. For data-changing queries, it is likely
that the overlay logic needs to know beforehand the part of the overlay table 
that is going to be modified,
before submitting the changes to the DBMS (this is the case for
DB-trees, described in Section~\ref{sec:db-tree}). The obvious approach is to
perform changes in (at least) two rounds. The first (\emph{read round}), 
to retrieve the part of the
overlay table that is involved in the change  and, the second (\emph{update round}),
to actually perform the change. The introduction of
caching may help in reducing the amount of data transferred from the database in
the read round. However, since the data needed for the change depends on the
request performed by the user, caching is unlikely to make the  read round
unnecessary, unless the application always updates the same set of data.
A particular case is the insertion of a large quantity of data (also called \emph{batch insert}
in technical DBMS literature). In this case, many insertions may be cumulated in cache and written
in one round, possibly using batch insert support from the DBMS itself.
We consider all aspects introduced by caching to be outside of the scope of this
paper. In the realization adopted for the experimentation of
Section~\ref{sec:exepriments}, we do not adopt any form of caching and we have
only one client.

It is worth mentioning that overlay logic may also be realized  exploiting
programmability facilities of certain DBMSes, usually called \emph{stored
	procedures}. 
In this case, the impact of communication between the overlay logic
and the DBMS would be negligible. This has two notable effects:
\begin{inparaenum}[(1)]
	\item the inefficiency of data-changing operations due to the need of two query rounds is mitigated and %
	\item the adoption of a realization that performs
	in many rounds (e.g., logarithmic in the data size) becomes more affordable.
\end{inparaenum}
Anyway, having many sequential query rounds still makes a poor use of RAID arrays, DBMS clusters, and
disk schedulers.  Hence, also in this setting, it is advisable to adopt special data
structures that limit the number of query rounds, like DB-trees. 
We note that stored procedures are proprietary features
of DBMSes, hence, exploiting them links overlay logic to a specific DBMS
technology. Further, not all DBMSes support them, especially in the NoSQL world.

%
%
%
%
%
%
%
%
%


\makeatletter
\newlength{\trianglerightwidth}
\settowidth{\trianglerightwidth}{$\triangleright$~}
\algnewcommand{\LineComment}[1]{\Statex \hskip\ALG@thistlm $\triangleright$ #1}
\algnewcommand{\LineCommentCont}[1]{\Statex \hskip\ALG@thistlm%
	\parbox[t]{\dimexpr\linewidth-\ALG@thistlm}{\hangindent=\trianglerightwidth
		\hangafter=1 \strut$\triangleright$ #1\strut}}
\makeatother

\section{The DB-Tree Data Structure}\label{sec:db-tree}

In this section, we describe the DB-tree data structure, which we propose to
realize an overlay
index. We first describe it intuitively. Then, we provide a formal description of
the data structure with its invariants. Finally, we describe its fundamental properties.

\subsection{Intuitive Description}\label{ssec:db-tree-intuitive}

A DB-tree stores key-value pairs ordered by key. To simplify the description we
assume that keys are unique and both keys and values have bounded length. 
A DB-tree is a randomized data structure that shares
certain features with \emph{skip
	lists}~\cite{pugh1998skip} and \emph{B-trees}~\cite{comer1979ubiquitous}, which
are widely used data structures to store key-value pairs, in their order. We
start by recalling skip
lists, which are conceptually simpler than DB-trees, and then we show how  
DB-trees derive from them. A formal description of DB-trees is provided in
Section~\ref{ssec:db-tree-formal}.

\begin{figure*}
	\centering
	\includegraphics[width=\textwidth]{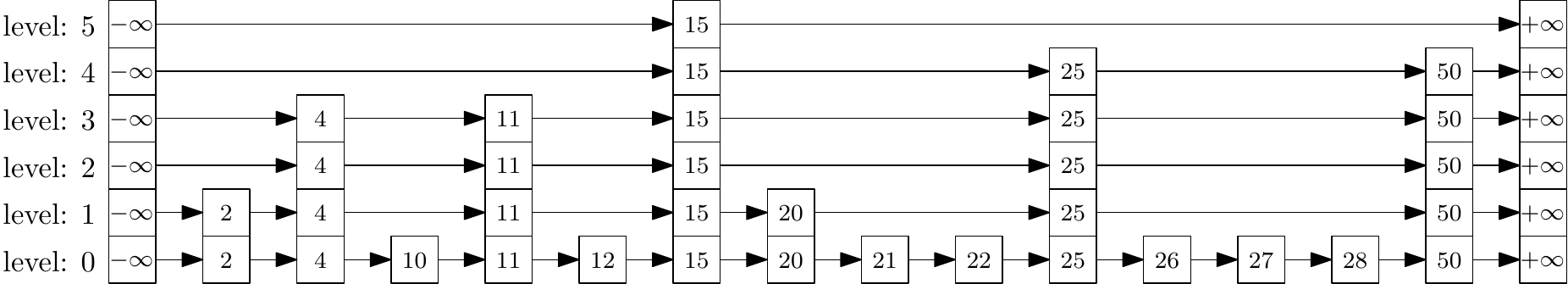}
	\caption{An example of a skip list. For each element, only the key is shown.
		Level 0 also stores values, not shown. } \label{fig:skiplist}
\end{figure*}


A skip list is made of a number of linked lists. An example is shown in
Figure~\ref{fig:skiplist},
where linked lists are drawn horizontally.
Each element of those linked lists
stores a key and each list is ordered. Each list is a \emph{level} that is 
denoted by a number. Level zero contains all the keys and it is the only level
that also stores values. Higher levels are progressively decimated, that is,
each level
above zero contains only a fraction of the keys of the level below. In this
way, each key $k$ is associated with a \emph{tower} of elements up to level
$l(k)$, called \emph{height} of the tower. For the example of
Figure~\ref{fig:skiplist}, we have $l(10)=0$ and $l(11)=3$. Beside regular pointers of
linked
lists, in skip lists, elements have also pointers that vertically link elements
of the same tower. 

\algnewcommand\algorithmicswitch{\textbf{switch}}
\algnewcommand\algorithmiccase{\textbf{case}}
\algnewcommand\algorithmicassert{\texttt{assert}}
\algnewcommand\Assert[1]{\State \algorithmicassert(#1)}%
\algdef{SE}[SWITCH]{Switch}{EndSwitch}[1]{\algorithmicswitch\ #1\
	\algorithmicdo}{\algorithmicend\ \algorithmicswitch}%
\algdef{SE}[CASE]{Case}{EndCase}[1]{\algorithmiccase\ #1}{\algorithmicend\
	\algorithmiccase}%
\algtext*{EndSwitch}%
\algtext*{EndCase}%

\begin{algorithm}[H]
	\caption{Extraction of a random level. }
	\label{algo:random-level}
	
	\begin{algorithmic}[1] 
		
		\Ensure A random level for a skip list or a DB-tree.
		
		\LineCommentCont{We denote by \emph{RandomChoice} a random value in
			$\{\mbox{GO-UP}, \mbox{STOP}\}$,
			where \mbox{GO-UP} is extracted with probability $p$ and \mbox{STOP} with
			probability $1-p$.}
		\State $l \gets 0$
		\While { RandomChoice is GO-UP }
		\State $l \gets l + 1$
		\EndWhile		
		\State \Return $l$
	\end{algorithmic}
	
\end{algorithm}

In traditional skip lists, when
$k$ is inserted, $l(k)$ is randomly selected using the approach described in
Algorithm~\ref{algo:random-level}.
We assume to have a random generator with two possible outcomes: GO-UP, with
probability $p$, and STOP, with probability $1-p$.
Initially, $l(k)$ is set to zero. We iteratively produce a random
outcome and increase $l(k)$ each time we obtain GO-UP. The procedure ends when
we obtain the first STOP.
In this way, each level contains a fraction $p$ of the
elements of the previous level. A typical value for $p$ is $1/2$. In a skip
list, the search of a key $k$ proceeds as follows. 
First, we linearly search the
highest level for the largest key less than or equal to $k$. If we have not
found $k$ yet, we traverse the tower link descending of one level and start to
search again until we either find $k$ (success) or reach level zero and a
key greater than $k$. This procedure takes $O(\log n)$, where $n$
is the number of keys in the skip list. Deletion and insertion can be also
performed in $O(\log n)$ (further details can be found in~\cite{pugh1998skip}).

Skip lists are regarded as efficient and easy-to-implement data structures,
since 
they do not need any complex re-balancing.
However, they are meant to be stored in memory, where traversing pointers is
fast. An attempt to represent a skip list in databases was made
in~\cite{di2007authenticated}, resulting in operations taking  $O(\log
n)$ query rounds, since the execution of each round depends on the result of the previous one.

\begin{figure*}
	\centering
	\includegraphics[width=\textwidth]{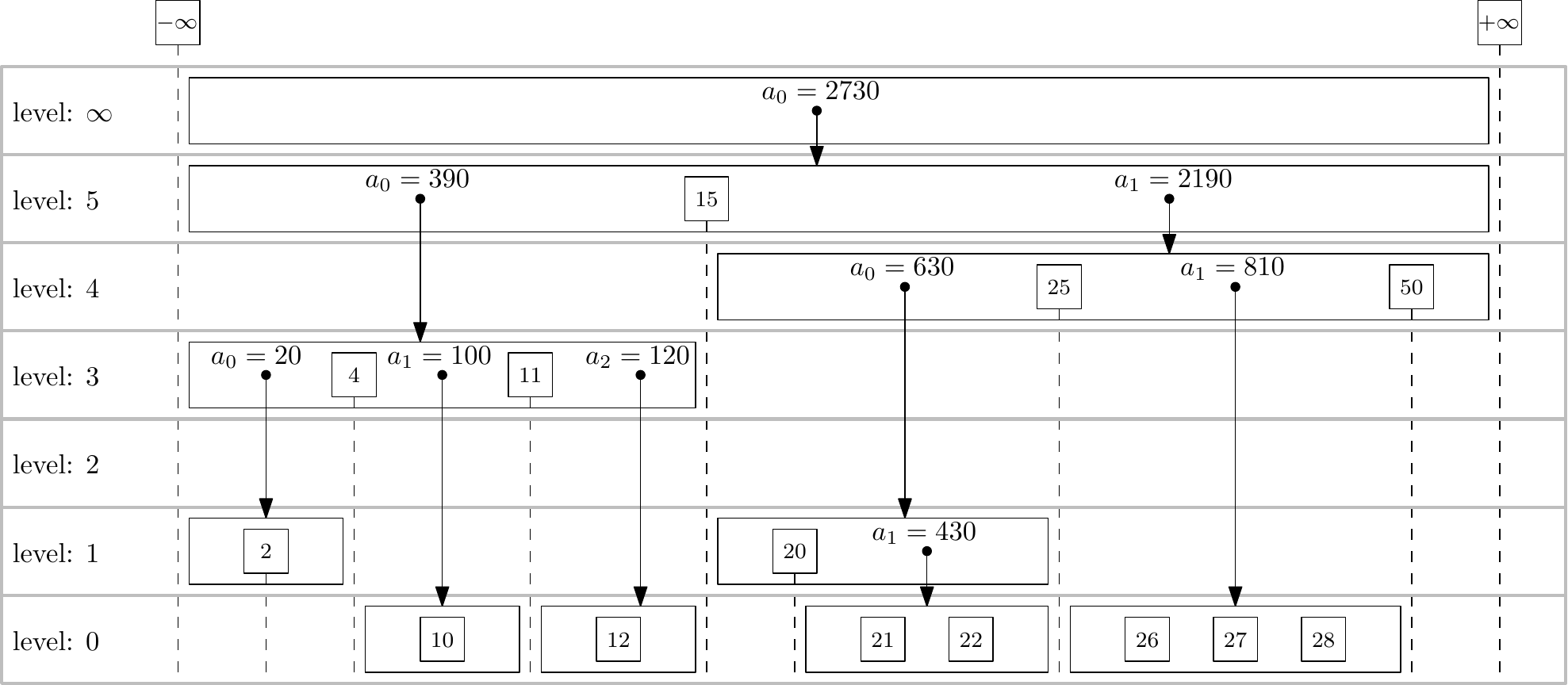}
	\caption{A DB-tree corresponding to the skip list shown in
		Figure~\ref{fig:skiplist} for the SUM aggregation function. For simplicity, we
		do not show values. In this and in the following examples, for each key $k$ the
		corresponding value is intended to be $v=10k$.} \label{fig:db-tree}
\end{figure*}
DB-trees can be interpreted, and intuitively explained, as a clever way of
grouping elements of a corresponding skip list so that update and query
operations can be performed in only one or two query rounds. The DB-tree
corresponding
to the skip list shown in Figure~\ref{fig:skiplist} is shown in
Figure~\ref{fig:db-tree}. To simplify the picture, in this and in the following
examples, we
show only keys and omit values, implicitly intending (only for the purpose of
the examples) that for each $k$ the corresponding value $v$ is derived by $k$ so
that $v=10k$.
We construct a DB-tree starting from a skip list as follows.
Levels of the DB-tree are in one-to-one correspondence with levels of the skip
list.
Only the elements at the top of each skip list tower are represented in the
DB-tree and 
grouped into
\emph{nodes}. Nodes are associated with levels. Each node spans consecutive keys
of its level, but it cannot span 
keys having, in between, others contained in some level above.
In other words, nodes at level $l$ only contains keys $k$ such
that $l(k)=l$.  Each key at level above $l$ separates nodes at level $l$ and
below.  
In Figure~\ref{fig:db-tree}, this separation is represented by vertical dashed
lines.
The number of keys associated with a node is not fixed.
Consider two adjacent, and not consecutive, keys in a node. The missing keys in that
level are represented by nodes at inferior levels, which are \emph{children}
or \emph{descendants} of that node. 

Since we intend to use DB-trees to support aggregate range
queries,  we store in each node the \emph{aggregate values} of the key-value
pairs
that are between two adjacent keys in the node, which are indeed explicitly
contained in
its descendants. We have one aggregate value for each child and the sequence
of keys of a node is interleaved with the aggregate values related to children. 
An aggregate value is omitted if the corresponding child is not present.
In the example of Figure~\ref{fig:db-tree}, aggregate values are shown in the
nodes interleaved with keys. We recall that, for the sole purposes of this
example, values $v$ are derived from the corresponding $k$ ($v=10k$).
\rnote{C3.2}

To realize a DB-tree, the overlay logic (see Section~\ref{sec:overlay-indexes})
stores its data in a single overlay table $T$. We use the relational DBMS jargon
just
for clarity, but we do not restrict the kind of the \emph{underlying DBMS} to
this class. The
value associated
with each key may be
stored in $T$ itself, if the application does not need to perform other queries
that are unrelated with the DB-tree. Otherwise, a regular table $D$ should be
present, and $T$
is treated as an index that should be kept up-to-date with $D$. 
In the following, we consider only table $T$, intending that the shown results
can be
applied also if a corresponding $D$ is present. In the rest of the paper, we use
the symbol $T$ to denote also an abstract DB-tree and we denote by $|T|$ the
number
of keys it contains.

\subsection{Formal Description}\label{ssec:db-tree-formal}
A DB-tree contains key-value pairs, where keys are non-null, distinct, and
inherently
ordered, and values are from a set $V$. We intend to support, efficiently,
arbitrary \emph{aggregate range queries}. An aggregate range query performs
aggregation on values
related to keys within a range that is chosen by the user and it is not known
in advance. We assume that the kinds of aggregation queries to support are based
on
a \emph{decomposable aggregation function}\footnote{We define a decomposable
	aggregation function in a similar way as in distributed systems
	literature, see for example~\cite{jesus2014survey}.}  
A decomposable aggregation function $\alpha:V^n \rightarrow C $ is
obtained by composing a triple of functions $\langle f, g, h \rangle$ defined as
follows.
\begin{itemize}
	\item $f: A \times A \rightarrow A$ is associative, that is  
	$f(a,b,c)=f(f(a,b),c)=f(a,f(b,c))$, and there exists an identity element denoted
	by $1_A$, that is for any $x \in A: f(1_A, x)= f(x,1_A)= x$. The associativity
	of $f(a,b)$ allows us to write $f(a_1,\dots,a_n)$, since the grouping according
	to which $f$ is applied is irrelevant for the final result.
	\item $g: V \rightarrow A$.
	\item $h: A \rightarrow C$.
\end{itemize}

We define $\alpha(v_1, \dots,v_n) = h(f(g(v_1),\dots,g(v_n)))$. We call $f$ the
\emph{core aggregation function}. \rnote{C2.1}In the following, by 
\emph{aggregation function}, we refer to either the single core aggregation
function $f$ or the whole decomposable aggregation function (i.e., the triple).
The
actual meaning should be clear from the context. 
Depending on $g$, $V$ may or may not comprehend the null value.
In the following, results of evaluation of $f(\cdot)$ are called \emph{aggregate
	values}.

This model can support many practical aggregation functions, like, count, sum,
average, variance, top-$n$, etc., which are referred to as \emph{distributive} and
\emph{algebraic} aggregation functions in literature~\cite{gray1997data}.
\rnote{C3.5 C2.1}For example, we can support top-2 (giving the first and second maximum) with the following definitions.
\begin{itemize}
	\item  $A = \{V,\bot\} \times \{V,\bot\}$, where
	we intend that the first element of each pair is the \emph{maximum}, the second
	is the \emph{second maximum}, $\bot$ means \emph{undefined}, and $\bot$ is less
	than
	any element in $V$,
	\item $g(v)= (v, \bot)$,
	\item $h(x)=x$,
	\item $f( (v_1,v_2) , (v_3,v_4) )= (m_1, m_2)$ where $m_1=\max\{v_1,v_3\}$ and
	$m_2$ is the second maximum in
	$U=\{v_1,v_2,v_3,v_4\}$, that is $m_2=\max( U-\{m_1\})$, and
	\item the identity element is $(\bot,\bot)$.
\end{itemize}
The above definitions can be easily generalized to support top-$n$.
\rnote{C3.10}Our model does not directly support so-called \emph{holistic}
aggregation
functions. In this kind of aggregation functions,  there is no constant bound on
the size of the storage needed to describe a sub-aggregate~\cite{gray1997data}. 
For this reason, they are
commonly recognized as hard to optimize (see for
example~\cite{wesley2016incremental,li2005no,chiou2001optimization}).
Examples of this kind of aggregation functions are median, n-th percentile, and
mode. 

\rnote{C3.5}We are now ready to formally describe the DB-tree data structure.
DB-trees support any aggregation function that fits the definition of
decomposable aggregation function stated above. \rnote{C3.6} We assume that the
aggregation function to be supported is known before the creation of the
DB-tree,
or at least $g$ and $f$ are known.

A DB-tree, keeps certain aggregate values ready to be used 
to compute aggregate range queries on any range, quickly. Its distinguishing
feature with respect
to other results known in literature is that it is 
intended to be efficiently stored and managed in
databases. While typical research works in this area show data structures whose
elements
are accessed by memory pointers (for in-memory data structures) or by block
addressing (for disk/filesystem based data structures), the primary mean to
access the elements of a DB-tree is
by range selection queries provided by the DBMS itself.

We do not restrict the kind of \emph{underlying database} that can be used to
store a DB-tree. However, for simplicity, we describe our model using
terminology taken from relational databases. We only require the underlying 
database to have %
\begin{inparaenum}[(i)]
	\item the ability to perform range selection on several columns, efficiently, 
	which is usually the case when proper indexes are defined
	on the relevant columns, and
	\item the ability to get, efficiently, the tuple containing the maximum
	(or minimum) value for a certain column among the selected tuples.
\end{inparaenum}

\newcommand{\standsfor}[1]{\mathrm{standsfor}(#1)}
\newcommand{\range}[1]{\mathrm{range}(#1)}

\rnote{C2.1} A DB-tree $T$ \emph{contains} a sequence of key-value pairs,
ordered according to their keys. A DB-tree is logically made of \emph{nodes}
forming a rooted tree. Each node $n$ of a DB-tree \emph{stands for} a
subsequence of contiguous key-value pairs of $T$, denoted by $\standsfor{n}$. 

\rnote{C2.1} A node $n$ is associated with an open interval $(n.\mathrm{min},
n.\mathrm{max})$ called \emph{range}, denoted $\range{n}$, where
$n.\mathrm{min}$ and $n.\mathrm{max}$ are either keys in $T$ or can assume
values $-\infty$ or $+\infty$. In any case, it should hold that
$n.\mathrm{min}<n.\mathrm{max}$. A node $n$ stands for the subsequence of
key-value pairs in $ T $ whose keys are strictly contained in $\range{n}$.
In other words, $n.\mathrm{min}$ is the key in $T$ right before $\standsfor{n}$,
or $-\infty$ if it does not exist, and  $n.\mathrm{max}$ is the key in $T$ right
after $\standsfor{n}$, or $+\infty$ if it does not exist.

A node $n$ explicitly \emph{contains} only some key-value pairs among those of
$\standsfor{n}$. The key-value pairs that are explicitly contained in $n$, do not need to be
necessarily contiguous in $ \standsfor{n} $. Those that are not explicitly contained in $ n $ are contained in 
nodes that are descendants of $n$.
The root of $T$ stands for the whole sequence contained in $T$ and has range
$(-\infty,+\infty)$. A node $n$ is associated with its \emph{aggregate
	sequence}, which is derived from $\standsfor{n}$ by substituting the
key-value pairs that are not explicitly contained in $n$ with the corresponding aggregate values. More
formally, the aggregate sequence for
a node $n$ is $a_0, p_1, a_1, \dots, a_{m-1}, p_m, a_m$, denoted
$n.\mathrm{aseq}$, where  $p_i$'s are key-value pairs $\langle k_i, v_i\rangle$,
$m$ is the number of keys in $n.\mathrm{aseq}$, and $a_i$'s are aggregate values
(note that subscripts indicate positions of key-value pairs within
$n.\mathrm{aseq}$ and not within the whole sequence in $T$). 
We say that $n$ \emph{contains} a key $k$ when $k$ is in $n.\mathrm{aseq}$.
\rnote{C2.1}  Some of the aggregate values in $n.\mathrm{aseq}$ may be \emph{missing}, as it
is explained in the following.
It is worth noting that, if $k_1,k_2,\dots, k_m$ are contained in $n$, it should
hold that
$n.\mathrm{min}<k_1<k_2<\dots<k_m<n.\mathrm{max}$. 
The number $m$ of key-value pairs contained in a node is not the same for all
nodes and may vary when the DB-tree is updated. The value $m$ related to a node
$n$ is denoted $n.m$.

For each node $n$, the children of $n$ are in one-to-one correspondence with
aggregate values in $n.\mathrm{aseq}$.
We denote $n_i$ the child of $n$ associated with aggregate value $a_i$ in
$n.\mathrm{aseq}$.
Keys in $n.\mathrm{aseq}$ impose limits on the keys contained in the children.
\rnote{C2.1}Namely, $n_i.\mathrm{min}=k_i$ and $n_i.\mathrm{max}=k_{i+1}$, but for $n_0$,
for which $n_0.\mathrm{min}=n.\mathrm{min}$, and $n_m$, for which
$n_m.\mathrm{max}=n.\mathrm{max}$, respectively.
If $k_i$ and $k_{i+1}$ are consecutive in $\standsfor{n}$, 
the $i$-th aggregate value in $n.\mathrm{aseq}$ is \emph{missing}, and the
corresponding child is also \emph{missing}.  If an aggregate value and its
corresponding child are not missing we say that they are \emph{present}.
The (present) aggregate value $a_i$ is the value of the core aggregation
function on
$\standsfor{n_i}$. In practice, exploiting the associative property of
$f(\cdot)$, we compute $a_i$ on the basis of $n_i.\mathrm{aseq}$.
Let $n.\mathrm{aseq}=a_0, \langle k_1,v_1 \rangle, a_1, \dots, a_{m-1}, \langle
k_m,v_m \rangle , a_m$, we define $f(n)=f(a_0, g(v_1), a_1, \dots, a_{m-1},
g(v_m), a_m)$, where we intend that any missing aggregate value should be
omitted from the list of the arguments of $f(\cdot)$.\ Hence, we can write $a_i
= f(n_i)$.

Each node $n$ has a level denoted $n.\mathrm{level}\geq 0$, which is $\infty$
for the root. The children of $n$ have levels that are strictly lower than
$n.\mathrm{level}$. We say that a node $n'$ is \emph{above} (\emph{below}) $n$
if level of $n'$ is greater (lower) than the level of $n$.

When a key $k$ is inserted into a DB-tree, it is associated with a randomly
selected level obtained using Algorithm~\ref{algo:random-level}. 
The level of $k$ is denoted $l(k)$.
Key $k$ is then inserted in a node at that level.

\subsection{Summary of Invariants}\label{ssec:invariants}

We now formally summarize the invariants that must hold for each node $n$ in a
DB-tree $T$.
In the following $m=n.m$, $l=n.\mathrm{level}$, 
$n.\mathrm{aseq} = a_0, \langle k_1, v_1\rangle, a_1, \dots, a_{m-1}, \langle
k_m, v_m\rangle, a_m$ where some $a_i$ are possibly missing, as explained above.
Children of $n$ are denoted $n_0,\dots,n_m$ and some of them may be possibly
missing. 
\begin{enumerate}
	
	\item \label{inv:root} If $n$ is the root of $T$, $n.\mathrm{min}=-\infty$,
	$n.\mathrm{max}=+\infty$, 
	$n.\mathrm{level}=+\infty$. If $T$ is empty, $n.m=\bot$, $n.\mathrm{aseq}$ is
	an empty sequence and $n$ has no children. If $T$ is not empty, $n.m=0$,
	$n.\mathrm{aseq}=a_0$ and $n$ has one child.
	
	\item \label{inv:keys-order-range} If $ n $ is not root, $n.m>0$ and
	\\$n.\mathrm{min}<k_1<k_2<\dots<k_m<n.\mathrm{max}$.
	
	\item \label{inv:child-range} If $n_i$ (with $0\leq i \leq m$) is present,
	$n_i.\mbox{level} < n.\mbox{level}$,\\
	$n_i.\mathrm{min}= 	\begin{cases}
	n.\mathrm{min}& \text{if } i=0\\
	k_i           & \text{otherwise}
	\end{cases}$ \\
	and\\
	$n_i.\mathrm{max}= \begin{cases}
	n.\mathrm{max}& \text{if } i=m\\
	k_{i+1}           & \text{otherwise}
	\end{cases}$
	
	\item \label{inv:recursive-aggregation}for all $i=0,\dots,m$, $a_i$ is present
	iff $n_i$ is present and $a_i= f(n_i)$
\end{enumerate}

\subsection{Fundamental Properties}\label{ssec:DBtrees:fundamental-properties}
In this section, we introduce some fundamental properties that are important for
proving the efficiency of the algorithms described in
Section~\ref{sec:algorithms}.

The following property is a direct consequence of
Invariants~\ref{inv:keys-order-range} and~\ref{inv:child-range}.
\begin{property}[Range Monotonicity]\label{prop:range_ancestors}
	Given a node $n$ of a DB-tree and $n'$ parent of $n$,
	$\range{n}\subseteq\range{n'}$.
\end{property}

Now we analyze the relationship among nodes and between nodes and keys.
Let $n$ be a node that contains $k$. Key $k$ cannot be contained in a
node that is above $n$, since for all nodes above $n$ (which have $k$ in their
range) $k$ is represented by an aggregate value (see
Invariants~\ref{inv:recursive-aggregation}). Key $k$
cannot be contained in a node that is below $n$, since it does not exist any of
those nodes whose range contains $k$ (see Invariants~\ref{inv:child-range} plus
the definition of $\range{\cdot}$ as an open interval).
From the above considerations, the following properties hold.

\begin{property}[Unique Containment]\label{prop:uniqueness}
	A key $k$ contained in a DB-tree $T$ is contained in one and only one node of
	$T$.
\end{property}

\begin{property}[Lowest Level]\label{prop:lowest-level}
	The node that contains a key $k$ is the one with minimum level among those that
	have $k$ in their range.
\end{property}

\rnote{C3.3} Concerning space occupancy, we note that in a DB-tree, a node
contains one or more key-value pairs and each of them is contained in at most
one node. This means that nodes are at most as many as the key-value pairs
stored in the DB-tree. Hence, the space occupancy for a DB-tree is $O(|T|)$, in
the worst case. In Section~\ref{sec:exepriments}, we also provide some details
about space occupancy in practice.

To reason about the size of the data that are transferred between the overlay
logic and the DBMS, it is  useful to
state the following property.

\begin{property}[Expected Node Size]\label{prop:constant-size-node}
	The expected number of keys that are contained in each node of a DB-tree is the
	same for all nodes.
\end{property}

This property can be derived from a consideration on the homogeneity of the
levels. The number of
keys contained in a node at level $i$ depends from the probability $p$ with
which
a key has level greater than $i$. In fact, these are the keys that partition
keys of level $i$ into distinct nodes. 
Since $p$ does not depends on $i$, this proves the property.


The following lemma states a fundamental result on which the efficiency of
querying a DB-tree is based.

\begin{lemma}[Expected maximum level]\label{lem:expected-maximum-level}
	Given a set of $r$ keys ${k_1, \dots, k_r}$ contained in a DB-tree $T$, the 
	expected value of $\max\{l(k_1), \dots, l(k_r)\}$, where $l(k_i)$ is the level
	of $k_i$, is $O(\log r)$.
\end{lemma}
In other words, considering $r$ keys, the expected value of the maximum of their
levels is logarithmic in $r$. Note that, this result also holds for skip lists,
but we were not able to find its proof in standard skip list literature. For
example,
in~\cite{pugh1998skip}, the proposed method for dimensioning the number of
levels of a skip list focuses on the level whose expected number of elements is
$1/p$, which is $O(\log n)$ with $n$ the number of keys in the skip list.
While this approach is viable for the maximum level, it cannot be applied 
when we should measure the number of levels spanned by a subset of the keys that
are contained in a much larger skip list or DB-tree.

The probability that $r$ keys are all at levels less than or equal to $i$ is
$(1-p^{i+1})^r$, since all random levels are independent. The probability that
the maximum is at level $i$ is  $(1-p^{i+1})^r-(1-p^{i})^r$, since this is the
probability of having all keys below $i+1$ minus the probability that all of
them are below $i$. Hence, the expected maximum level for $r$ keys can be
expressed as $\sum_{i=0}^\infty{i\left((1-p^{i+1})^r-(1-p^{i})^r\right)}$, which
can be rewritten as $\sum_{k=1}^{r}(-1)^{k-1}\binom{r}{k}\frac{p^k}{1-p^k}$ by
expanding binomials powers, reordering, and solving the infinite sum.
Sums similar to this 
are known to be quite subtle to deal with.
Flajolet and Sedgewick have provided several examples of 
how to tackle this kinds of sums in~\cite{flajolet1995mellin} using 
a mathematical tool called ``Nørlund--Rice integral''.
The proof that the above sum is $O(\log r)$ can be found
in~\cite{stackexchange3269979}.

\begin{figure*}
	\centering
	\begin{subfigure}[b]{.48\textwidth}
		\includegraphics[trim={2cm 0 6cm
			0},clip,width=.95\linewidth]{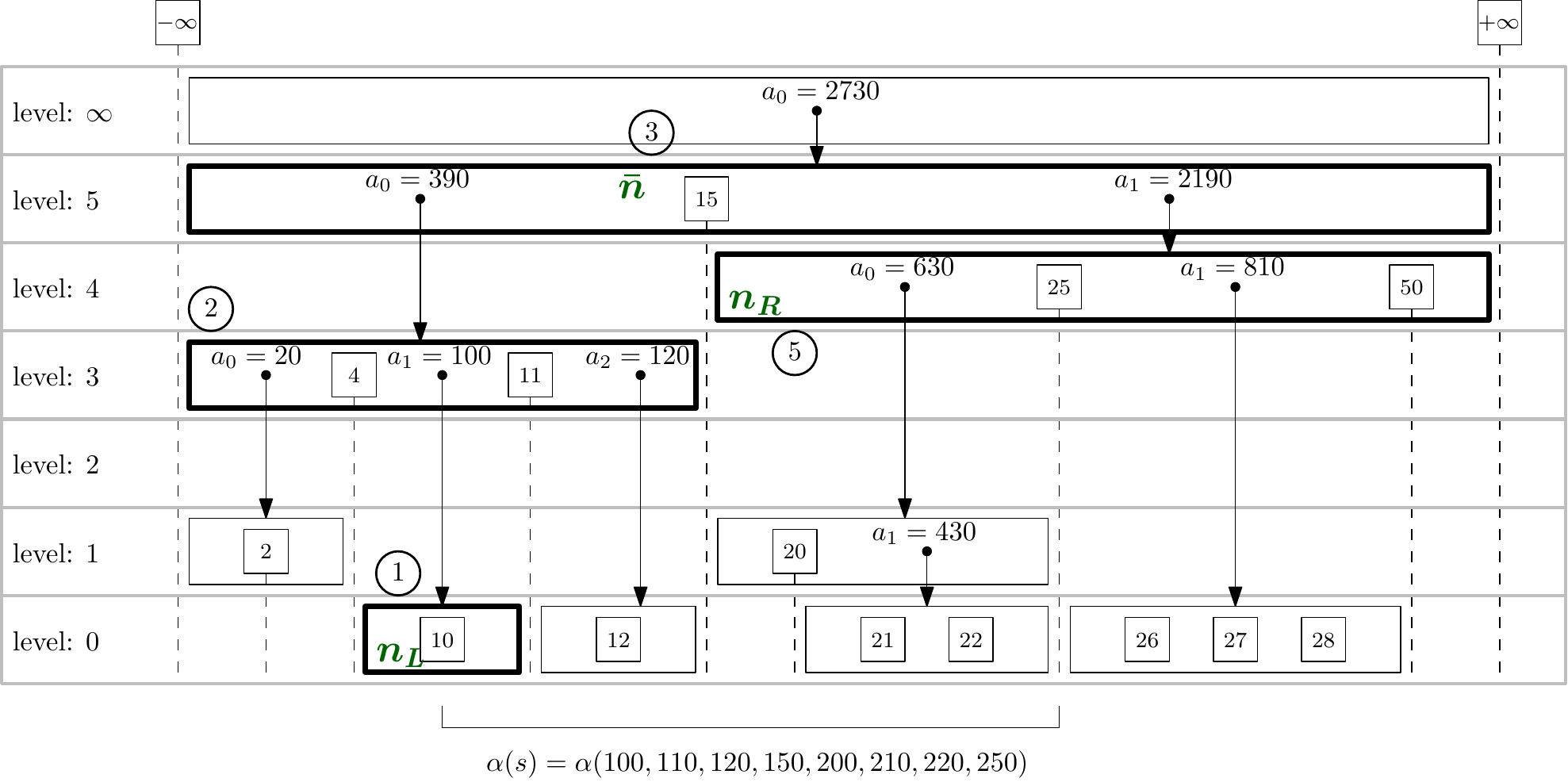}
		\caption{The nodes contributing to the construction of the sequence $s$,
			i.e., in the set $N$, are highlighted with thick contours.\\}
	\end{subfigure}%
	\hfill
	\begin{subfigure}[b]{.48\textwidth}
		
		\begin{enumerate}
			\item $n_L$: $s=[100]$
			\item ancestor of $n_L$: $s = s \| [110, a_2]= [100, 110, 120]$
			\item $\bar n$: $s = s \| [150]= [100, 110, 120, 150]$
			\item no nodes are between of $\bar n$ and $n_R$ 
			\item $n_R$: $s = s \| [a_0, 250]= [100, 110, 120, 150, 630, 250]$
		\end{enumerate}
		
		
		\vskip 1cm
		\caption{The steps for the construction of the sequence $ s $ for the example.
			Sequences of values are shown within square brackets. Concatenation of sequences
			is denoted by $||$.}
		
	\end{subfigure}
	\caption{An example of construction of an aggregate range query according to 
		the proof of Lemma~\ref{lem:aggreagate-range-query}. 
		In this example, the query regards the sum of the values of the keys in the
		range from 10 to 25.  The aggregation function is SUM and for each $k$ its value
		is $v=10k$.}
	\label{fig:prooflemma}
\end{figure*}

The following lemma is relevant to understand the correctness of the query
procedure that is algorithmically introduced in Section~\ref{sec:algorithms}.

\begin{lemma}[Aggregate Range Query
	Construction]\label{lem:aggreagate-range-query}
	Given a sequence of keys $k_1<k_2<\dots<k_r$ contained in a DB-tree $T$ with
	their associated values $v_1,\dots,v_r$, 
	the aggregation function $\alpha(v_1,v_2,\dots,v_r)$ can be computed 
	from selected parts of the aggregate sequences of a set of nodes $N$ of $T$
	containing
	\begin{itemize}
		\item the node that contains $k_1$, denoted $n_L$,
		\item the node that contains $k_r$, denoted $n_R$, 
		\item the common ancestor of $n_L$ and $n_R$ with minimum level, denoted $\bar
		n$,
		\item the ancestors of $n_L$ and $n_R$ up to $\bar n$.
	\end{itemize} 
\end{lemma}

\begin{proof}
	We build a sequence $s$, that contains some of the values in
	$\{v_1,v_2,\dots,v_r\}$ and the missed values are substituted by aggregate
	values form nodes of $N$. 
	We show that $\alpha(s)=\alpha(v_1,v_2,\dots,v_r)$.
	We build $s$ from the aggregate sequences of nodes in $N$.
	Sequence $s$ is built from left to
	right. An example of the construction described in this proof is
	shown in Figure~\ref{fig:prooflemma}.
	
	We start with $s$ containing, from $n_L.\mathrm{aseq}$, value $v_1$ associated
	with $k_1$
	and all values (aggregated or not) at its right, in their order. 
	We proceed by considering ancestors of $n_L$ in $N$ excluding $\bar n$, ordered
	by increasing
	level. For each of them, denoted by $n$, let $n'$ be its descendant in $N$. We
	add to $s$, from $n.\mathrm{aseq}$, the value associated with $n'.\mathrm{max}$,
	and
	all values (aggregated or not) at its right in their order. We do that, if
	$n'.\mathrm{max}$ exists in $n.\mathrm{aseq}$, otherwise we skip to the next
	node.
	Let $n'$ and $n''$ be the two descendants of $\bar n$, with $n'.\mathrm{max}\leq
	n''.\mathrm{min}$.
	We add to $s$, from $\bar n.\mathrm{aseq}$, the value associated with
	$n'.\mathrm{max}$
	and all values (aggregated or not) at its right, in their order, up to the value
	associated with $n''.\mathrm{min}$ ($n'.\mathrm{max}$ and $n''.\mathrm{min}$
	exists in $\bar n.\mathrm{aseq}$ by construction of $\bar n$, $n'$ and $n''$,
	and by Invariant~\ref{inv:child-range}).
	We continue by considering each 
	ancestor $n$ of $n_R$ in $N$, excluding $\bar n$, ordered by decreasing
	level. Let $n'$ be the descendant of $n$ in $N$. We add to $s$, from
	$n.\mathrm{aseq}$, all 
	values (aggregated or not) starting from left up to
	the value associated with $n'.\mathrm{min}$, in their order. We do that, if it
	exists in $n.\mathrm{aseq}$, otherwise we skip to the next node.
	Finally, we add to $s$, from $n_R.\mathrm{aseq}$, all values (aggregated or not)
	starting from left up to
	value $v_r$ associated with $k_r$, in their order.
	To prove that
	$\alpha(s)=h(f(s))=h(f(g(v_1),g(v_2),\dots,g(v_r))=\alpha(v_1,v_2,\dots,v_r)$,
	we 
	inductively apply Invariant~\ref{inv:recursive-aggregation}, to $f(s)$.
	This allows us to express 
	each aggregate value in $s$ as application of $f(\cdot)$ to the aggregate
	sequence of the corresponding child. The associative property of $f(\cdot)$
	ensures that the aggregate value remains the same. 
\end{proof}

Lemma~\ref{lem:aggreagate-range-query} states that it is possible to answer to
an
aggregate range query considering only certain parts of the aggregate
sequences of a small set of nodes $N$. We now estimates the size of $N$. The
following result relates this size to the results stated by
Lemma~\ref{lem:expected-maximum-level}.

\begin{lemma}\label{lem:lowest-common-ancestor-minimum-level}
	Given a sequence of keys $k_1<k_2<\dots<k_r$ contained in a DB-tree $T$, 
	let $n_L$ and $n_R$ be the nodes containing $k_1$ and $k_r$, respectively,
	and call $\bar n$ their common ancestor with minimum level, 
	the level of $\bar n$ is given by the maximum of the levels 
	$l(k_1),l(k_2),\dots,l(k_r)$.
\end{lemma}

The above lemma is easily proven considering the construction used for proving
Lemma~\ref{lem:aggreagate-range-query}. Each of the keys $k_1,k_2,\dots,k_r$ is
either contained in one of the nodes in $N$ or in one of its descendants. Note
that, at least one of the keys is contained in $\bar n$, which is above of all
other nodes in $N$.
This means the $\bar
n.\mathrm{level}= \max \{ l(k_1),l(k_2),\dots,l(k_r) \}$.

\begin{lemma}[Aggregate Range Query Size]\label{lem:aggr-range-query-efficiency}
	Given a DB-tree $T$ and $k'<k''$ two keys contained in $T$, 
	such that $[k',k'']$ contains $r$ keys in $T$, the aggregate 
	range query on $[k',k'']$ can be answered considering
	expected $O(\log r)$ values spread on expected $O(\log r)$ nodes.
\end{lemma}
\begin{proof}
	The proof of Lemma~\ref{lem:aggreagate-range-query} constructs an aggregated
	sequence $s$ out of the aggregates sequences of nodes in $N$ which is made of
	two descending paths with a common ancestor $\bar n$.
	Lemma~\ref{lem:lowest-common-ancestor-minimum-level} states that the level of
	$\bar n$ is the maximum of the levels assigned to keys in $[k_1,k_r]$, with
	$k'\leq k_1$ and $k_r \leq k''$.
	Lemma~\ref{lem:expected-maximum-level} states that this 
	maximum is expected to be $O(\log r)$. \rnote{C4.1 C2.1}Since each of the two ascending paths
	constructed in the proof of Lemma~\ref{lem:aggreagate-range-query} contains at
	most one node for each level, the maximum number of nodes needed to 
	answer to an aggregate range query is expected $O(\log r)$.
	This proves the statement about the expected number of nodes 
	involved in the query.
	Since the size of each node has constant expected size
	(Property~\ref{prop:constant-size-node}), 
	it can provide only an expected constant number of values. Hence, also the
	number of
	values are expected $O(\log r)$.
\end{proof}

\section{DB-Tree Algorithms}\label{sec:algorithms}

\begin{algorithm}
	\caption{Algorithm for performing an aggregate range query on a DB-tree. }
	\label{algo:aggregate-range-query}
	
	\begin{algorithmic}[1] 
		 
		\Require Two keys $k'$ and $k''$ (with $k' < k''$ ) and a DB-tree $T$.
		\Ensure $\alpha(v_1,\dots, v_r)$ where $\langle k_1,v_1\rangle, \dots, \langle k_r,v_r\rangle$ and 
		$k_1,\dots,k_r$ are all the keys contained in $T$ within $k'$ and $k''$, and
		$k'\leq k_1 < \dots< k_r \leq k''$.

		\LineComment Execute Lines~\ref{line:query:L}-\ref{line:query:-bar-n} in one query round.
		\State \label{line:query:L}$L \gets $ a sequence of nodes resulting from the selection from $T$ of all nodes $n$ such that 
			$n.\mathrm{min} < k' < n.\mathrm{max} \leq k'' $ ordered by ascending level.

		\State \label{line:query:R}$R \gets $ a sequence of nodes resulting from the selection from $T$ of all nodes $n$ such that 
			$k' \leq n.\mathrm{min} < k'' < n.\mathrm{max} $ ordered by ascending level.
			
		\State \label{line:query:-bar-n}$\bar n \gets $ the node $n$ in $T$ such that $n.\mathrm{min} < k' < k'' < n.\mathrm{max} $,
		 with minimum level.

%
%
		
		\State \label{line:query:computation-begin}$a_L \gets 1_A$, where $1_A$ is the identity element of $A$ 
		\ForAll { nodes $n$ in $L$ in ascending order }
			\State \label{line:query:L-groupby-changed} $s \gets 
			\begin{cases}				
				f(v_i, a_i, \dots, v_m, a_m), \mbox{ where } k' \leq k_i, & \mbox{if there exists at least one } \langle k_i,v_i \rangle \mbox{ with } k' \leq k_i\\
				 1_A,  & \mbox{otherwise}
			\end{cases} $
			\State $a_L \gets f(a_L,s)$
		\EndFor
		
		\State $a_R \gets 1_A$ 
		\ForAll { nodes $n$ in $R$ in ascending order }
			\State \label{line:query:R-groupby-changed} $s \gets 
			\begin{cases}				
				f(a_0, v_1, \dots a_{i-1},v_i),  \mbox{ where } k_i\leq k'', & 
				\mbox{if there exists at least one } \langle k_i,v_i \rangle \mbox{ with }  k_i\leq k''\\
				1_A,  & \mbox{otherwise}
			\end{cases} $ 
			\State $a_R \gets f(s,a_R)$
		\EndFor
		\State Let $\bar n.\mathrm{aseq}$ be $a_0,\langle k_1,v_1 \rangle, a_1, \dots, \langle k_m,v_m \rangle, a_m$ 
		\State \label{line:query:n-groupby-changed}$s \gets f( v_i, a_i, \dots a_{j-1},v_j)$ where  $k_i$ is the lowest key such that $k'\leq k_i $ and $k_j$ is the highest key such that $k_j\leq k''$.
		\State \label{line:query:computation-end}\Return $f(a_L, s, a_R)$
	
	\end{algorithmic}
	
\end{algorithm}

In this section, we describe the algorithms to perform queries on a DB-tree and to modify its content. We also state their correctness and efficiency on the basis of the 
fundamental properties provided in Section~\ref{ssec:DBtrees:fundamental-properties}.
At the end of the section, we provide a comparison with similar data structures.

In the following, to simplify the description, we always assume $g(v)=v$ and $h(v)=v$
and never apply $g(\cdot)$ and $h(\cdot)$ explicitly. The adaption to the case of non-trivial $g(\cdot)$ and $h(\cdot)$ is straightforward.

When writing the algorithms, we denote by $a_0, \langle k_1, v_1 \rangle,$ $a_1,
\dots,$ $\langle k_m, v_m \rangle, a_m$ the aggregate sequence of a node $n$
containing $m$ keys. We use that notation implicitly meaning that some $a_j$ may be missing
(see Section~\ref{ssec:db-tree-formal}). For example, using that notation we includes
sequences like $\langle k_1, v_1 \rangle, \langle k_2, v_2 \rangle, a_2$,  for
$m=2$, where $a_0$ and $a_1$ are missing. We also includes the special case in which 
$m=0$, where the sequence is only $a_0$. Analogously, we write
$f(v_i, a_i,\dots,a_{j-1} v_j)$, $f(a_0, v_1,\dots,a_{i-1} v_i)$, and $f(v_i,
a_i,\dots v_m,a_m)$ to aggregate a subsequence of an aggregate sequence of a node, again, implicitly
meaning that some aggregate values may be missing.

\subsection{Aggregate Range Query}

The procedure to compute an aggregate range query is shown in
Algorithm~\ref{algo:aggregate-range-query}. Two keys $k'$ and
$k''$, not necessarily contained in DB-tree $T$, are provided as input. The algorithm computes the
aggregate value for the values of all keys between  $k'$ and $k''$ in $T$. The
procedure closely follows the proof of Lemma~\ref{lem:aggreagate-range-query}.
Firstly, it retrieves the nodes cited in the statement of that lemma, plus some others that
we will prove, are irrelevant. This is
done in Lines~\ref{line:query:L}-\ref{line:query:-bar-n}. These queries can be
performed in parallel in a single query round. 
The size of the data is $O(\log r)$ by Lemma~\ref{lem:aggr-range-query-efficiency}, where $r \leq |T|$ is the 
number of keys between $k'$ and $k''$.
Hence, the following theorem about efficiency of Algorithm~\ref{algo:aggregate-range-query} holds.

\begin{theorem}[Query Efficiency]
	On a DB-tree $T$, an aggregate range query on a range containing $r$ keys in $T$ can be executed, by Algorithm~\ref{algo:aggregate-range-query}, in only one query round and 
	the expected size of the transferred data
	is $O(\log r)$.
\end{theorem}

\rnote{C3.3} Since Algorithm~\ref{algo:aggregate-range-query} processes each retrieved node a constant number of times, 
the execution of the part of the algorithm after data retrieval 
has expected time complexity $O(\log r)$.

We now focus on the correctness of Algorithm~\ref{algo:aggregate-range-query}.
Referring to the symbols
introduced by the statement of Lemma~\ref{lem:aggreagate-range-query}, we consider $k_1$ as the lowest key in $T$ such
that $k'\leq k_1$ and $k_r$ as the highest key in $T$ such that $k_r\leq k''$. %
%
By construction, there is no other key between $k'$ and $k_1$ and between $ k_r $ and $ k''$ in $T$. 

We show that the set of nodes $ N'=L\cup R \cup \{\bar n
\} $, selected by the algorithm, contains the set of nodes $ N $ identified by
the statement of Lemma~\ref{lem:aggreagate-range-query}. The correctness of
Algorithm~\ref{algo:aggregate-range-query} will be given by the fact that the result
of the aggregate value computed on both sets of nodes are equal.

\begin{figure*}
	\centering
	\includegraphics[width=0.9\linewidth]{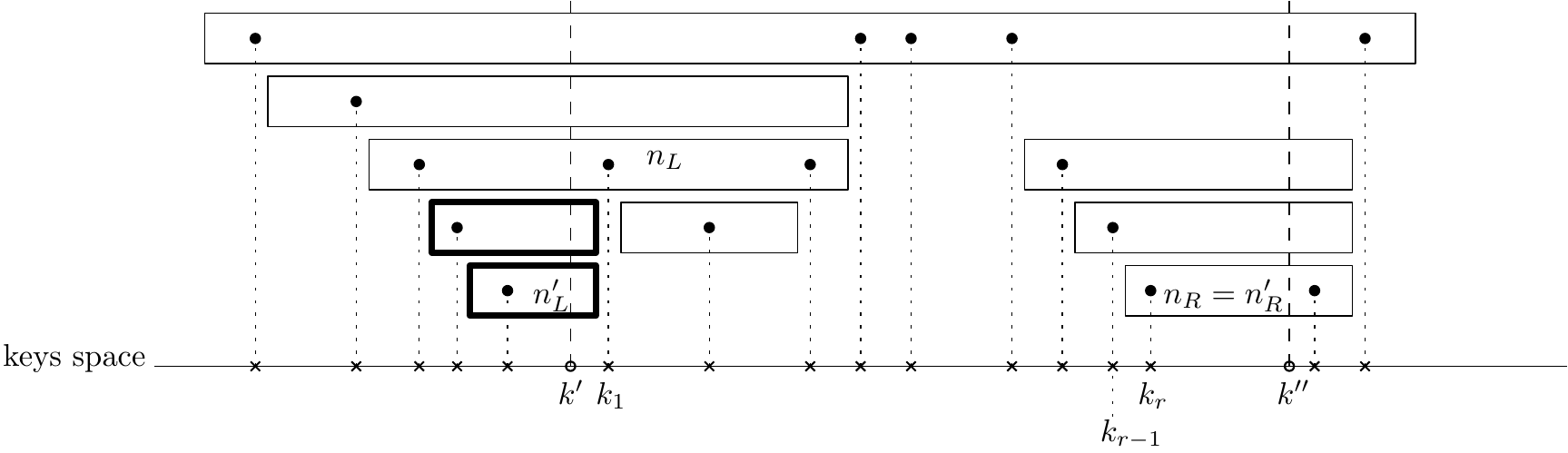}
	\caption{\infloatrnote{C2.1} An example in which Algorithm~\ref{algo:aggregate-range-query} selects more nodes than those considered in Lemma~\ref{lem:aggreagate-range-query}. The additionally selected nodes are evidenced with fat borders. See text for details.}
	\label{fig:algolemma}
\end{figure*}

We now analyze the differences between $N'$ (of the algorithm) and $N$ (of the lemma). \rnote{C2.1}Refer to Figure~\ref{fig:algolemma} for an example. 
The algorithm selects a $ n'_L $ such that $n'_L.\mathrm{min} < k' <
n'_L.\mathrm{max}$, if $ k_1 < n'_L.\mathrm{max}$ the node $ n'_L $ is equal to
$ n_L $ selected by the lemma, otherwise $ n_L $ is an ancestor of $ n'_L $. 
Analogous reasoning can be done for $ n'_R $ and $ n_R $. 
\rnote{C2.1}Figure~\ref{fig:algolemma} shows a case in which $n'_L.\mathrm{max}=k_1$ and hence
$ n_L $ is an ancestor of $ n'_L $. In the same figure, $ n'_R.\mathrm{min} = k_{r-1} < k_r$
and hence $n'_R= n_R$.

 This means that $L$ and $R$ 
of the algorithm may additionally include descendants of $ n_L $ and $n_R$ down to $ n'_L $ and $n'_R$, respectively.
The remaining part of $L$ ($R$) coincides with ancestors of $n_L$ ($n_R$), for both the 
algorithm and the lemma. We now show that $\bar n$ is the same for both the algorithm and the lemma.
From the relationship  $k' \leq k_1 < k_r  \leq k''$, 
$\bar n$ in the algorithm might be an ancestor of $ \bar n $ in the lemma, since
the range of the first may only be larger than the range of the second.
By Invariant~\ref{inv:child-range}, the range of a child is smaller only when 
there is a key in the parent that limits the range of the child, but this is not possible 
in our case, since no other key is present between $k' $ and $k_1$ and/or between $k_r$ and $ k'' $, by construction.
This means that $\bar n$ is the same for both the algorithm and the lemma.

We have just proven that $N'$ can differ from $N$ only by some descendants of $ n_L $ and $ n_R $. 
Now, we prove that the contribution of those descendants to the overall aggregate value is null.
Consider the relevant case $n'_L \neq n_L$, we have that $k_1$ is contained in $n_L$ by construction, $k'$ is contained in $\range{n'_L}$ 
by construction, no key is between $k'$ and $k_1$ by construction. This implies that $n'_L.\mbox{max}=k_1$ by Invariant~\ref{inv:child-range} and that there is no key in $n'_L.$aseq greater that $k'$ to considered in 
the computation of the aggregate value (see Line~\ref{line:query:L-groupby-changed}). The same reasoning holds even if 
there are several nodes between $n'_L$ and  $n_L$. An analogous reasoning can be done about the contribution of nodes
from $n'_R$ up to and excluding $n_R$.

The above considerations together with the Lemma~\ref{lem:aggreagate-range-query} prove the following theorem about the correctness of Algorithm~\ref{algo:aggregate-range-query}.

\begin{theorem}[Query Correctness]
	An aggregate range query executed by Algorithm~\ref{algo:aggregate-range-query}
	on a DB-tree $T$ with range $[k',k'']$ correctly returns
	$\alpha(v_1,\dots,v_r)$ where $v_1,\dots,v_r$ are the values corresponding to all keys $k_1,\dots,k_r$ in $T$ such that
	$k' \leq k_1 < \dots < k_r \leq k''$.
\end{theorem}

\subsection{Update, Insertion and Deletion}\label{ssec:DBTree:upd-ins-del}

\begin{algorithm}
	\caption{This algorithm propagates the change of values of a certain node into
	the aggregate values in all its ancestors. It also inserts aggregate values that should be present in the aggregate sequence since the corresponding child is present. This algorithm is intended to be
	performed on a pool of nodes stored in main memory. }
\label{algo:update-aggrvalue-up}
	
	\begin{algorithmic}[1] 
		
		\Require A sequence $U$ of nodes, in ascending order of level, representing an ascending path in a DB-tree. The first node of $U$ is denoted $\bar n$. 
		\Ensure Aggregate values of nodes in $U$ above $n$ are updated to reflect the current state of $\bar n$.
		
		\State $a \gets f(\bar n)$
		\State $k \gets $ any key contained in $\bar n$.
		\State Remove $ \bar n $ from $ U $		
		\For { each node $n$ in $U$ in ascending order of level}
		\State Let $a_0, p_1, a_1, \dots,p_m, a_m$ be $n.\mathrm{aseq}$ with $m=n.m$ and $p_i=\langle k_i, v_i \rangle$ \Comment Some $a_i$'s might be \emph{missing}.
		\Switch{}  \Comment{This switch can assigns a previously \emph{missing} $a_i$ making it \emph{present}.}
		\Case{ $m=0$ }
		\State $a_0 \gets a$ 
		\EndCase
		\Case{ $k < k_1$}
		\State $a_0 \gets a$ 
		\EndCase
		\Case{ $k_m < k $}
		\State $a_m \gets a$ 
		\EndCase
		\Case{ $k_i<k<k_{i+1}$ for any $i\in\{0,\dots,m\}$}
		\State $a_i \gets a$ 
		\EndCase
		\EndSwitch
		\State $a \gets f(n)$
		\EndFor
		
	\end{algorithmic}
	
\end{algorithm}

\begin{algorithm}
	\caption{This algorithm updates a key $k$ in a DB-tree $T$ with a new value $v$ assuming that $k$ is already present in $T$.}
	\label{algo:update}
	
	\begin{algorithmic}[1] 
		
		\Require A key $k$, a DB-tree $T$ that contains $k$, and a new value $v$ to be assigned to $k$.
		\Ensure The value for $k$ in $T$ is $v$.
		
		\State \label{line:alg:update-retrieve}$N \gets $ the sequence of nodes resulting from the selection from $T$ of all nodes $n$ such that $n.\mathrm{min} < k < n.\mathrm{max}$ ordered by ascending $n.\mathrm{level}$. This is performed in a single query round.
		
		\State $n \gets $ the first node in $N$. 
		
		\State Update the value for $k$ in $n$ with $v$.
		
		\State Update affected aggregate values of nodes in $N$ by calling Algorithm~\ref{algo:update-aggrvalue-up} on $N$.
		
		\State Store all nodes in $N$ to $T$, in one round.
	\end{algorithmic}
	
\end{algorithm}

All the algorithms to change a DB-tree $T$ shown in this section can be divided in the following three phases.

\begin{enumerate}[label=P\arabic*.,ref=P\arabic*]

	\item \label{phase:retrieve} The nodes of $T$ that are relevant for the change 
	are \emph{retrieved} from the database, in one single round, and put into a 
	\emph{pool} of nodes stored in main memory (\emph{read round}). 
	
	\item \label{phase:change} The pool is \emph{changed} to reflect the change required for $T$. Old
    nodes may be updated or deleted and new ones may be added.

    \item \label{phase:store} The nodes of the pool are \emph{stored} into the database, in one round (\emph{update round}).
\end{enumerate}

In our description, we always assume the pool to be empty when the algorithm
starts. A rough form of caching could be obtained by not cleaning up the pool
between two operations, but we do not do consider that, to simplify algorithms
descriptions. It is worth mentioning, that the expected size of the pool is
always $O(\log |T|)$ during the execution of each of the algorithms and that at
beginning and end of Phase~\ref{phase:change}, the DB-tree invariants hold.
\rnote{C3.3} Further, as it will be clear from the following descriptions, all
algorithms process nodes of the pool at most a constant number of times and
hence the execution time of the processing part of the algorithms, i.e.,
excluding interaction with the DBMS, is $O(\log |T|)$ on average, for all of
them.

A particular operation, performed in Phase~\ref{phase:change}, recurs in all
algorithms. When a value of a key is changed or a key is inserted or deleted, the
corresponding aggregate values of nodes in the path to the root should be
coherently updated (see Invariant~\ref{inv:recursive-aggregation} in
Section~\ref{ssec:invariants}). This is always performed as the last step of 
Phase~\ref{phase:change}. This procedure is shown in
Algorithm~\ref{algo:update-aggrvalue-up}. It simply traverses all nodes in a
path to the root computing and updating aggregate values from bottom to top.
If previously missing children was added, it also restore the corresponding aggregate values.

\textbf{Update.} The procedure to update the value of a key already present in a DB-tree $T$ is
shown in Algorithm~\ref{algo:update}.
 The algorithm follows the three phases
listed above. In Line~\ref{line:alg:update-retrieve}, the node containing $k$ is
retrieved with all its ancestors. This follows from
Properties~\ref{prop:range_ancestors},~\ref{prop:uniqueness}
and~\ref{prop:lowest-level}. From Lemma~\ref{lem:expected-maximum-level}, the
expected number of nodes retrieved is $O(\log |T|)$. Invariants are clearly preserved, 
since the structure of $T$ is unchanged and Invariant~\ref{inv:recursive-aggregation}
is ensured by the execution of Algorithm~\ref{algo:update-aggrvalue-up}.

 \begin{algorithm}
 	\caption{This algorithm inserts a new key-value pair into a DB-tree $T$, supposing the key is not contained in $T$.}
 	\label{algo:insert}
 	
 	\begin{algorithmic}[1] 
 		
 		\Require A key-value pair $\langle k, v\rangle$ and a representation of a DB-tree $T$ that does not contain $k$.
 		
 		\State \label{line:alg:insert:getN}$N \gets $ be the sequence of nodes resulting from the selection from $T$ of all nodes $n$ such that $n.\mathrm{min} < k < n.\mathrm{max}$ ordered by ascending $n.\mathrm{level}$. This is performed in a single query round.
 		
 		\State \label{line:alg:insert:level}$l \gets$ a random level extracted according to Algorithm~\ref{algo:random-level}.

 		\State $\bar n \gets $ the node in $N$ such that $\bar n.\mathrm{level}=l$, or $\bot$ if it does not exist.

 		\State \label{line:alg:insert:getU}$U \gets $ the subsequence of nodes $n$ in $N$ with $n.\mathrm{level}>l$
 		
 		\State $n' \gets $ the first node in $U$
 		
 		\State \label{line:alg:insert:getD}$D \gets $ the subsequence of nodes $n$ in $N$ with $n.\mathrm{level}<l$

%

 		\If {$\bar n = \bot $ } \Comment Init with a dummy node, if needed.
 			\State $k_\mathrm{prev} \gets$ the highest key contained in $ n' $ such that $k_\mathrm{prev} < k $, otherwise $ k_\mathrm{prev} =  n'.\mathrm{min}$

			\State $k_\mathrm{next}  \gets$ the lowest key contained in $ n' $  such that $ k < k_\mathrm{next} $, otherwise $ k_\mathrm{next} =  n'.\mathrm{max}$

	 		\State $\bar n \gets$ a new node with 
	 		$\bar n.\mathrm{level}=l$, 
	 		$\bar n.m = 0$, 
	 		$\bar n.\mathrm{min} = k_\mathrm{prev}$, 
	 		$\bar n.\mathrm{max} = k_\mathrm{next}$, 
	 		$\bar n.\mathrm{aseq}$ is empty.
 		\EndIf

 		\State \label{line:algo:insert:insertnewpair}Insert   $\langle k, v \rangle$ in $\bar n.\mathrm{aseq}$ at the correct position to keep Invariant~\ref{inv:keys-order-range}, and coherently $\bar n.m \gets \bar n.m + 1$.
 		
%

	 		
	 		\State Let $L$ and $R$ be two empty sequences of nodes.
	 		
	 		\ForAll { nodes $n$ in $D$ in descending order of level} 			\label{line:algo:insert:split-for}

%
	 		
	 		\State Create nodes $n_L$ and $n_R$ such that 
	 		\begin{quote}
	 			\begin{quote}
	 				\begin{itemize}
	 					\renewcommand\labelitemi{$\circ$} 
	 					\item $n_L.\mathrm{level} \gets n.\mathrm{level}$,
	 					$n_L.\mathrm{min} \gets n.\mathrm{min}$, 
	 					$n_L.\mathrm{max} \gets k$, 
	 					$n_L.\mathrm{aseq} \gets$ the order-preserving subsequence of $ n.\mathrm{aseq}$ containing all keys less than $k$ with all aggregate values at their left (if present). 
	 					\item $n_R.\mathrm{level} \gets n.\mathrm{level}$,  
	 					$n_R.\mathrm{min} \gets k$, 
	 					$n_R.\mathrm{max} \gets n.\mathrm{max}$, 
	 					$n_R.\mathrm{aseq} \gets$ the order-preserving subsequence of $ n.\mathrm{aseq}$ containing all keys greater than $k$ with all aggregate values at their right (if present).
	 				\end{itemize}
	 			\end{quote}
	 		\end{quote}
 		
	 		\State Add $n_L$ at the beginning of $L$, only if $n_L$ contains at least one key.
	 		\State Add $n_R$ at the beginning of $R$, only if $n_R$ contains at least one key. \LineComment{Nodes in $L$ and $R$ turn out to be in ascending order of level.}
	 		\EndFor
	 		
	 		\State Update affected aggregate values on $L || \{\bar n\}$ and $R|| \{\bar n\}$ (if $L$ or $R$ are not empty) by calling Algorithm~\ref{algo:update-aggrvalue-up}. Note that, this also adds currently missing aggregate values that do have corresponding children. 
	 		
	 		\State Update affected aggregate values in $U$ by calling Algorithm~\ref{algo:update-aggrvalue-up} on $\{\bar n\} || U$.
	 		
	 		\State Write nodes in $L$, $R$, $\{\bar n\}$, and $U$ in one round in the representation of $T$.

 		
 	\end{algorithmic}
 	
 \end{algorithm}

 \begin{figure*}
 	\centering
 	\begin{subfigure}{.48\textwidth}
 		\includegraphics[page=2,trim={0 0 9cm 0},clip,width=.95\linewidth]{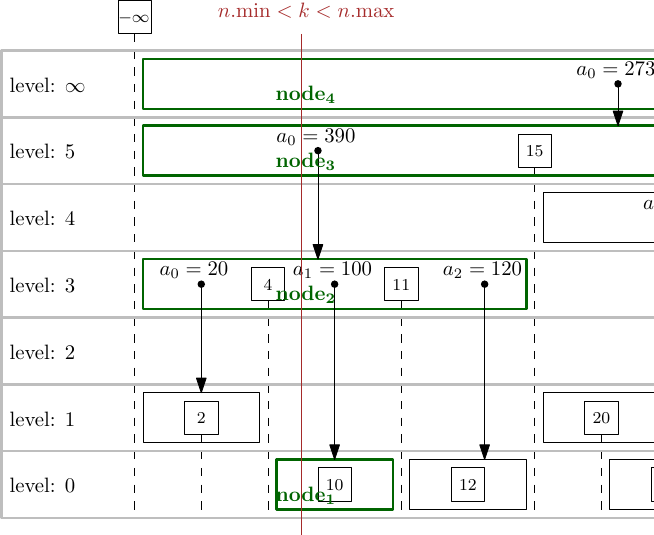}
 		\caption{Before the insertion of $k=9$. Nodes selected for $ N $ are shown with green thick contours. Partition of $N$ is as follows: $D=\{\mbox{node}_1,\mbox{node}_2\}$, $\bar n= \bot$, $U=\{n'=\mbox{node}_3,\mbox{node}_4\}$. }
 	\end{subfigure}
 	\hfill
	 \begin{subfigure}{.48\textwidth}
	 	\includegraphics[page=3,trim={0 0 9cm 0},clip,width=.95\linewidth]{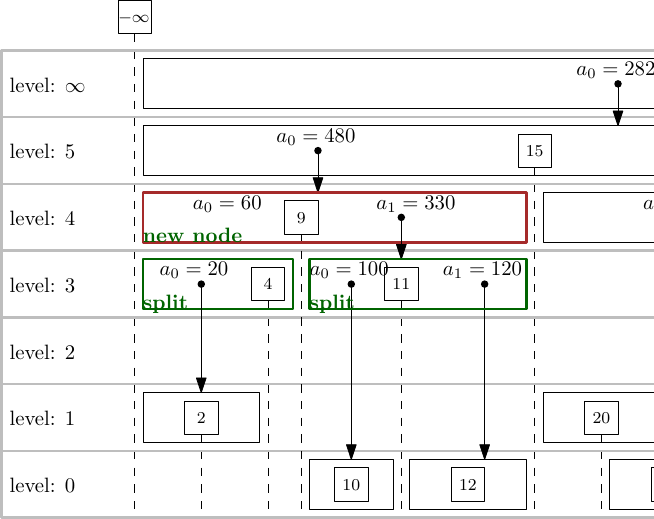}
	 	\caption{State of $ T $ after the execution of the algorithm. The red node is new and it is $\bar n$. Green nodes are created from the split of part of the nodes in $D$ ($\mbox{node}_1$ was not split).}
	 \end{subfigure}
 	\caption{Example of execution of Algorithm~\ref{algo:insert}: insertion of a new key-value pair  $\langle k=9, v=90 \rangle$. The randomly obtained level for this key is $4$. The aggregation function is SUM and for each $k$ its value is $v=10k$. }
 	\label{fig:insert}
 \end{figure*}

\textbf{Insertion.} Algorithm~\ref{algo:insert} shows the procedure to insert a key $k$ in $T$ under the
hypothesis that $k$ is not already contained in $T$. The first line performs
the same query as for Algorithm~\ref{algo:update} to retrieve an ascending path
in $T$, denoted by $N$, of nodes that are involved in the insertion. The expected size of $N$ is $O(\log |T|)$ by Lemma~\ref{lem:expected-maximum-level}. 
Then, a random level $l$, where $k$ should be inserted, is obtained
(Line~\ref{line:alg:insert:level}). 
The following lines aim at identifying the node $\bar n$ in $N$ at
level $l$ and the node $n'$ that is the parent of $\bar n$. 
They also handle the case in which 
$\bar n$ does not exist. In this case, a new node is
inserted with no keys (for the moment). 
Note that, $n'$ always exists since at least the root of $T$ exists.
This and other operations in the rest of the algorithm may make \emph{present} some previously \emph{missing} aggregate value.
These aggregate values are actually fixed or inserted by Algorithm~\ref{algo:update-aggrvalue-up}, which is called at the end of Phase~\ref{phase:change}.
The path $N$ is also partitioned into $D$ (nodes below $\bar n$), $\{\bar n\}$, and $U$ (nodes above $\bar n$, starting with $n'$).

Actual insertion of $\langle k, v\rangle$ in $\bar n$ is performed in Line~\ref{line:algo:insert:insertnewpair}. Then, all nodes below $\bar n$ in $N$, i.e., those that are in $D$,
are split to keep Invariant~\ref{inv:child-range}. In fact, all nodes in $D$
(actually all nodes in $N$) were selected to have $k$ in their range, but after
the insertion of $k$ in $\bar n$, $k$ is contained in a node above them.
The split is
performed from top to bottom in the cycle starting at
Line~\ref{line:algo:insert:split-for}.
It sets min and max of each node considering the existence of $k$, 
according to Invariant~\ref{inv:child-range} (see Figure~\ref{fig:insert}) 
and sets the aggregate sequences of the resulting nodes, according to 
Invariant~\ref{inv:keys-order-range}.
The splits of nodes in $D$ form two 
branches whose nodes are stored in sequences $L$ and $R$ (in ascending order of level).
Special care is taken not to store in $L$ and $R$ nodes that do not contain any key.
When creating or modifying nodes in $L$ and $R$, as well as $\bar n$ and $n'$, the algorithm do not
care about setting the correct aggregate values that intersect or are adjacent to $k$. These are
updated and possibly inserted by calling
Algorithm~\ref{algo:update-aggrvalue-up} on $L|| \bar n$ and $R || \bar n$, if $L$ or $R$ are not empty, and then on $\bar n || U$. 
This makes the affected aggregate values to comply with 
Invariant~\ref{inv:recursive-aggregation}.
Clearly the resulting number of nodes is at most $2|N|$ and hence still expected
$O(\log |T|)$.

\begin{algorithm}
	\caption{This algorithms deletes a key $k$ from a DB-tree $T$.}
	\label{algo:delete}
	\begin{algorithmic}[1] 
		
		\Require A key $k$ and a representation of a DB-tree $T$. We assume $k$ is contained in $T$.
		\Ensure The key $k$ is no longer contained in $T$.
		
		\State $N \gets $ the sequence of nodes resulting from the selection from $T$ of all nodes $n$ such that $n.\mathrm{min} \leq k \leq n.\mathrm{max}$ ordered by ascending $n.\mathrm{level}$. This is performed in a single query round.
		
		\State $\bar n \gets $ the node in $N$ that contains $k$.

		\State $l \gets  \bar n.\mathrm{level}$
		
		\State Let $L$ and $R$ two arrays indexed by $0,\dots,l-1$, all their elements are initialized with $\bot$.
		\State $L[i] \gets$ the node $n$ in $N$ such that $n.\mathrm{level}=i<l$. \Comment Note that, it also holds $n.\mathrm{max}=k$.
		\State $R[i] \gets$ the node $n$ in $N$ such that $n.\mathrm{level}=i<l$. \Comment Note that, it also holds $n.\mathrm{min}=k$.

		\State $U \gets$ the sequence of nodes $n$ in $N$ such that $n.\mathrm{level}>l$.

		\State Delete from $\bar n.\mathrm{aseq}$ the key-value pair for $k$ and coherently $\bar n.m \gets \bar n.m - 1$. Delete also the aggregate values at the left and right of $k$, if they are present. 

		\State Let $D$ an empty sequence of nodes.

		\State $n_\mathrm{prev} \gets \bar n$
		\For { $i$ from $l-1$ down to $0$ }
		
			\State $k_\mathrm{min} \gets $ the key right before $k$ in $n_\mathrm{prev}.\mathrm{aseq}$, if it exists, otherwise $n_\mathrm{prev}.\mathrm{min}$ 

			\State $k_\mathrm{max} \gets $ the key right after $k$ in $n_\mathrm{prev}.\mathrm{aseq}$, if it exists, otherwise $n_\mathrm{prev}.\mathrm{max}$					

			\Switch {}
				\Case { $L[i]\neq\bot$ and $R[i]\neq\bot$} 
					\State Create a new node $n$.  \Comment Merge

					\State $n.\mathrm{level} \gets i$, 
					    $n.\mathrm{min} \gets k_\mathrm{min}$, 
					    $n.\mathrm{max} \gets k_\mathrm{max}$.

					\State Let  $L[i].\mathrm{aseq} = a^L_0, p^L_1,a^L_1,\dots,p^L_m, a^L_m$.

					\State Let  $R[i].\mathrm{aseq} = a^R_0, p^R_1,a^R_1,\dots,p^R_{m'}, a^R_{m'}$. 
					    
					\State $n.\mathrm{aseq} \gets a^L_0, p^L_1,a^L_1,\dots,p^L_m,  p^R_1,a^R_1,\dots,p^R_{m'}, a^R_{m'}$ 
					
						 
				\EndCase
				
				\Case { $L[i]\neq\bot$ and $R[i]=\bot$} 
					\State Let $n = L[i]$.
					\State $n.\mathrm{max} \gets k_\mathrm{max}$ \Comment Expand $n$ rightward
				\EndCase
				
				\Case { $L[i]=\bot$ and $R[i]\neq\bot$} 
					\State Let $n = R[i]$.
					\State $n.\mathrm{min} \gets k_\mathrm{min}$ \Comment Expand $n$ leftward
				\EndCase
			\EndSwitch
			\State Add $n$ to the beginning of $D$ 			\Comment Nodes in $D$ turn out to be in ascending order of level.
			\State $n_\mathrm{prev} \gets n$
		\EndFor

		\State $P \gets 
				\begin{cases}				
					\{ \bar n \}, & \mbox{if }  \bar n.m > 0 \\
					\emptyset,  & \mbox{otherwise}
				\end{cases}$ 

		\State Update affected aggregate values by calling Algorithm~\ref{algo:update-aggrvalue-up} on $D||P||U$.
		Note that, this also adds currently missing aggregate values that have corresponding children. 
		
		\State Write nodes in $P$, $D$, and $U$ in the representation of $T$, in one round. 
		
	\end{algorithmic}
	
\end{algorithm}

 \begin{figure*}
	\centering
%
	\begin{subfigure}{.48\textwidth}
		\includegraphics[page=5,trim={0 0 9cm 0},clip,width=.95\linewidth]{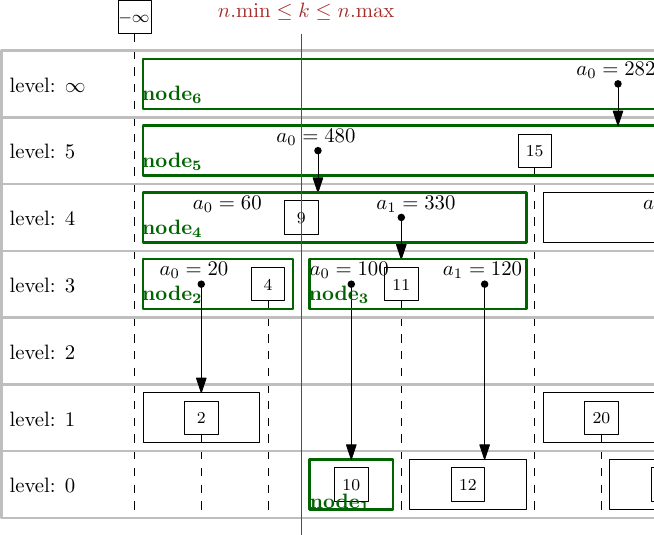}
		\caption{ Before the deletion of $k=9$. Nodes selected for $ N $ are shown with green thick contours. Partition of $N$ is as follows: $\bar n= \mbox{node}_4$,  $U=\{\mbox{node}_4,\mbox{node}_5,\mbox{node}_6\}$, $L=\{\mbox{node}_2\}$, $R=\{\mbox{node}_1,\mbox{node}_3\}$.}
	\end{subfigure}
	\hfill
	\begin{subfigure}{.48\textwidth}
		\includegraphics[page=6,trim={0 0 9cm 0},clip,width=.95\linewidth]{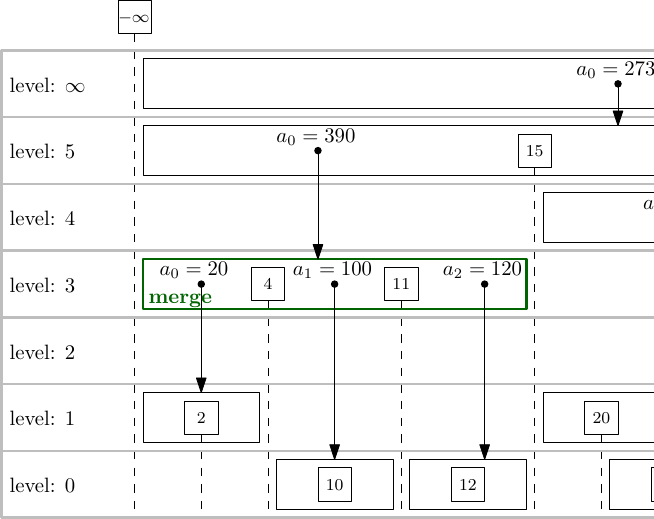}
		\caption{State of $ T $ after the execution of the algorithm. The node shown with green thick contour derives
			from the merge of $\mbox{node}_2$ and $\mbox{node}_3$, $\mbox{node}_4$ was removed,  the execution set $\mbox{node}_1.\mbox{min}=4$ (was $\mbox{node}_1.\mbox{min}=9$ before).}
	\end{subfigure}
	\caption{Example of execution of Algorithm~\ref{algo:delete}: deletion of key $k=9$. The aggregation function is SUM and for each $k$ its value is $v=10k$.}
	\label{fig:delete}
\end{figure*}

\textbf{Deletion.} Algorithm~\ref{algo:delete}  shows the procedure to delete a key $ k $ in $T$
under the hypothesis that $ k $ is already present in $ T $. See
Figure~\ref{fig:delete} for an example of execution of this algorithm. The first
line performs a query to retrieve all nodes $n$ having $k$ in their range or
$n.\mbox{min}=k$ or $n.\mbox{max}=k$. The nodes whose min and max are equal to $k$
are important for this algorithm, since the deletion of $k$ affects their range.
This set, denoted $N$, contains a node
$\bar n $ that contains $k$. We denote by $l$ its level. The other elements of
$N$ are partitioned into three paths, $U$, $L$, $R$ corresponding to the three
disjoint conditions of the query.  Path $U$ is from $\bar n$ to the root of $T$
(like in Algorithms~\ref{algo:update} and~\ref{algo:insert}). Path $L$ ($R$)
contains nodes $n$ such that $n.\mbox{max}=k$ ($n.\mbox{min}=k$). Paths $L$ and
$R$ are made of descendants of $\bar n$ by Invariant~\ref{inv:child-range}. By
applying Lemma~\ref{lem:expected-maximum-level} on the three paths, we get that
the expected size of $N$ is $O(\log |T|)$. In the algorithm, $L$ and $R$ are
represented as arrays having one element for each level below $l$. The $\bar
n.\mathrm{aseq}$ is updated removing the key-value pair for $ k $. The algorithm
proceeds by performing opposite operations with respect to
Algorithm~\ref{algo:insert}, that is, nodes in $L$ and $R$ that are at the same
level are merged so that they cover a range that is the union of the ranges of the merged nodes. This is
done so that Invariants~\ref{inv:keys-order-range} and \ref{inv:child-range} are preserved. The resulting nodes end up to be
a path of descendants of $\bar n$, denoted by $D$. 
Aggregate values that ``overlap'' 
the deleted key $k$ are not set during the merge. However, this is fixed when  
Invariant~\ref{inv:recursive-aggregation} is enforced by calling
Algorithm~\ref{algo:update-aggrvalue-up} on $D || \{\bar n\}  || U$.

\textbf{Bulk/Batch insertion.} \rnote{C1.1} The insertion of a large number of elements that
are known in advance is more efficiently performed as a single operation. In
particular, if the insertion starts from an empty DB-tree, the whole DB-tree can
be built from scratch in main memory (supposing that this is large enough). This can be easily performed applying, for
each element, a procedure similar to that shown in Algorithm~\ref{algo:insert},
where sets $N$, $U$, and $D$ are got from main memory. The DB-tree in
main memory is updated and kept ready for the next element of the bulk insertion. A the end, the
resulting DB-tree can be written into the database in a single round, possibly
using bulk insertion support provided by the DBMS.

In Section~\ref{sec:exepriments}, we show that single insertions
into DB-trees are quite slow with respect to insertion into regular tables. We
can adopt a bulk-like approach to speed-up many insertions (a \emph{batch}) in a non-empty
DB-tree. In this case, we have to be careful to read first into main memory all the
nodes that are affected by all insertions. This can be done by executing
Lines~\ref{line:alg:insert:getN},~\ref{line:alg:insert:getU},
and~\ref{line:alg:insert:getD} of Algorithm~\ref{algo:insert}, for each element.
Note that, this can be done in one query round. We also have to
keep track of deleted and changed nodes to correctly apply these changes at the
end of the bulk insertion. We do not further detail these procedures that are
simple variations of the algorithms shown in this section.

\subsection{Comparison of DB-trees with Skip Lists and B-trees}

\rnote{C3.2}

Now, we present a detailed comparison of DB-trees with skip
lists~\cite{pugh1998skip} and B-trees~\cite{comer1979ubiquitous}.

The most important distinguishing feature of DB-trees is related to the
representation of relationships between nodes. In DB-trees, we do not store any
pointer in the nodes: a parent-child relationship between nodes $n_1$ and $n_2$ is
implicitly represented by a containment relationship of $\range{n_1}$ and
$\range{n_2}$, where the extremes of ranges are explicitly represented in the
nodes. This approach allows us to obtain a path to the root in a single query
round and makes DB-trees very well suited to be represented in databases. On
the
contrary, skip lists and B-trees do use explicit pointers, thus implicitly
assuming that pointer traversal is an efficient primitive operation, which is
not true for databases.

\rnote{C3.3} Another fundamental difference with respect to common data
structures is that DB-trees are targeted to support aggregate range queries and
authenticated data structures, and are not targeted to speed up search. In fact,
DBMSes already perform searches very efficiently. On the contrary, DB-trees rely
on DBMS search efficiency to speed up aggregate range queries. DB-trees reach 
this objective without relying on traversals.

In Section~\ref{ssec:db-tree-intuitive}, we described DB-trees starting from
skip lists. In fact, each skip list instance is in one-to-one correspondence with
a DB-tree instance. Both data structures associate levels with keys and the level
of each key is selected in the same random way. However, while a skip list
redundantly stores a tower of elements for each key, the corresponding DB-tree
stores each key only once: essentially only the top element of the tower is
represented. Further, in a DB-tree, sequential keys in a level may be grouped
into a single node equipped with some metadata (level, min, and max) that allow
them to be efficiently retrieved from the database. In DB-trees, a node is the
smallest unit of data that is selected, retrieved, or updated when the DBMS is
contacted. Due to all these differences, algorithms performing operations on DB-trees turns out to be very different from the corresponding algorithms for skip lists.
Since the level associated with each key is the same in DB-trees and
skip lists, statistical properties of DB-trees (see
Section~\ref{ssec:DBtrees:fundamental-properties}) also hold for skip lists
(e.g., Lemma~\ref{lem:expected-maximum-level}), or can be recast to be
applicable in the skip list context (e.g.,
in the skip list context, Property~\ref{prop:constant-size-node} can be interpreted as related to the
expected number of consecutive tower-top elements in a level).

While statistical aspects of DB-trees are very similar to skip list ones,
storage aspects are somewhat similar to B-trees. In B-trees, each node $n$ is
meant to be stored in a disk block and contains a selection of keys (or
key-value pairs) interleaved by pointers to blocks storing the children of $n$.
In both DB-trees and B-trees, descendants of $n$ store keys that are between two
keys that are consecutively stored in $n$. In B-trees, each node can contain a number of
keys between $d$ and $2d$, where $ d $ is a fixed parameter. Hence, a node can have a number 
of children between $d+1$ and $2d+1$. In DB-trees, each node contains at least one key, but there is no
maximum number of keys for a node. The number of keys in a node is only
statistically constrained (see Property~\ref{prop:constant-size-node}). Insertion of a new key in a B-tree
occurs in a node that is deterministically identified and may provoke the split of
zero or more nodes, possibly all the way up to the root of the tree, in order to respect
the maximum number of keys per node. In DB-trees, the node $n$ where a new key is
inserted is at a level that is randomly selected, possibly creating a new one if
needed. Insertion of a key into node $n$ may provoke the split of nodes below
$n$ to respect Invariant~\ref{inv:child-range}. Deletion in B-trees may require
re-balancing through rotations to respect the minimum number of keys per node.
In DB-trees deletion of a key from node $n$ may provoke the merge of nodes below
$n$, which is exactly the opposite of what occurs during insertion. The maximum
depth of
a B-tree is 
$\left\lfloor \log_{d} (x+1)/2 \right\rfloor$, where $x$ is the number of keys
in the B-tree~\cite{cormen2009introduction}. The depth of the DB-tree is only
statistically characterized (see Lemma~\ref{lem:expected-maximum-level}).

\rnote{C3.3}Finally, we note that the depth of a B-tree and the number of
levels of a skip list are closely related to the efficiency of the operations on
the respective data structure.  For DB-trees, the depth of the tree is related
to the size of the data that should be handled by any operation.
As described in Section~\ref{sec:algorithms}, DB-tree operations, which are locally performed by the overlay logic, take
a time that is
proportional to the handled data and hence turns out to be $O(\log |T|)$ on
average, as for skip lists. However, it should be noted that
the selection of nodes involved in a query or change is not performed by the
implementation but is delegated to the DBMS. Its ability in efficiently executing
the queries requested by the overlay logic has a significant impact on the
overall performance (see also the experiments reported in
Section~\ref{sec:exepriments}).


\section{Experiments}\label{sec:exepriments}

We developed a software prototype with the intent to provide experimental
evidences of the performances of our approach.

The objectives of our experiments are the following.

\begin{enumerate}[label=O\arabic*.,ref=O\arabic*]
	
	\item \label{obj:range-time} We intend to show the response times of aggregate
range queries that we obtain by using DB-trees. We expect them to be ideally 
logarithmic in the size of the selected range. We compare them against the
response time obtained by using a plain DBMS with the best available indexing.

	\item \label{obj:ins-del-time} We intend to measure the response time for 
insert and delete operations using DB-trees and compare them against the response time of the
same operations performed using a plain DBMS. 
	
\end{enumerate}


Tests were performed on a 
high-performance notebook with 16GB of main memory (SDRAM DDR3), M.2 solid-state
drive (SSD), Intel Core i7-8550U processor (1.80GHz, up to 4.0GHz, 8MB Cache). 
To understand the behavior of our approach on a less performant hard-disk-based system, 
we also performed the tests on an old machine equipped with Intel Core 2 Quad Q6600 2.4GHz, 8 GB of main memory and a regular hard disk drive (HDD). 
In the following, the two platforms are shortly referred to as \emph{SSD} and \emph{HDD}, respectively. Both platforms are linux-based 
and during the tests no other
cpu-intensive or i/o-intensive tasks were active. 

We repeated all experiments with two widely used DBMSes: PostgreSQL (Version 11.3) and MySql (Version 8.0.16).
DBMSes were run within docker
containers, however, no limitation of any kind was configured on those containers.

Our software runs on the Java Virtual Machine (1.8). Since the JVM
incrementally compiles  the code according to the ``HotSpot''
approach~\cite{paleczny2001java}, for each test we took care to let the system
run long enough before performing the measurements, to be sure that compilation
activity was over: we run the whole tests more than once and consider only the
times measured in the last run. The DB-tree code is written using the Kotlin
programming language and adopt the standard JDBC drivers for the connection with
the DBMSes. Time measurements are performed using the Kotlin standard API
\emph{measureNanoTime}, that executes a given block of code and returns elapsed
time in nanoseconds. In our experiments, we set the GO-UP probability for DB-trees to 0.5 (see Section~\ref{ssec:db-tree-intuitive}).

Regarding Objective~\ref{obj:range-time}, we prepared a dataset of 1 million
key-value pairs, \rnote{C4.2} whose keys are the integers from 1 to 1 million and whose values are
random integers. We randomly shuffled the pairs and inserted them into the DBMS. For
testing the plain DBMS, we created a table with two columns, with the key of the pair as primary
key of the table, and the two possible indexes on (key, value) and on (value, key). For the table that 
represents the DB-tree, 
we have columns for $n.$min, $n.$max, $n.$level, plus a column that contains a serialized form of the 
whole node. We do not perform any selection on this last column. Its content is just returned to the client.
\rnote{C3.9} We configured ($n.\mathrm{min}$, $n.\mathrm{max}$,
$n.\mathrm{level}$) as the primary key index and ($n.\mathrm{max}$, $n.\mathrm{min}$,
$n.\mathrm{level}$) as index.
In both PostgreSQL and MySQL, indexes are plain B-trees. In MySQL, primary key indexing is clustered.

\rnote{C1.1}Insertion of DB-trees was performed with a prototypical implementation of the
bulk insertion approach described in Section~\ref{ssec:DBTree:upd-ins-del}. The
time taken for bulk insertions are reported in
Table~\ref{table:spaceoccupation}.

We show tests using the SUM aggregation function, which is a very widely used one and whose optimization is likely to be an objective of DBMS designers. 
In the tests based on a plain DBMS, we performed SQL queries like the following, where
we used  self explanatory names.
\begin{verbatim}
SELECT SUM(value) 
FROM test_table
WHERE range_start<=key AND key<=range_end
\end{verbatim}
For the DB-trees tests, we used our realization of Algorithm~\ref{algo:aggregate-range-query}.

\rnote{C3.8} Actually, we also performed the same tests with aggregation
functions MIN, MAX, COUNT, and AVG. The MIN and MAX functions are highly
optimized in both DBMSes, and queries take very small time to execute,
independently from the size of the range. Hence, there is no point in using
DB-trees for MIN and MAX, at least in the considered DBMSes. On the contrary,
the results for tests with COUNT and AVG are very much like those that we show
for function SUM, hence we decided not to show additional charts for them.

\begin{figure}
	\centering
	\includegraphics[width=\linewidth,trim={3cm 0 3.5cm 0},clip]{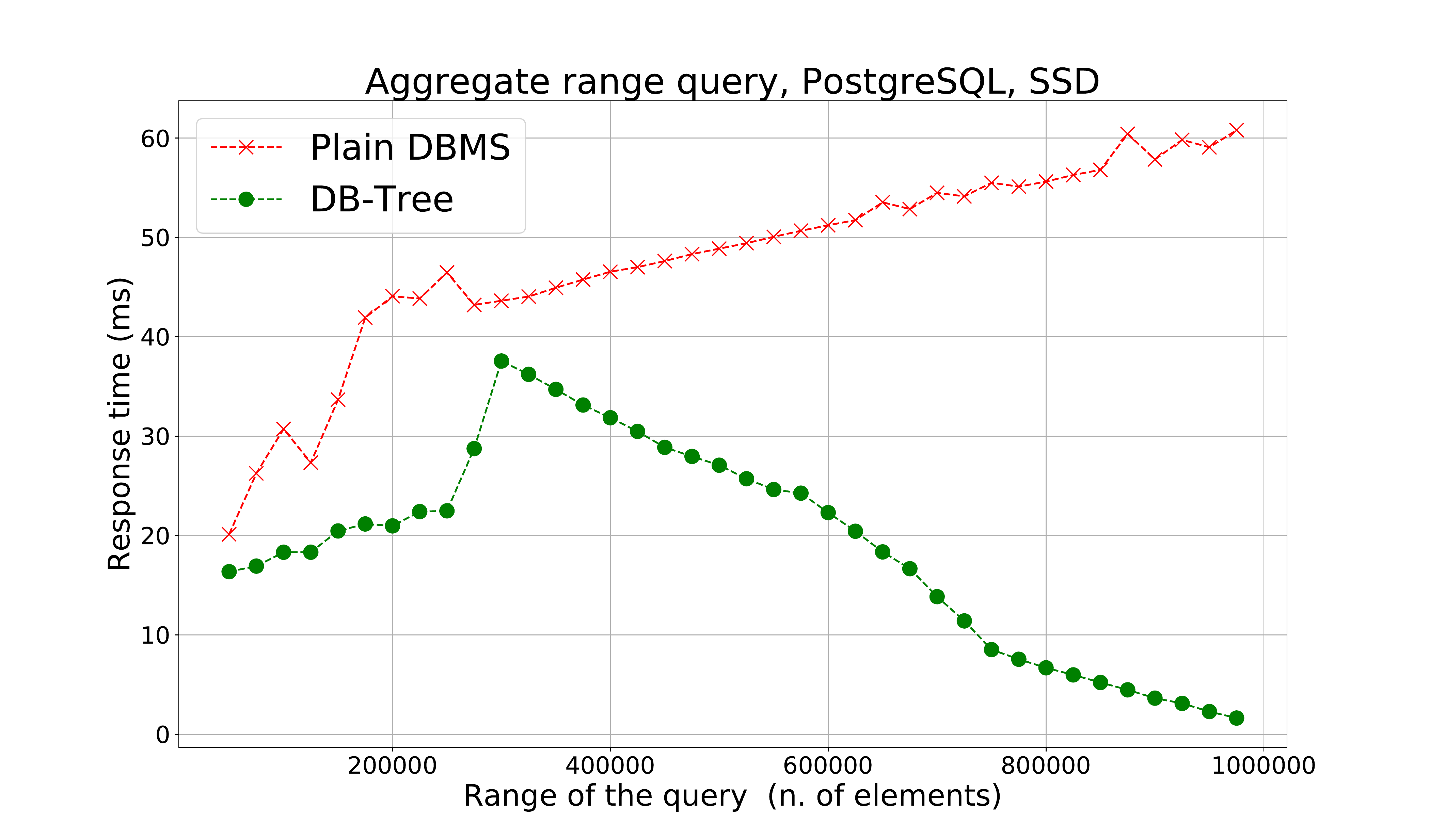}
	\infloatrnote{C4.2}\caption{Performances of aggregate range queries with a PostgreSQL DBMS, on the SSD platform.} 
	\label{fig:pos_SUM}
\end{figure}
\begin{figure}
	\centering
	\includegraphics[width=\linewidth,trim={2.5cm 0 3.5cm 0},clip]{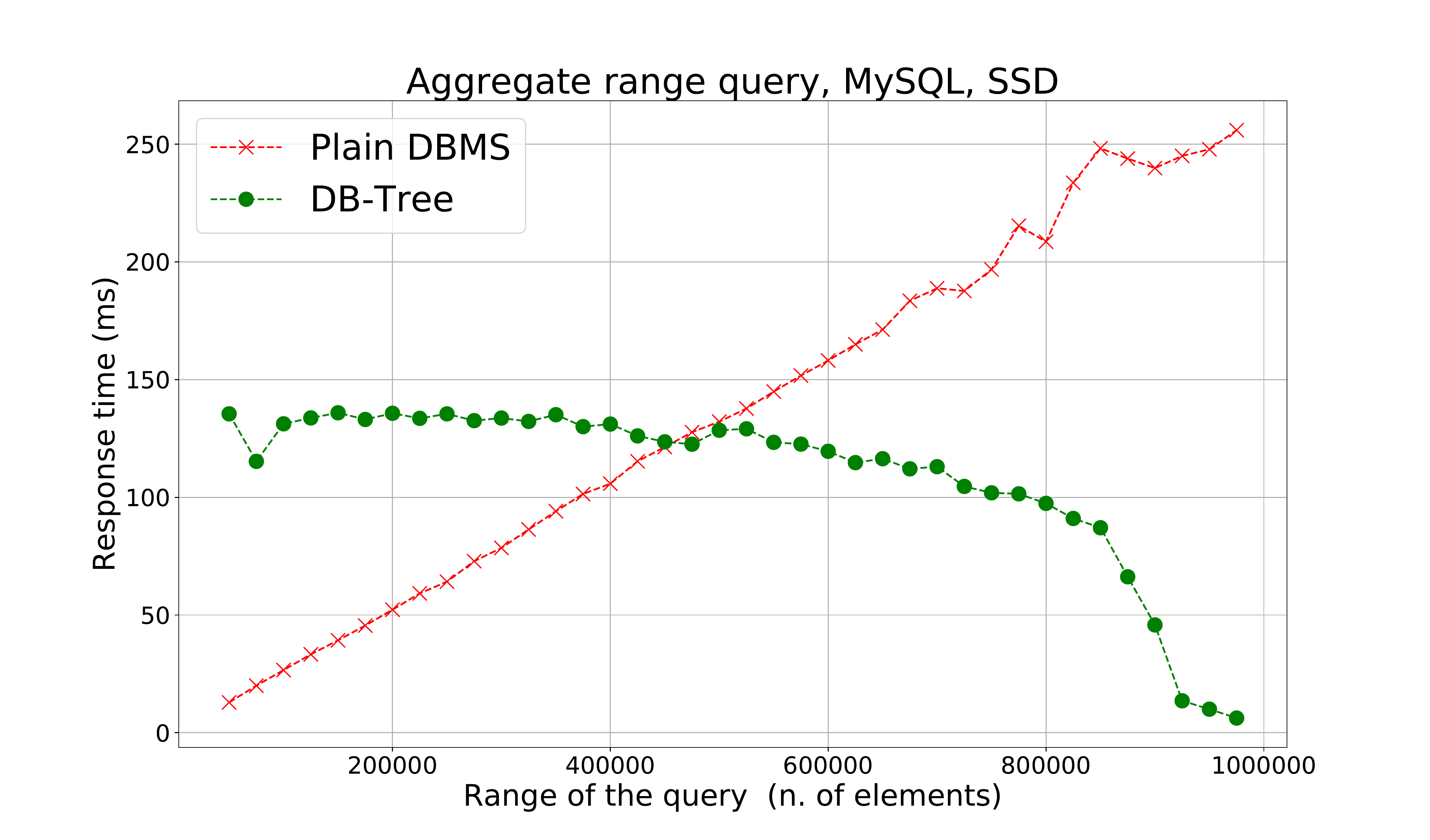}
	\infloatrnote{C4.2}
	\caption{Performances of aggregate range queries with a MySQL DBMS, on the SSD platform.}
	\label{fig:my_SUM}
\end{figure}
\begin{figure}
	\centering
	\includegraphics[width=\linewidth,trim={3cm 0 3.5cm 0},clip]{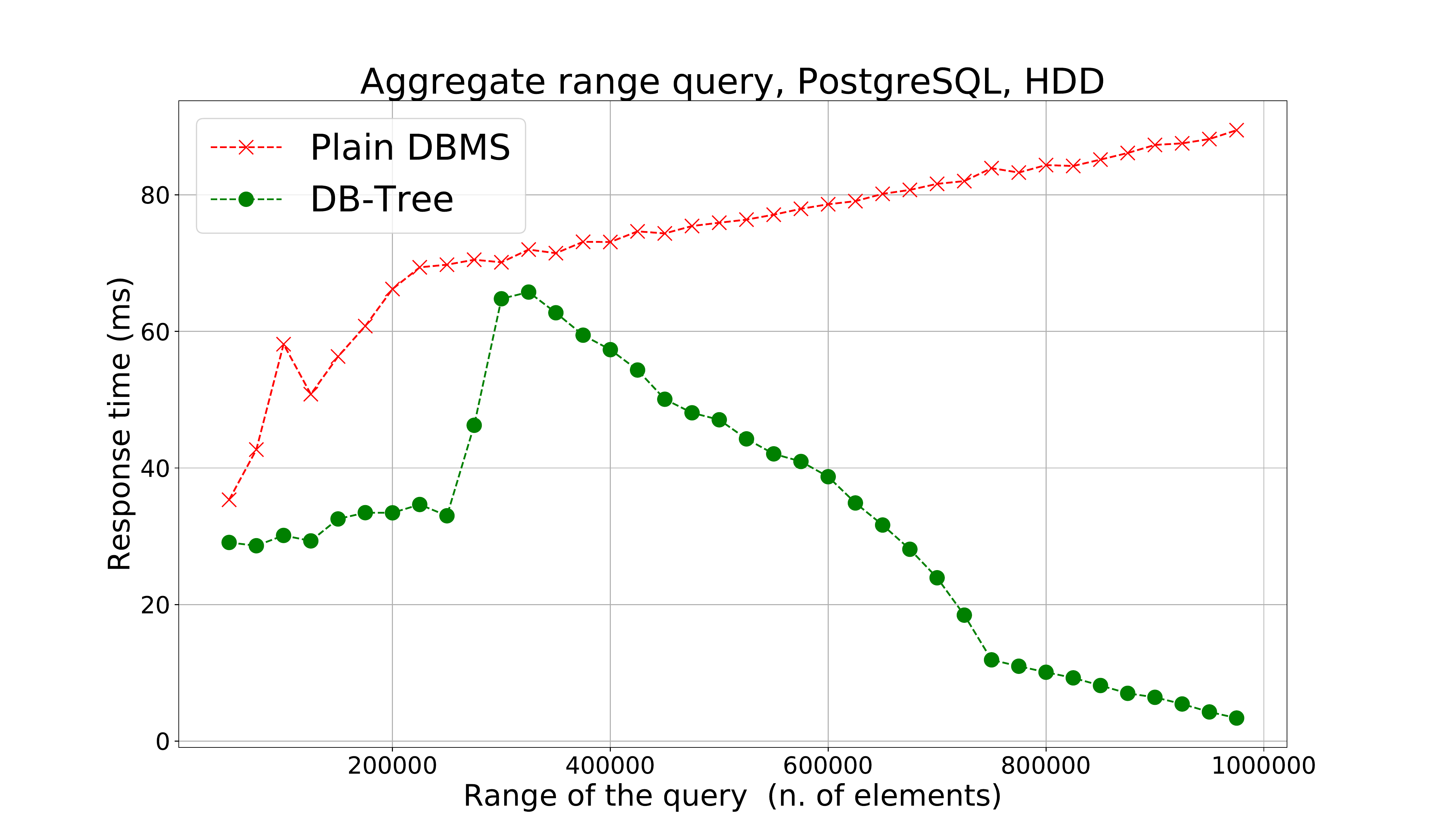}
	\infloatrnote{C4.2}
	\caption{Performances of aggregate range queries with a PostgreSQL DBMS, on the HDD platform.}
	\label{fig:pos_SUMHDD}

\end{figure}
\begin{figure}
	\centering
	\includegraphics[width=\linewidth,trim={3cm 0 3.5cm 0},clip]{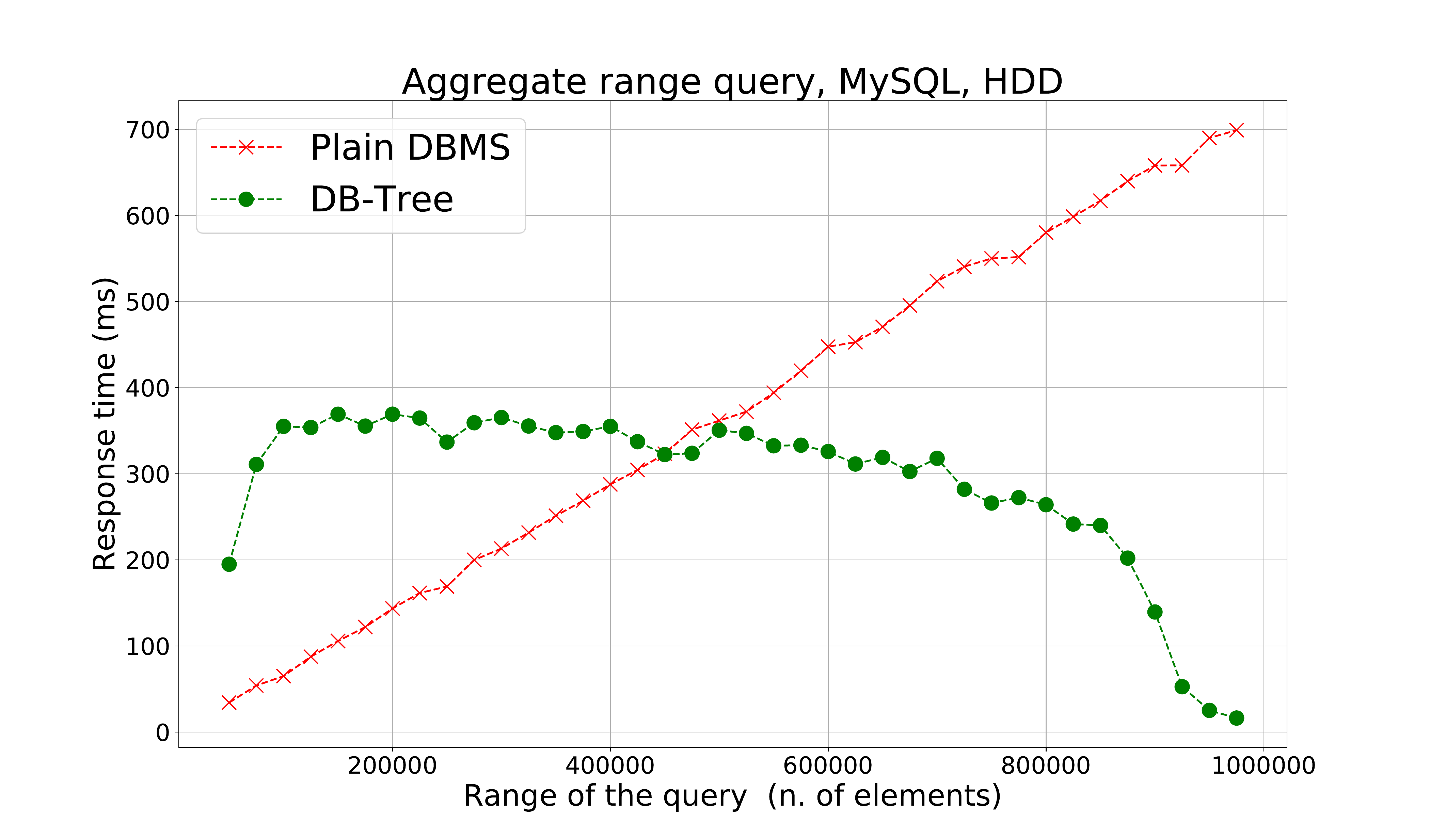}
	\infloatrnote{C4.2}
	\caption{Performances of aggregate range queries with a MySQL DBMS, on the HDD platform.} 
	\label{fig:my_SUMHDD}
\end{figure}

We generated the queries to be performed for the tests in the following way. We
considered range sizes from \rnote{C4.2}50.000 to 975.000, with steps of 25.000.
For each range size, we executed 200 queries with random ranges of that size
(i.e., with a random shift) and took the average of execution times. The dataset
and the ranges are the same for the tests using plain DBMSes and for the tests
using DB-trees.

Queries are executed without any delay in between. Results are shown in
Figures~\ref{fig:pos_SUM} and~\ref{fig:my_SUM}, for platform SSD, and in
Figures~\ref{fig:pos_SUMHDD} and~\ref{fig:my_SUMHDD}, for platform HDD. In each
figure, the x-axis shows the size of the range and the y-axis show the query
duration in milliseconds. Figures show performances for tests on plain DBMSes and for
tests adopting DB-trees. Using plain DBMS queries, the response time for the aggregate
range query is linear, for both PostgreSQL and MySQL. The tests show that, using
DB-trees, the response time is limited and well below the one obtained using
plain DBMS queries, starting from a certain range size. Concerning the shape
of the curves, we note that aggregate range queries are theoretically easy in
two extreme cases %
\begin{inparaenum}[(i)]
\item when the range is just a tuple, in this case an aggregate range query is
equivalent to a plain selection,  and %
\item when the range covers all the data, in this case a good strategy is to
keep an accumulator for the whole dataset. %
\end{inparaenum}
In all charts related to aggregate range queries, we note that the
curve for DB-trees response time is somewhat bell-shaped, reflecting that hard
instances are in the middle between the two extreme cases just mentioned.

We point out that the DBMS internally performs several non-obvious query
optimizations that can profoundly change the way the same kind of query is
performed on the basis of estimated amount of data to retrieve. A clear effect
of this can be seen in the roughly piecewise linear trend, shown in
Figures~\ref{fig:pos_SUM} and~\ref{fig:pos_SUMHDD}, for the performances of
aggregate range queries in PostgreSQL for plain DBMS tests.
Further, comparing charts for SSD and HDD platforms, we notice the expected
degradation due to the slower HDD technology and a less regular behavior for the HDD
case. However, trends are quite similar for both platforms for all cases.

\rnote{C3.3} Concerning space occupancy on disk, this clearly depends on actual
data types, indexes, and DBMS used. As an example, for the above described
experiments, the occupancy is summarized in Table~\ref{table:spaceoccupation}.
We report occupancy of data and indexes for a plain table, i.e., the one used for
plain DBMS tests, and for the table used for DB-tree
tests. We note that, the number of rows in the DB-tree table
is about half of the number of rows of the plain table.
In fact, in our representation each row represents a node and each node contains
one or more key-value pair. For PostgreSQL, the size of the data for the
DB-tree is roughly the same as the size of the data for the plain table, while
indexes are larger for the latter. For MySQL, the size of the data for the
DB-tree is about twice the size of the data for the plain table, while indexes
are comparable.

\begin{table*}
	\centering
\begin{tabular}{|c|c|c|c|c|c|c|}
	\hline 
	DBMS &  & Rows & Tot. Sz. & Data Sz. & Index Sz. & Bulk Ins. Time \\ 
	\hline 
	\multirow{2}{*}{PostgresSQL} & Plain & 1M & 124MB & 42MB & 82MB & 35s\\ 
	\cline{2-7} 
	& DB-Tree & 505K & 89MB & 43MB & 46MB & 19s\\ 
	\hline 
	\multirow{2}{*}{MySQL} & Plain & 1M & 81MB & 43MB & 38MB &  2m35s\\ 
	\cline{2-7}
	& DB-tree & 505K & 118MB & 86MB & 32MB & 50s \\ 
	\hline 
\end{tabular}\infloatrnote{C1.1, C3.3} 

\caption{ Disk space occupancy and bulk insertion times for data used in the
	experiments of Section~\ref{sec:exepriments}.}
	\label{table:spaceoccupation}
\end{table*}


%

\begin{figure}
	\centering
	\includegraphics[width=\linewidth,trim={3cm 0 3.5cm 0},clip]{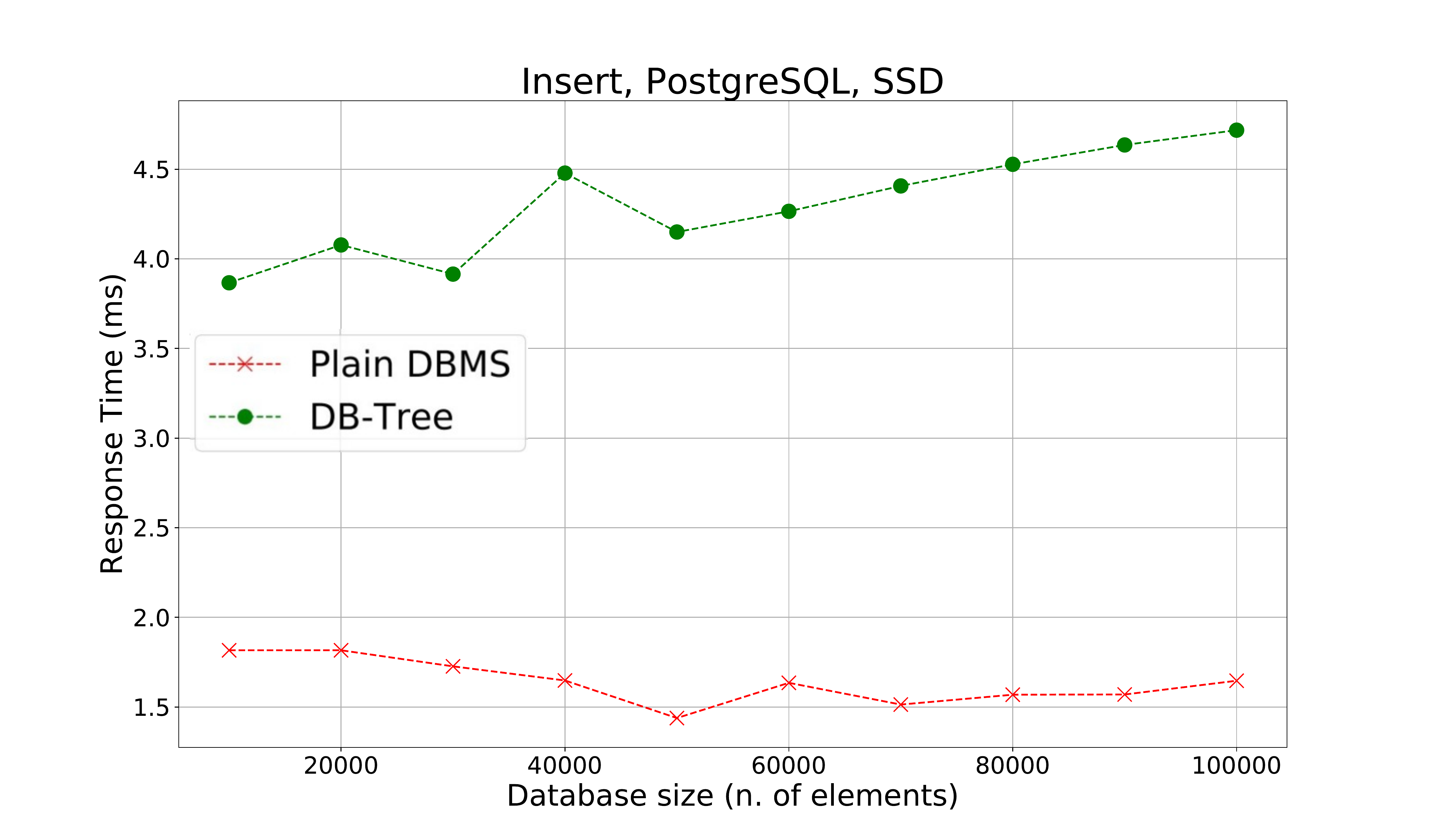}
	\infloatrnote{C4.2}\caption{Performance of the insert operation of a key with a PostgreSQL DBMS, on the SSD platform.} \label{fig:pos_ins}
\end{figure}

\begin{figure}
	\centering
	\includegraphics[width=\linewidth,trim={2.5cm 0 3.5cm 0},clip]{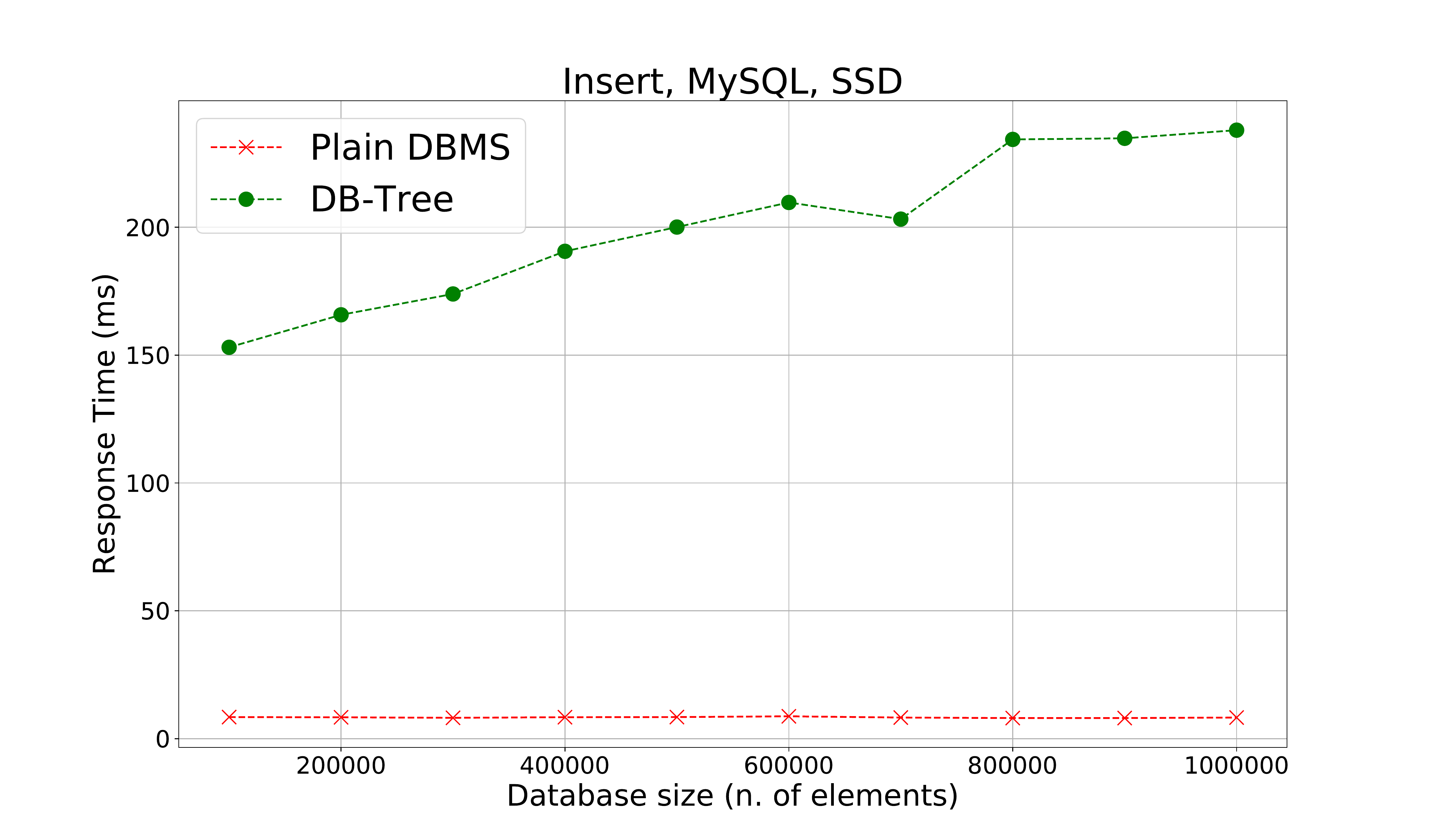}
	\infloatrnote{C4.2}\caption{Performance of the insert operation of a key with a MySQL DBMS, on the SSD platform.} 
	\label{fig:my_ins}
\end{figure}

\begin{figure}
	\centering
	\includegraphics[width=\linewidth,trim={3cm 0 4cm 0},clip]{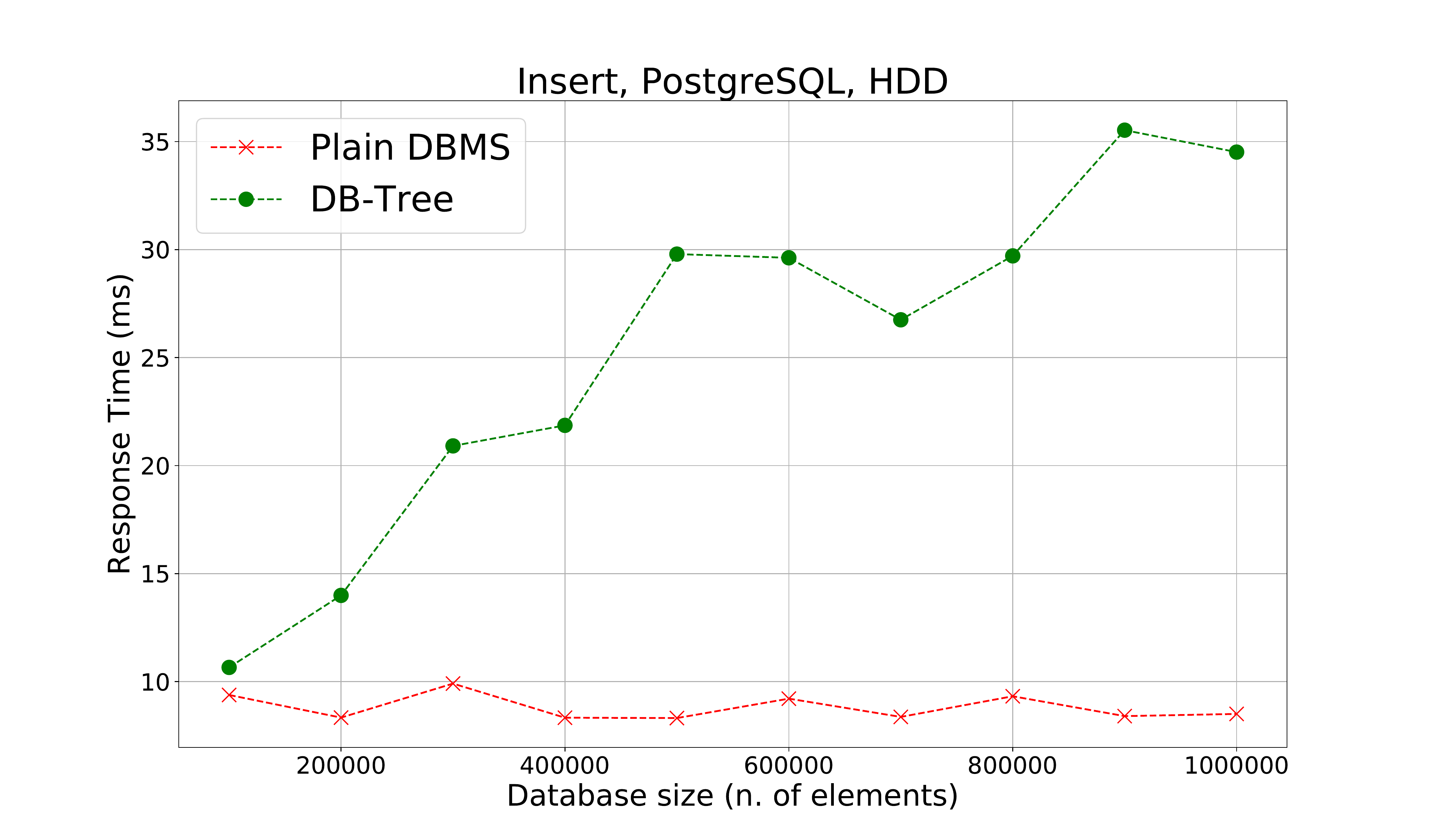}
	\infloatrnote{C4.2}\caption{Performance of the insert operation of a key with a PostgreSQL DBMS, on the HDD platform.} 
	\label{fig:pos_ins_SUMHDD}

\end{figure}

\begin{figure}
	\centering
	\includegraphics[width=\linewidth,trim={2.5cm 0 4cm 0},clip]{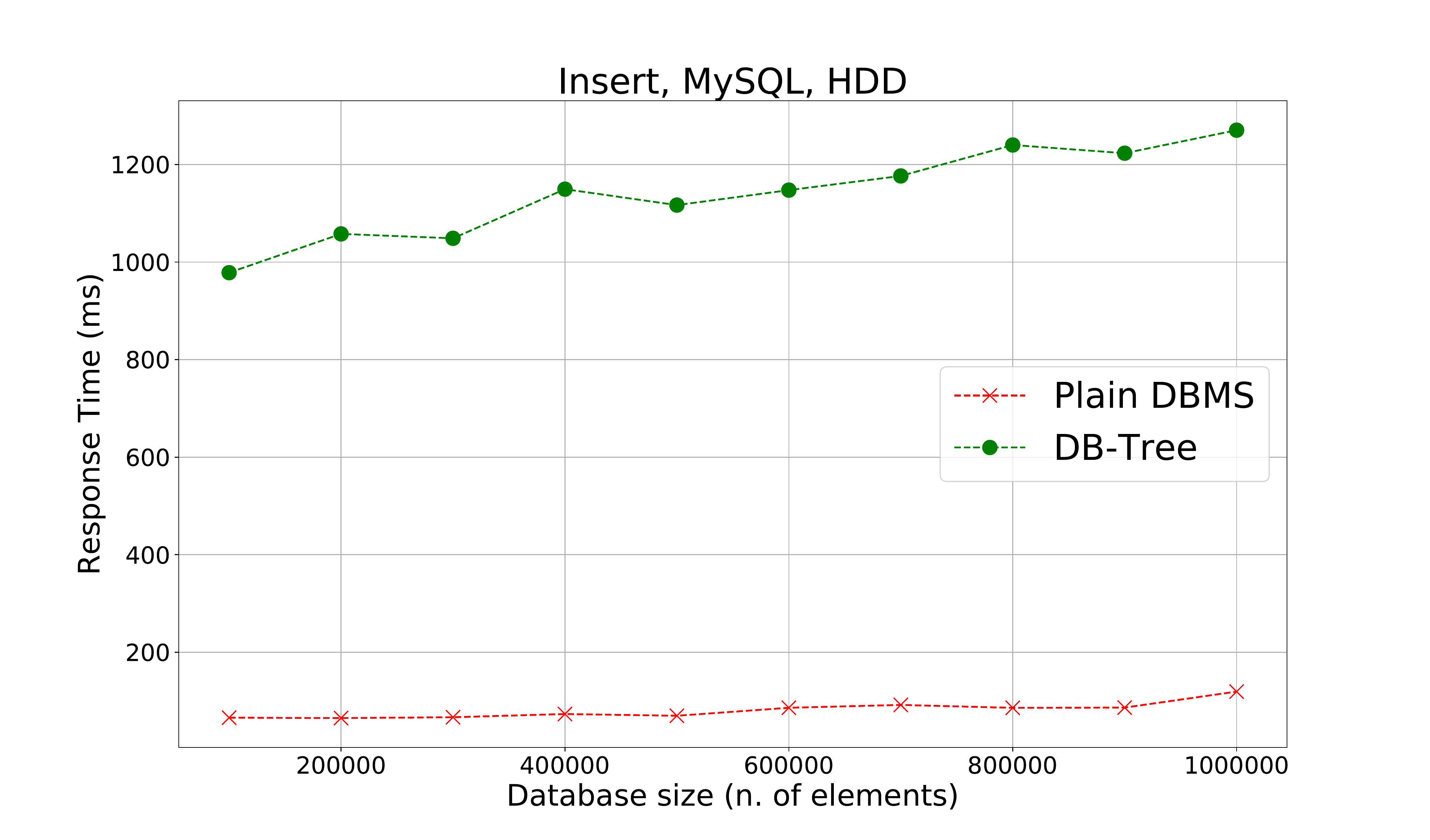}
	\infloatrnote{C4.2}\caption{Performance of the insert operation of a key with a MySQL, on the HDD platform.} 
	\label{fig:my_ins_SUMHDD}
\end{figure}

\begin{figure}
	\centering
	\includegraphics[width=\linewidth,trim={2.5cm 0 4cm 0},clip]{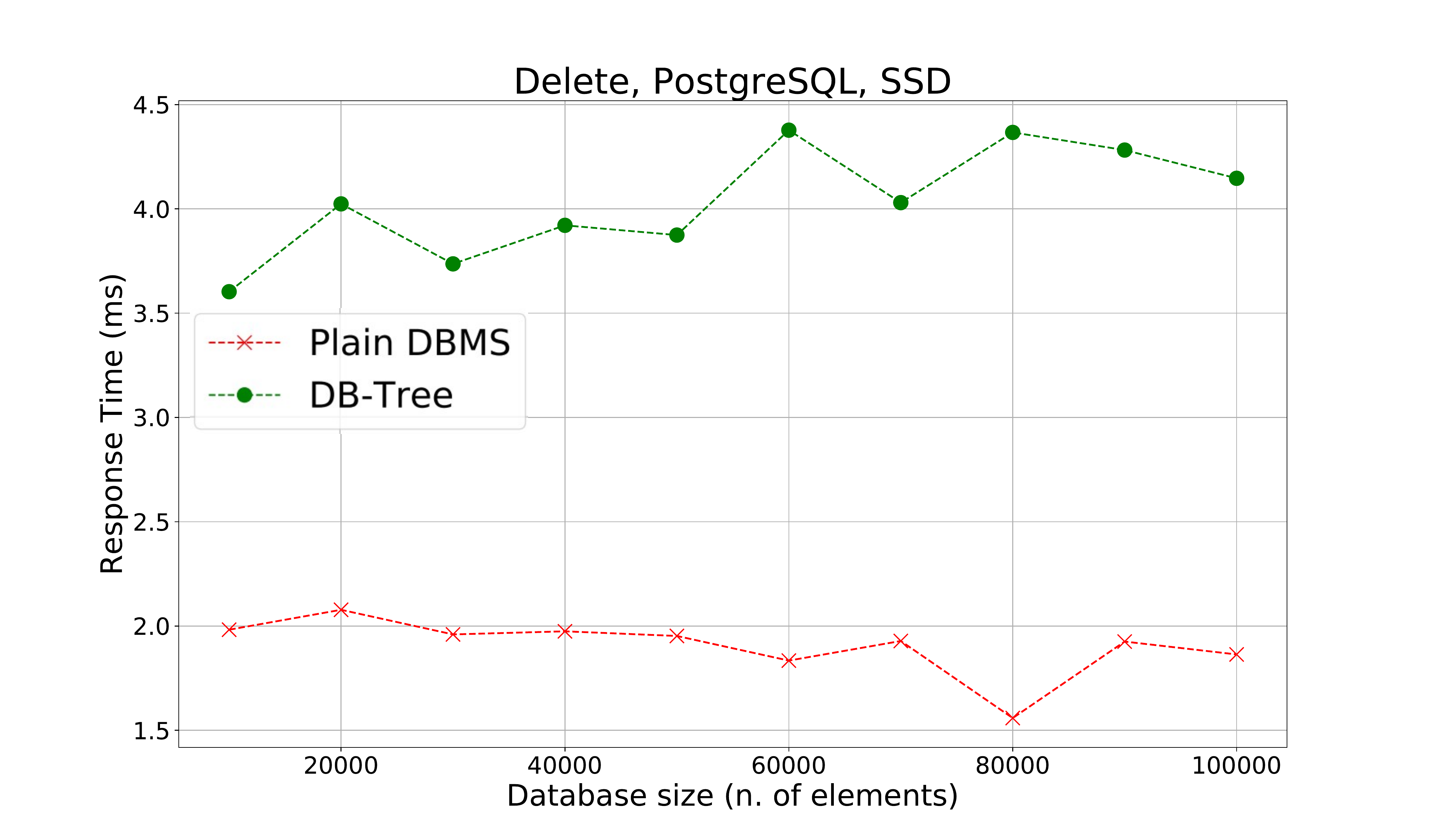}
	\infloatrnote{C4.2}\caption{Performance of the delete operation of a key with a PostgreSQL DBMS, on the SSD platform.} 
	\label{fig:pos_del}
\end{figure}

\begin{figure}
	\centering
	\includegraphics[width=\linewidth,trim={2.5cm 0 4cm 0},clip]{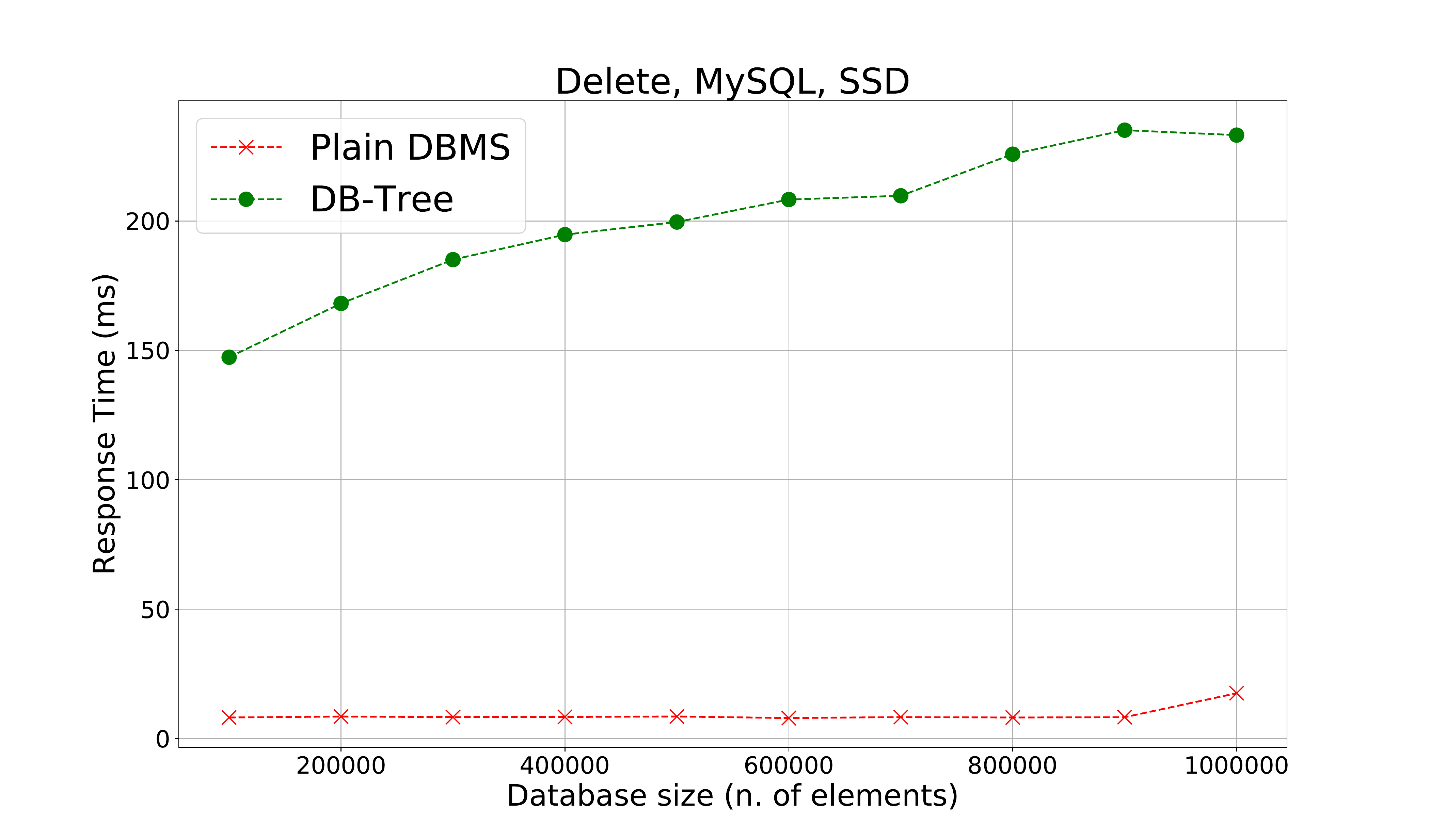}
	\infloatrnote{C4.2}\caption{Performance of the delete operation of a key with a MySQL DBMS, on the SSD platform.} 
	\label{fig:my_del}
\end{figure}

\begin{figure}
	\centering
	\includegraphics[width=\linewidth,trim={3cm 0 4cm 0},clip]{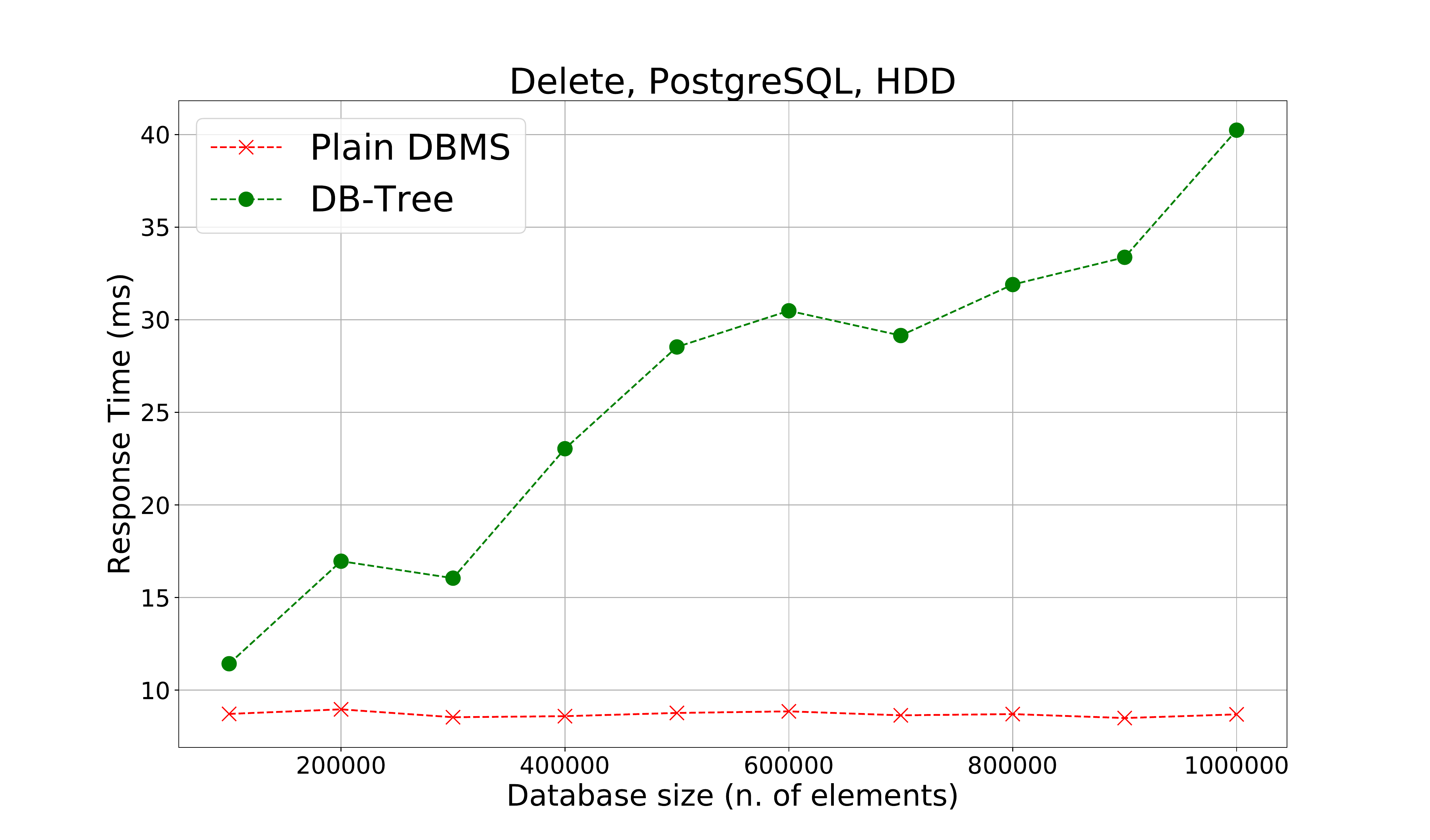}
	\infloatrnote{C4.2}\caption{Performance of the delete operation of a key with a PostgreSQL, on the HDD platform.} 
	\label{fig:pos_del_SUMHDD}
\end{figure}

\begin{figure}
	\centering
	\includegraphics[width=\linewidth,trim={2.5cm 0 4cm 0},clip]{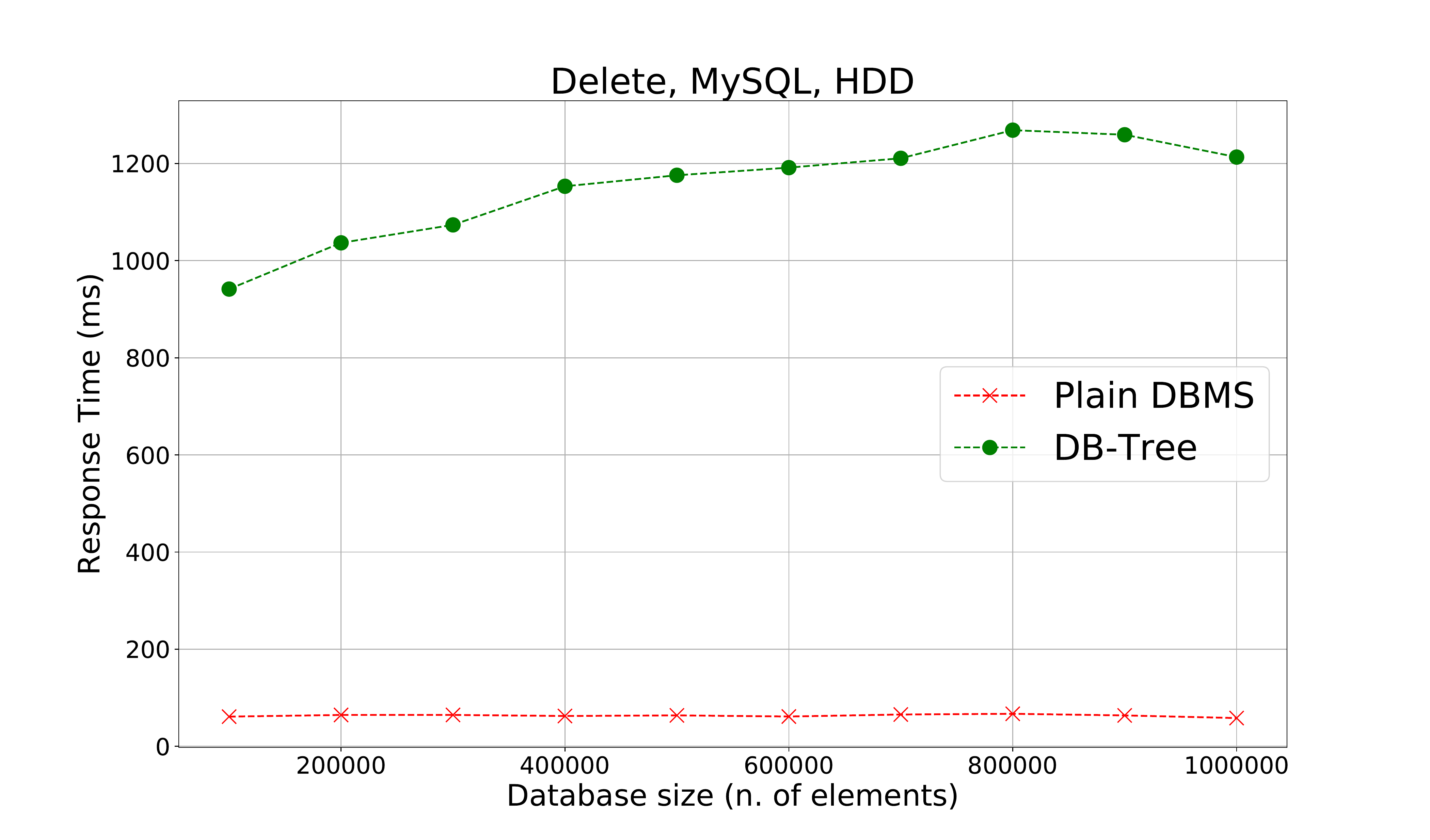}
	\infloatrnote{C4.2}\caption{Performance of the delete operation of a key with a MySQL, on the HDD platform.} 
	\label{fig:my_del_SUMHDD}
\end{figure}

Concerning Objective~\ref{obj:ins-del-time}, we created 10 datasets containing
\rnote{C4.2}from 100.000 to 1 million key-value pairs, with step of 100.000. For each dataset,
we created the corresponding databases for both plain DBMS tests and for
DB-trees tests. Tables and indexes are as for Objective~\ref{obj:range-time}. To
test insert response time, we preformed 200 insertions of random keys that were
not in the dataset. For each of them, we measured the response time and then
deleted the just inserted key to keep the dataset size constant. 
Operations are executed without any delay in between. Results are
shown in Figures~\ref{fig:pos_ins} and~\ref{fig:my_ins}, for platform SSD, and
in Figures~\ref{fig:pos_ins_SUMHDD} and~\ref{fig:my_ins_SUMHDD}, for platform
HDD. 
In these figures, the x-axis shows the size of the number
of key-value pairs in the database and the y-axis shows the response time in
milliseconds.

Analogously, we measured response time for deletion. We preformed 200 deletions
of random keys that were present in the dataset. For each of them, we measured
the response time and then re-inserted the just deleted key-value pair to keep
the dataset size constant. Results are shown in Figures~\ref{fig:pos_del}
and~\ref{fig:my_del}, for platform SSD, and in Figures~\ref{fig:pos_del_SUMHDD}
and~\ref{fig:my_del_SUMHDD}, for platform HDD. Axes are as in the
figures for insert tests. Operations are executed without any delay in between for
all cases, but for the experiment shown in Figure~\ref{fig:my_del_SUMHDD}. In
fact, in this case, with no delay, we observed an anomalous ramp-up trend, which
we think was due to the exceeding of the maximum throughput of the system. In
this case, we performed the test waiting 400ms after each operation.

As expected, charts show that, adopting DB-trees, insert and delete operations
are more costly. 
For plain DBMS tests, the response time looks
essentially constant. This is likely due to the fact that changes are just
written in a log before acknowledging the operation to the client, and
operations are all equal in size. On the contrary, for DB-trees, we can see a
slow but clearly increasing trend, which conforms to the expected logarithmic
trend predicted by the theory. In fact, also in this case operations are just
stored in a log, but the number of actual operations requested to the DBMS is
logarithmic.
For our experiments, when using DB-trees, the slow-down factor for both insert and
delete operations is within 2-4, for PostgreSQL, and within 10-20, for MySQL. 
These factors are quite independent from the kind of platform adopted (i.e., SSD vs. HDD).


About the comparison of charts for SSD and HDD platforms, the same remarks we
made for Objective~\ref{obj:range-time} applies.

%

%


\section{Supporting Group-By Range Queries}\label{sec:groupBy}

In this section, we address a common need that is slightly more complex than an aggregate range query. We deal with the aggregation of values of a certain column performed on the basis of distinct values of a second column limited to a certain range of a third one.
This need is usually fulfilled using the SQL construct GROUP-BY and we refer to this kind of queries as \emph{group-by range queries}. 
\begin{figure}
	\begin{verbatim}
	SELECT Department, SUM(Amount)
	FROM Sale
	WHERE '2010-02-01' <= Date 
	       AND Date <= '2010-03-15'
	GROUP BY Department
	\end{verbatim}
	\caption{Example of GROUP-BY range query.}
	\label{fig:SQLgroupby}
\end{figure}
For example, suppose to have a table that represents the sales of goods performed by each department of a company for each day, and we want to  know the total sales in a specific period for each department.
Suppose the table is called Sale and has columns Department, Date, and Amount. 
Figure~\ref{fig:SQLgroupby} shows an example of group-by range query to obtain this result, 
for an arbitrary range, expressed in plain SQL.


To exploit DB-trees with this purpose, we define the key of the DB-tree to be the pair (Department, Date),
where Department is the most significant part. 
A possible approach to perform the query, is to execute a plain query to obtain the distinct departments $d_1, d_2, \dots$ and then execute, in parallel, Algorithm~\ref{algo:aggregate-range-query} for each department 
with ranges  ($d_1$, start\_range)-($d_1$, end\_range), ($d_2$, start\_range)-($d_2$, end\_range), etc.
With this approach, we perform two query rounds. In the second round, we perform a number of independent queries (one for each department) and
each of them is independently optimized by the query planner.

\begin{algorithm}
	\caption{This algorithm performs a group-by range query on a DB-tree and its associated regular data table. }
	\label{algo:aggregate-range-query-GroupBy}
	
	\begin{algorithmic}[1] 
		
		\Require A regular data table $D$ with columns $x$, $y$ and $v$, an associated DB-tree $T$ on all keys $(x,y)$ with associated values $v$, and two values $y'$ and $y''$ (with $y' < y''$).
		
		\Ensure A set of pairs in the form $(x, \alpha_x)$, such that $\alpha_x=\alpha(v_1,\dots, v_{r_x})$, where $r_x$ is the number of rows of $D$, selected such that they have the given $x$ and have $y' \leq y \leq y''$, and 
		$v_1,\dots, v_{r_x}$ are the values for $v$ for those rows of $D$ (and corresponding key-value pairs in $T$).
				
		\LineComment Execute Lines~\ref{line:groupby:L_tot}-\ref{line:groupby:bar-n_tot} in one query round.

		\State \label{line:groupby:L_tot}$\bar L \gets $ nodes $n$ from $T$ such that 
		$y' < n.\mathrm{max}.y \leq y''$ AND $(n.\mathrm{min}.y < y'$ OR $n.\mathrm{min}.x \neq n.\mathrm{max}.x) $.

		\State \label{line:groupby:R_tot}$\bar R \gets $ nodes $n$ from $T$ such that 
		$y' \leq n.\mathrm{min}.y < y''$ AND $(y'' < n.\mathrm{max}.y $ OR $n.\mathrm{min}.x \neq n.\mathrm{max}.x)$. 
		
%
		\State Let $X$ denotes the distinct values of $x$ in $D$.

		\State \label{line:groupby:bar-n_tot}$\bar N \gets $ pairs $ (x, n) $ such that %
		$x\in X$, $n$ is in $T$, %
		$ n.\mathrm{min} < (x, y') < (x, y'') < n_x.\mathrm{max} $, and $n.\mbox{level} $ is minimum.



		\State For all $x \in X$
		\State \label{line:groupby:nx}\qquad let $n_x$ be the node associated with $x$ in $\bar N$,  
		\State \label{line:groupby:Lx}\qquad let $ L_x$ be sequence of nodes $n$ in $ L$ with $n.\mathrm{max}.x=x$,
		\State \label{line:groupby:Rx}\qquad let $ R_x$ be sequence of nodes $n$ in $ R$ with $n.\mathrm{min}.x=x$.

		
		\State Let $S$ be an empty set.
		\ForAll {triples $\langle n_x,  L_x,  R_x \rangle$ }
		\State Execute Algorithm~\ref{algo:aggregate-range-query} 
		starting from Line~\ref{line:query:computation-begin} 
		with $\bar n \gets  n_x$, $L \gets L_x$, $R \gets R_x$ and 	
		with $ k' \gets (x, y') $ and $k'' \gets  (x, y'') $. 
		In executing Algorithm~\ref{algo:aggregate-range-query}, modify the behavior of Lines~\ref{line:query:L-groupby-changed},~\ref{line:query:R-groupby-changed}, and~\ref{line:query:n-groupby-changed} so that aggregate values next to a key 
		that have the most significant part $\neq x$ should be ignored (since they are 
		not related to group for $x$).
		\State  Let $\alpha_x$ be the result of the above call to Algorithm~\ref{algo:aggregate-range-query}.
		\State Add $(x, \alpha_x)$ to $S$.
		\EndFor
		
		\State \Return $S$
	\end{algorithmic}
	
\end{algorithm}

We now show that, adopting DB-trees, we can perform the above query in one query
round using only three distinct queries as in
Algorithm~\ref{algo:aggregate-range-query}. The procedure is shown in
Algorithm~\ref{algo:aggregate-range-query-GroupBy}, which is a variation of
Algorithm~\ref{algo:aggregate-range-query}. We assume a setting with a regular
data table $D$, containing the data, and additional overlay-table $T$ containing
a DB-tree, also denoted by $T$ and coherent with $D$ (see
Section~\ref{ssec:architectural-aspects}). We assume that $D$ has columns
$x,y,v$ and that the user intends to perform aggregation on column $v$, grouping
on column $x$, while selecting a range on column $y$. For simplicity, in the
following, we use symbols $x$, $y$ and $v$ to denote both column names and the
corresponding generic values of the columns. To build the DB-tree $T$, we consider all the
triples $(x,y,v)$ taken from the rows of $D$ considering each pair $(x,y)$ as
key and $v$ as its corresponding value. We intend $x$ to be the most significant
part of the key, regarding ordering in $T$. 

The first objective of Algorithm~\ref{algo:aggregate-range-query-GroupBy} is to
obtain from $T$, node $n_x$, and sequences $L_x$ and $R_x$, for each distinct
$x$ in $D$. These play the same role that $\bar n$, $L$ and $R$ play in
Algorithm~\ref{algo:aggregate-range-query}. Then, the part of
Algorithm~\ref{algo:aggregate-range-query} that can run in main memory is
performed on each triple $\langle n_x, L_x, R_x\rangle$ to obtain each row of
the result. Lines~\ref{line:groupby:L_tot}-\ref{line:groupby:bar-n_tot} retrieve
all the data needed in one round. Then, we proceed with computations performed
in main memory. Nodes in $\bar N$, $\bar L$, and $\bar R$ are sorted out into
their respective $n_x$, $L_x$ and $R_x$, in
Lines~\ref{line:groupby:nx}-\ref{line:groupby:Rx}. Finally, we execute a
procedure very similar to that of Algorithm~\ref{algo:aggregate-range-query},
for each group, i.e., for each value of $x$. The only notable change is to
additionally check that aggregations take into account only aggregate values
between keys having $x$ as most significant part. In fact, if this is not true,
the aggregate value is not related to the current group (or at least not
completely) and should be ruled out. 

\begin{figure*}
	\centering
	\includegraphics[width=\textwidth]{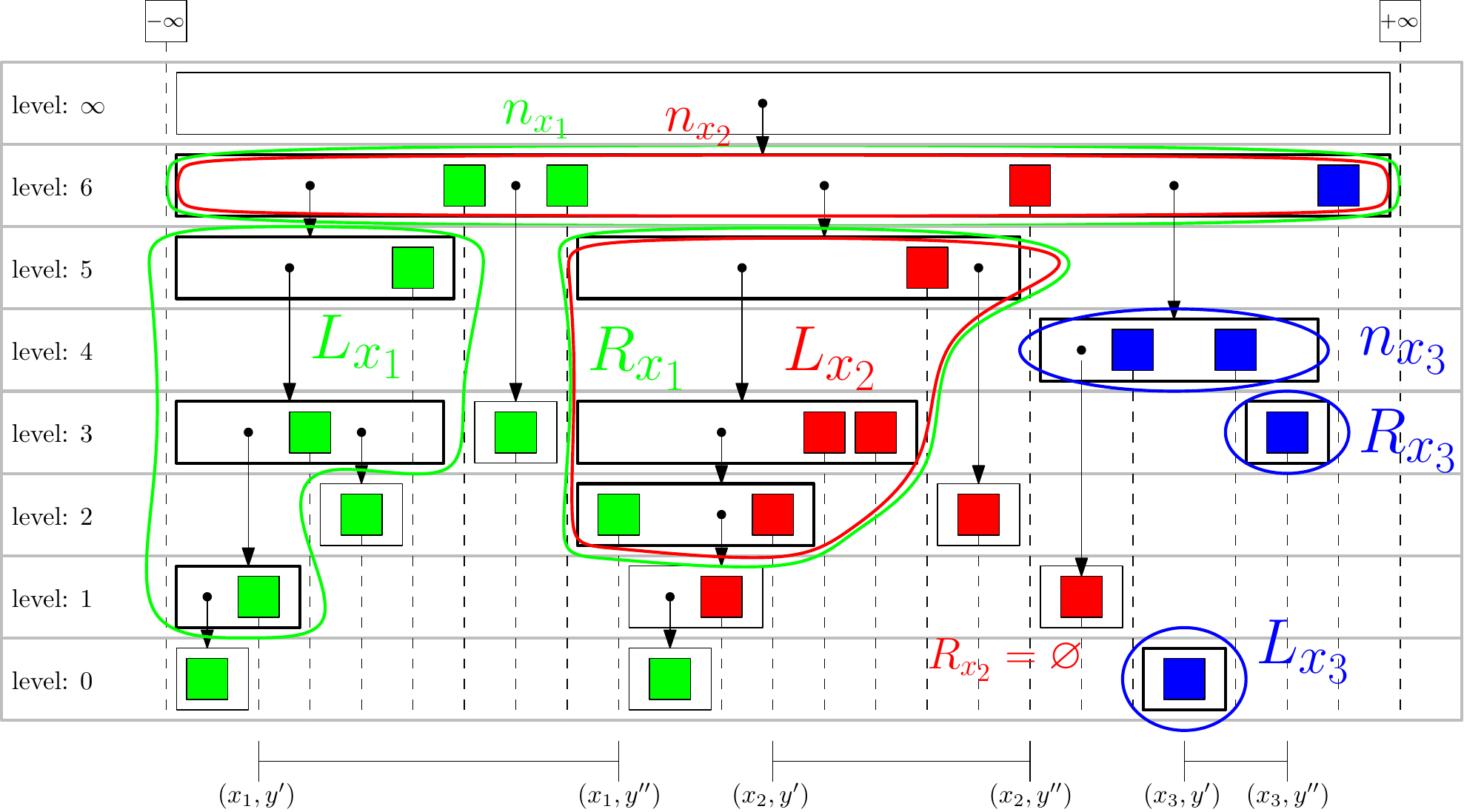}
	\caption{An example of DB-tree evidencing sets $L_x$, $R_x$, and $n_x$
	for a group-by range query executed by Algorithm~\ref{algo:aggregate-range-query-GroupBy}. The details are
    explained in the text.} 
    \label{fig:DB-Tree_groups}
\end{figure*}

\begin{figure*}
	\begin{lstlisting}
	SELECT DISTINCT ON (S.x) x, $T$.*
	FROM $T$, (SELECT DISTINCT x FROM $D$) AS S
	WHERE ((min_x = S.x and min_y < $y'$) OR min_x < S.x) AND
	((max_x = S.x and max_y > $y''$) OR max_x > S.x)
	ORDER BY S.x, level
	\end{lstlisting} 
	\caption{Example of SQL query (in PostgreSQL dialect) to obtain $ \bar N $ in Algorithm~\ref{algo:aggregate-range-query-GroupBy}.}
	\label{fig:SQLbarN}
\end{figure*}

Figure~\ref{fig:DB-Tree_groups} shows an example of DB-Tree on which a group-by
range query is performed. Squares of distinct colors correspond to keys with distinct
values of $x$ (green for $x_1$, red for $x_2$, and blue for $x_3$). The range on
$y$ is given by $ y' $ and $ y'' $. Lines on the bottom of the figure represent
ranges for each value of $x$. Colored contours show, for each value of $x$, the
nodes in $n_x$, $L_x$, and $R_x$. Note that, nodes related to distinct groups
(i.e., for distinct values of $x$) may overlap.

\begin{figure*}
	\centering
	\includegraphics[width=0.9\linewidth]{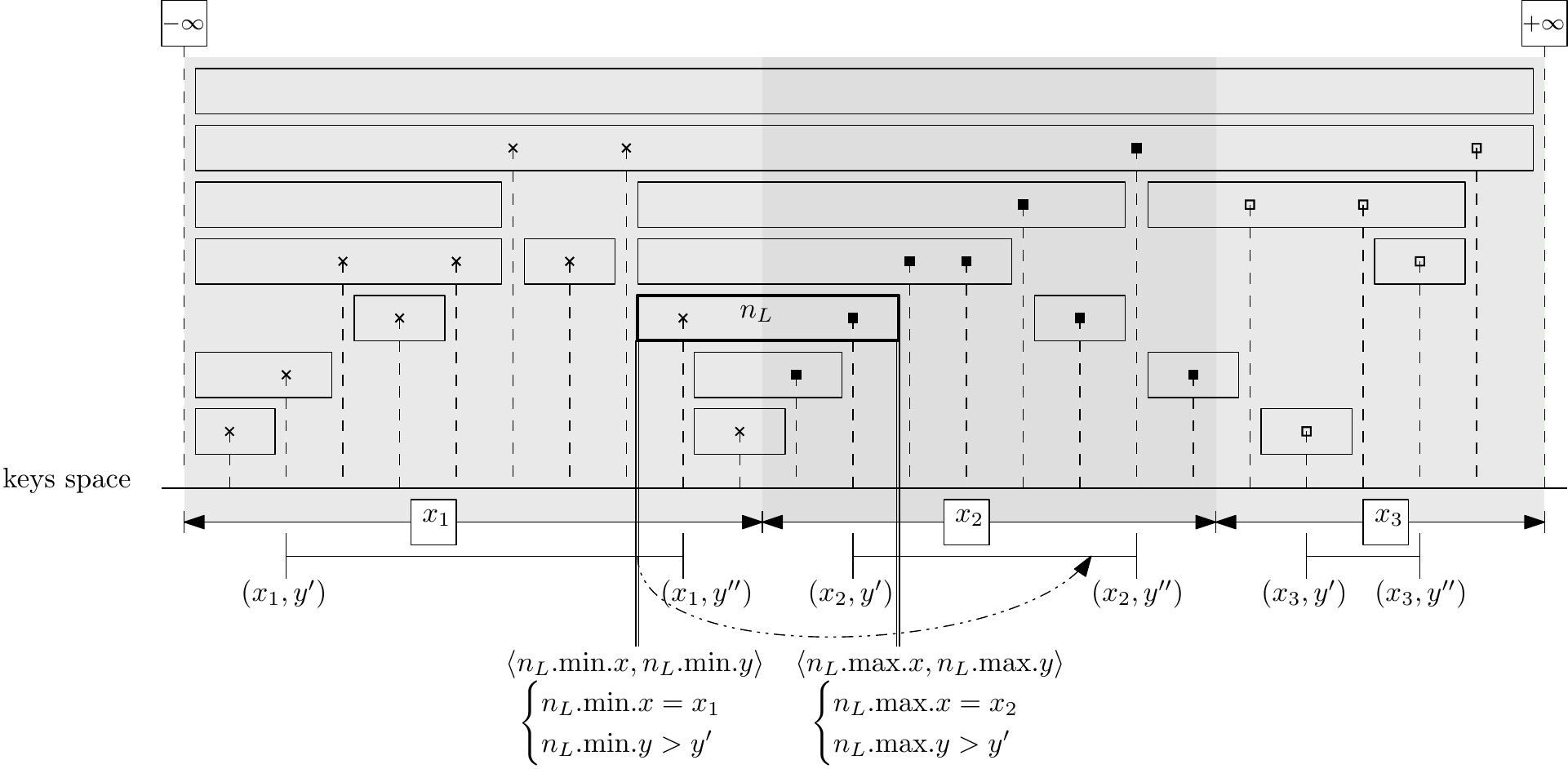} \infloatrnote{C2.1} 
	\caption{A case handled by Algorithm~\ref{algo:aggregate-range-query-GroupBy} in a special way. Node $ n_L $ spans more than one value for $ x $. See text for details.}
	\label{fig:groupbymissing}
\end{figure*}

It is useful to make some remarks on
Lines~\ref{line:groupby:L_tot}-\ref{line:groupby:bar-n_tot}.
Line~\ref{line:groupby:L_tot} introduces a further condition with respect to
what we find in Algorithm~\ref{algo:aggregate-range-query}. In fact, by requiring 
only $n.\mbox{min}.y < y' < n.\mbox{max}.y \leq y''$, as in
Algorithm~\ref{algo:aggregate-range-query}, we would have collected many nodes whose ranges
intersect $(x, y')$ for any $x$, but not all of them. In particular, if a node
has its range that spans more than one value of $x$, it may occur that $
n.\mbox{min}.y < y'$ is not true, but the node intersect a certain $(x, y')$.
The condition $n.\mathrm{min}.x \neq n.\mathrm{max}.x$ includes all the missing
cases.  Symmetrically, the same holds for Line~\ref{line:groupby:R_tot}. 
\rnote{C2.1}Figure~\ref{fig:groupbymissing} shows an example of this case. In this figure, 
node $n_L$ is required to compute aggregate for the group between 
$(x_2,y')$ and $(x_2,y'')$. However, since
$n_L.\mbox{min}.y > y'$, $n_L$ turns out not to be selected when only the 
condition $n.\mbox{min}.y < y' < n.\mbox{max}.y$, of 
Algorithm~\ref{algo:aggregate-range-query}, is considered. On the contrary, $n_L$ is selected by 
the condition $n_L.\mbox{min}.x \neq n_L.\mbox{max}.x$ considered in Algorithm~\ref{algo:aggregate-range-query-GroupBy}, since $n_L.\mbox{min}.x= x_1 \neq x_2 = n_L.\mbox{max}.x$.

About
Line~\ref{line:groupby:bar-n_tot}, performing this operation with \emph{one} SQL
query is not obvious. In Figure~\ref{fig:SQLbarN}, we show an example of this
query using the PostgreSQL dialect. %
The \verb|DISTINCT ON (S.x)| clause returns only rows with distinct values for
\verb|S.x|. The chosen one is the first according to the \verb|ORDER BY|
semantic. We use the regular data table $D$ to obtain all values for $x$. Other
approaches are possible to get all values of $x$ when only $T$ is present. We do
not go into the details of the schemes that can be used for $T$. We just note
that, for what described before this section, there is no reason to represent
keys contained in nodes in a way that can be easily extracted by a query. This
is also not trivial to do in a relational database, since their number for each
node can vary (see Section~\ref{sec:exepriments}). This is the main reason why we described our solution assuming to have a regular DBMS table $D$ as an
input to Algorithm~\ref{algo:aggregate-range-query-GroupBy}.

The correctness of this algorithm derives form the fact that, by construction,
ranges for each $x$ do not overlap (see Figure~\ref{fig:DB-Tree_groups}) and by
the correctness of Algorithm~\ref{algo:aggregate-range-query}. Further, we note
that, elements of aggregate sequences of retrieved nodes that aggregate values
for keys with different value of $x$ are never used by
Algorithm~\ref{algo:aggregate-range-query-GroupBy}. 

\subsection{Experiments with DB-Trees for Group-By Range Queries}\label{ssec:groupby-comparison-plaindb} \rnote{C2.3}
We performed some experiments to assess the performance of 
Algortihm~\ref{algo:aggregate-range-query-GroupBy} on realistic data. In
Section~\ref{sec:exepriments}, we noted that DB-trees are more advantageous for
aggregate range queries with large ranges. As we show in the following, this is
also true for group-by range queries. The dataset we used for our experiments is
derived from the TPC-H benchmark~\cite{tpc-h}. From that dataset, we considered
the \texttt{lineitem} table containing about six millions of rows. \rnote{C4.2} We picked
columns \texttt{L\_suppkey} (numeric IDs), \texttt{L\_shipdate} (dates that span
seven years), and \texttt{L\_extendedprice} (floating point numbers). We focused on a query that
aggregates the values of \texttt{L\_extendedprice} grouping by distinct values
of \texttt{L\_suppkey} on ranges defined on \texttt{L\_shipdate}. Other columns
were not imported into our test database. To support this group-by range query,
we had to choose, as key of the DB-tree, the pair (\texttt{L\_suppkey},
\texttt{L\_shipdate}). For this pair, TPC-H contains duplicated values, which is not
compatible with our prototypical implementation of DB-trees. To circumvent this
problem, we transformed each date to a timestamp adding a random time. To show
performances of DB-trees in a range where they can provide a substantial
benefit, we modified values in \texttt{L\_suppkey} as follows. The original
dataset contains 10,000 distinct values in \texttt{L\_suppkey}. We replaced each value $x$ in
\texttt{L\_suppkey} with $x \mod 20$ to obtain 20 distinct values. In this way,
we obtained 20 groups, each one containing a number of elements that is large enough to 
show the effectiveness of the adoption of DB-trees. 
We refer to the resulting dataset as \emph{modified TPC-H}.
Since elements are quite uniformly distributed over the
groups, each group turns out to contain about 300,000 elements, for the whole
dataset. The application of range selection reduces it. This reduction is quite
regular, since values in \texttt{L\_shipdate} are uniformly distributed over
groups and over time.     We performed our experiments with PostgreSQL (SSD platform). During
the preparation of the experiments, we realized that PostgreSQL is not able to
perform a thorough optimization of the query shown in Figure~\ref{fig:SQLbarN}. We substituted it with
that shown in Figure~\ref{fig:SQLbarN-optimized}. We also added a column \texttt{f}
to the overlay table $T$ to contain character \texttt{'t'} if the result of \texttt{min\_x <>
	max\_x} for the node is true, \texttt{'f'} otherwise.
\begin{figure}
	\begin{lstlisting}
(SELECT * FROM $T$ 
 WHERE ((min_y < $y'$)  AND (max_y > $y''$ )))
UNION 
(Select * FROM $T$ WHERE f = 't')
	\end{lstlisting} 
	\caption{Optimized version of the query shown in Figure~\ref{fig:SQLbarN} that was adopted 
	for testing group-by range queries. See text for further details.}
	\label{fig:SQLbarN-optimized}
\end{figure}
Actually, this query returns a larger set of nodes with respect to the query
shown in Figure~\ref{fig:SQLbarN}. In fact, the new query always returns all nodes
whose range overlaps group boundary (i.e., all nodes with \texttt{f = 't'}). To obtain
$\bar N$ (see Algorithm~\ref{algo:aggregate-range-query-GroupBy}), a further
selection is performed in memory by the overlay logic software when data are received. \rnote{C3.9}On table $T$, we
configured plain B-tree indexes on \texttt{f} and on (\texttt{min\_y},
\texttt{max\_y}). We also measured execution times for the plain DBMS query on
the regular table, where we configured a primary key on (\texttt{L\_suppkey},
\texttt{L\_shipdate}). \rnote{C3.9} In PostgreSQL, this means that a B-tree index kept on that columns. Since dates in \texttt{L\_shipdate} are uniformly
distributed over time, it was easy to pick random ranges such that each group
contained a number of elements from 50,000 to 300,000, with steps of 50,000. For each step, we
queried 200 different random ranges (i.e., each with a random shift of the range) and we took the average of the
execution times. The results, shown in Figure~\ref{fig:groupbyrange}, confirms that
DB-trees outperform the plain DBMS query on the regular table for ranges that
give rise to groups with large number of elements: greater that 200,000 in our
case.

\begin{figure}
	\centering
	\includegraphics[width=1\linewidth]{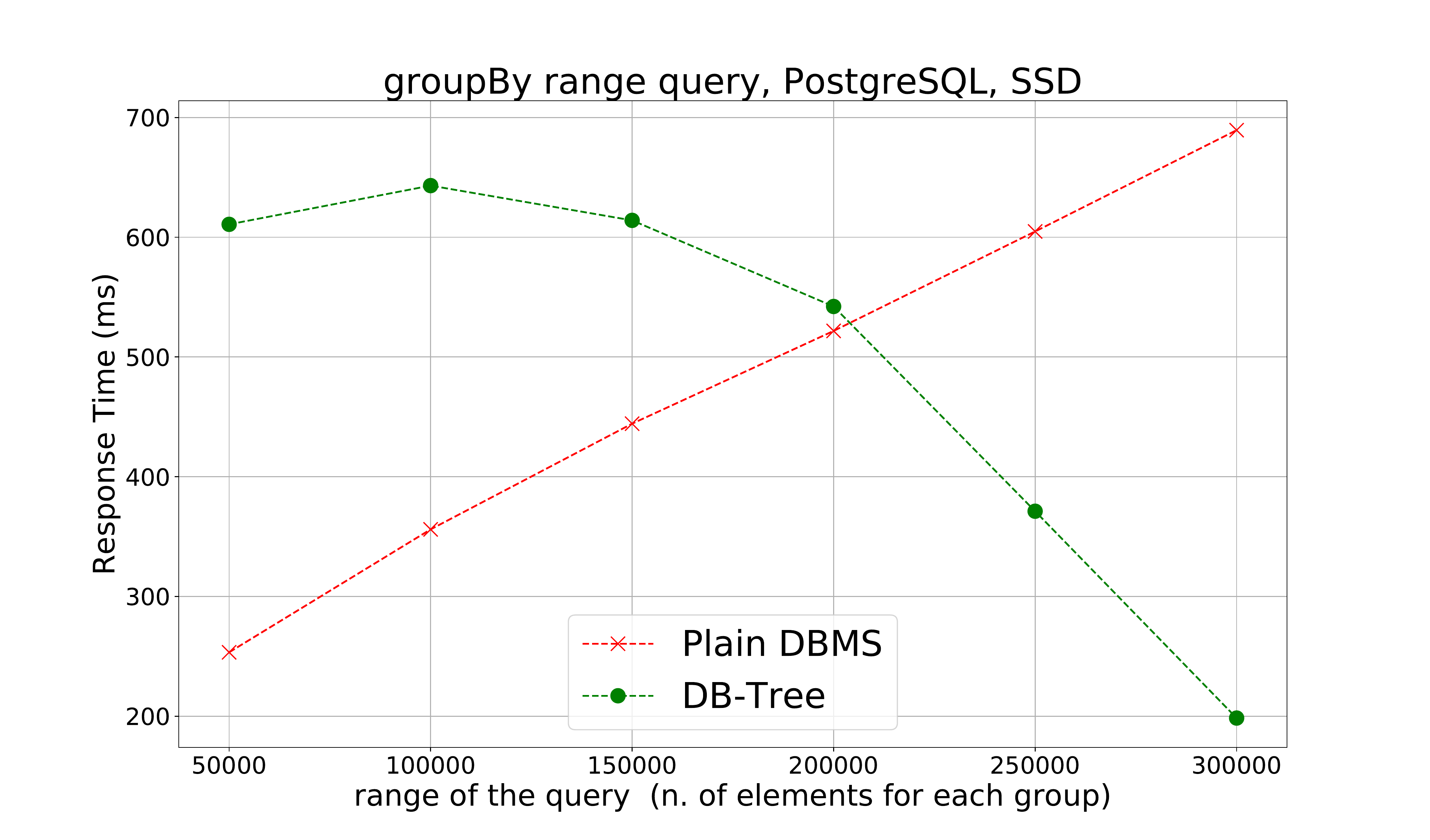}
	\caption{Performances of the execution of group-by range queries on PostgreSQL (SSD platform), on the modified TPC-H dataset.}
	\label{fig:groupbyrange}
\end{figure}

\subsection{Comparison with Materialized Views}\label{ssec:groupby-comparison-materialized-views}

\begin{figure}
	\centering
	\includegraphics[width=1\linewidth]{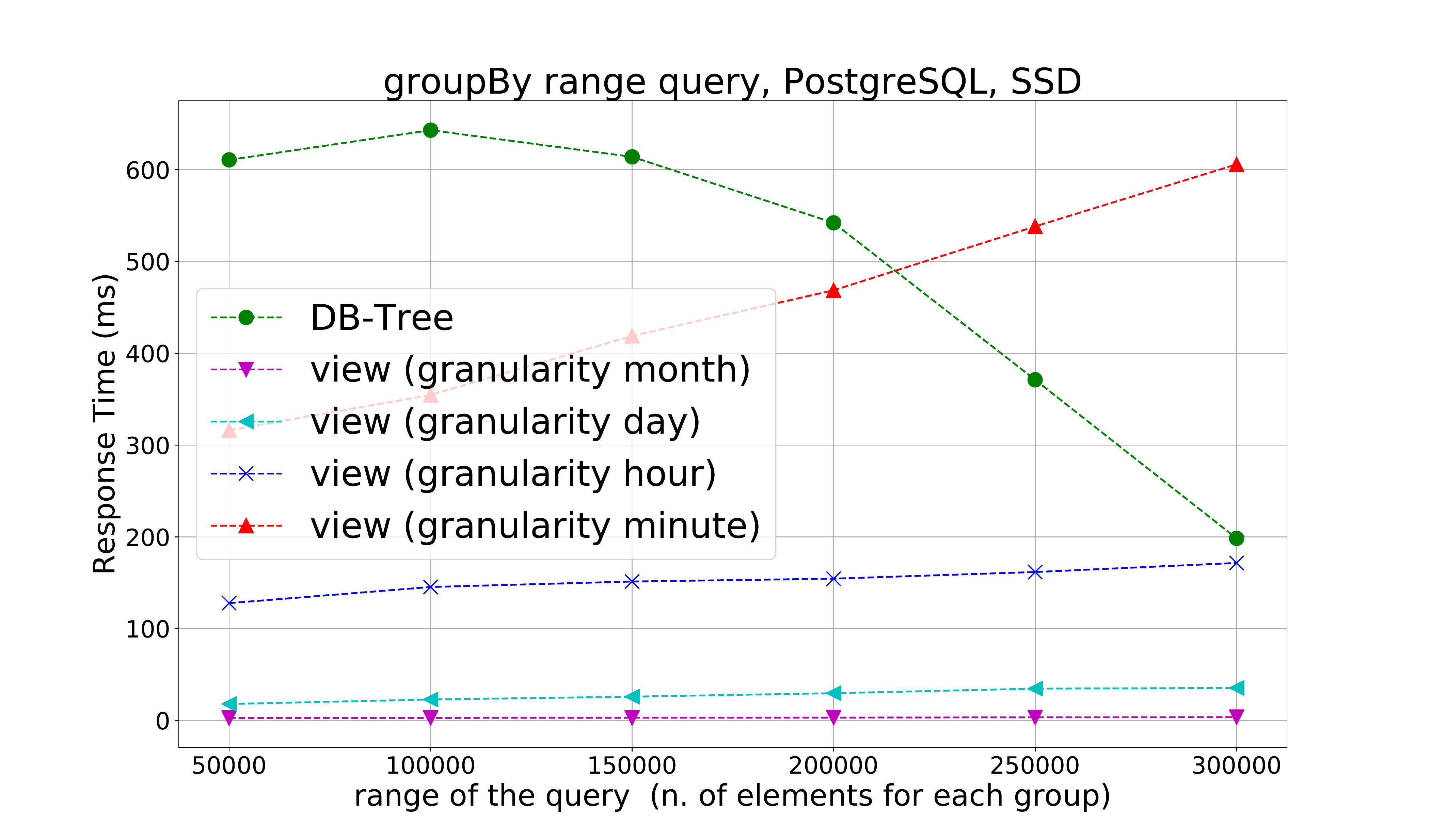}
	\caption{Performances of group-by range queries. Materialized views vs. DB-Trees on PostgreSQL (SSD platform), on the modified TPC-H dataset.}
	\label{fig:matViews}
\end{figure}

\begin{table}
	\centering
\begin{tabular}{|c|c|c|c|c|}
	\hline 
	Granularity & Refresh time & N. Rows\\ 
	\hline 
	Month & 5s & 1672 \\ 
	\hline 
	Day &  6s & 50,488 \\ 
	\hline 
	Hour &  8s & 1,145,760 \\ 
	\hline 
	Minute &  12s & 5,683,522 \\ 
	\hline 
\end{tabular} 
\caption{Refresh times and size of materialized views used in our experiments.}
\label{table:materialized views}
\end{table}
\rnote{C3.1 C2.3}We run the same tests for group-by range queries using materialized views, to
understand how this common technique compare with DB-trees. We performed our
experiments using PostgreSQL. For the same dataset adopted in the above
experiments, we created materialized views at different granularities: month,
day, hour, and minute. We executed the same group-by range queries described
above on these materialized views and we measured execution times. We picked 
ranges giving a number of elements per group between 50,000 and 300,000, with steps of 50,000 and 200 random queries, as above, for each step. 
Figure~\ref{fig:matViews} shows the execution times taken by the queries performed on 
the materialized views vs. execution times of the same queries performed using DB-trees.
For materialized views, the execution time increases linearly with the number of elements
contained in each group, while it decreases for DB-Trees. However, DB-Trees perform better only for minute granularity, in our tests.  Since the query is performed on
aggregated data, the precision of the result of the queries performed on materialized views depends on the granularity of the
view. It is possible to obtain precise results by independently querying the
extremes of the range but we did not performed any experiment regarding this
aspect. In fact, we think that the above results already show the great
potentiality of materialized views. 
Table~\ref{table:materialized
	views} shows the time taken to refresh the whole materialized view.
 They  are between 5 and 12 seconds for our
dataset. The number of rows of the views are also reported. We note that 
materializing with minute granularity provides very little benefit, since the
original table contains 6,001,215 rows.
 
From the trend shown by our experiments, DB-trees may result to be a better
approach than materialized views when the number of elements in the groups are
very large and this is not compensated by the coarseness of the granularity.
This may occur in practice, since granularities of materialized views are
statically decided in advance, while queries might be on ranges whose size can
vary across several orders of magnitude. For example, this may occur in
graphical systems if the user is allowed to zoom in and out on a timeline and a
corresponding histogram for the current zoom level should be shown. DB-trees
have a substantial overhead but they are adaptive, in the sense that no decision
in advance is needed about the order of magnitude of the ranges to be queried.
Further, we recall, that DB-trees may be the only viable solution in situations
in which the DBMS does not natively support the needed aggregation function.


\subsection{Comparison with VerdictDB}\label{ssec:groupby-comparison-verdictdb}

\begin{figure}
	\centering
	\includegraphics[width=1\linewidth]{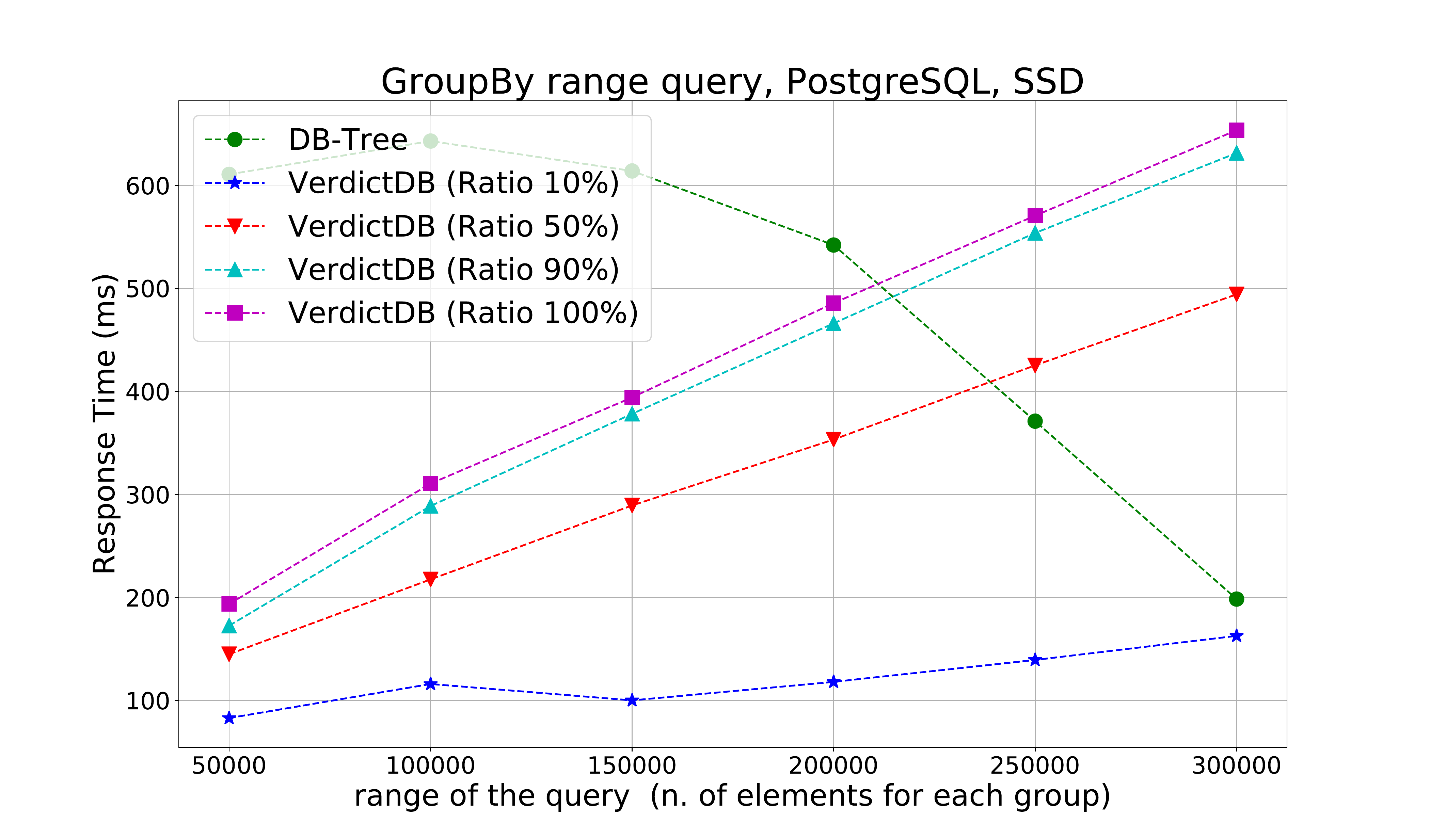}
	\caption{Performances of VerdictDB vs. group-by range queries using DB-Trees on PostgreSQL (SSD platform), on the modified TPC-H dataset.}
	\label{fig:verdict}
\end{figure}

\rnote{C3.1 C2.3} We now compare our DB-tree approach for the execution of group-by range queries  with
the solution provided by VerdictDB~\cite{park2018verdictdb}.
VerdictDB is a very interesting tool that practically realizes an approximate
query processing technique. It allows the user to perform 
group-by range queries choosing a trade-off between speed and precision of the
result. It relies on regular relational DBMS to store data. The user can ask the creation
of a ``sampled'' version of a table, called \emph{scramble}, with a certain
\emph{size ratio}. Higher size ratios provide better precision at the expense of
longer execution time. We performed our experiments using PostgreSQL (SSD platform) as
underlying DBMS. For the same dataset adopted in the above experiments (modified TPC-H), we
created several \emph{scrambles} with different size ratios: 10\%, 50\%,
90\%, and 100\%. We executed the same group-by range queries described above on
all scrambles. We again picked ranges giving a number of elements per group
between 50,000 and 300,000, with steps of 50,000, running 200 different random
queries in each step, as in Sections~\ref{ssec:groupby-comparison-plaindb} and~\ref{ssec:groupby-comparison-materialized-views}.

VerdictDB is very efficient when the size ratio is low. In any case, execution
times increase linearly with the number of elements contained in each group, as
for the execution using the plain DBMS approach. In Figure~\ref{fig:verdict}, we
show the time taken by VerdictDB to process the group-by range query for increasing range sizes
using the four size ratios mentioned above. 
For a clear comparison, we also report the time taken by our approach based on
DB-trees. The comparison is favorable to our approach when the number of
elements in the groups is large and when the size ratio is not too low. Since a
low size ratio negatively impacts the precision of the results produced by VerdictDB, we also
measured errors. In Table~\ref{table:verdictErrors}, we
reported, for  each point of the chart in Figure~\ref{fig:verdict} that is
related to VerdictDB results, the maximum error we obtained, relative to the
correct answer. Each shown percentage is the maximum of the relative errors
across all the queries (200) that we run for each step.
These results are specific for our case and are only useful for the purpose of
comparison with the DB-tree approach in this specific test. A broad description
of VerdictDB with respect to precision can be found in~\cite{park2018verdictdb}.

In our experiments, the time taken by VerdictDB to generate the scramble from
scratch is always about 6 seconds.

\begin{table}
	\centering
	\begin{tabular}{|c|c|c|c|c|c|c|}
		\hline
		\multirow{2}{2cm}{\centering Scramble\\Size Ratio}&\multicolumn{6}{c|}{Number of Elements for Each Group}\\
		\cline{2-7} 
		 & 50K & 100K & 150K & 200K & 250K & 300K \\ 
		\hline  
		10\% & 5.18\% & 3.79\% & 2.87\% & 1.99\% & 1.54\% & 1.45\%  \\ 
		\hline 
		50\% &  2.44\% & 2.53\% & 2.0\% & 1.74\% & 1.56\% & 1.53\% \\ 
		\hline 
		90\% &  3.0\% & 2.0\%7 & 1.89\% & 1.35\% & 1.35\% & 1.26\%  \\ 
		\hline 
		100\% & 3.02\% & 2.02\% & 1.76\% & 1.43\% & 1.18\% & 1.2\% \\ 
		\hline 
	\end{tabular} 
	\caption{Maximum relative errors of results of group-by range queries  using VerdictDB. See text for details.}
\label{table:verdictErrors}
\end{table}

The most evident difference between DB-trees approach and VerdictDB is that
DB-trees always produce an exact result, while results produced by VerdictDB may
have non-negligible errors. It depends on the application how much precision can
be traded for speed. In any case, when the number of elements in the groups are
large, the DB-tree approach outperform VerdictDB even for moderately small size
ratios. For example, in our case, when the number of elements for each group is 250,000 or more, DB-trees provide exact results and perform better even if VerdictDB is used with a size ratio of 50\%, which gives a maximum error of about 1.5\% in our case.


%

\section{Supporting Authenticated Data Structures}\label{sec:ADS}

In this section, we show how it is possible to use DB-trees to
support \emph{Authenticated Data Structures} (\emph{ADS}). We briefly introduce
ADSes with their properties and show a typical use case. Then, we show how we can
use DB-trees as a persistent and efficient ADS.

An Authenticated Data Structure (ADS)  is an ordered container of elements that
deterministically provides a constant-size digest of its content that has the same
properties of a cryptographic hash. We call this digest the \emph{root-hash} of the ADS, and we denote it by $r$. 
For our purposes, we limit ADSes to contain implicitly ordered elements, like key-value pairs.
In other words, we regard
ADSes as search trees augmented with security features.
If the content of the ADS changes, $r$ changes. 
Further, it is hard to find two sets of
elements with the same root-hash.
An ADS provides two operations: \emph{authenticated query}  and \emph{authenticated update}.
A query returns the queried element and the proof, associated with a certain $r$, 
that the result is indeed among the elements of the ADS instance having that $r$ as root-hash. If a trusted entity safely stores the current
$r$, it can query the ADS and execute a cryptographic check of the proof against its
trusted version of $r$ to verify that the query result matches what
expected. The update operation on key $k$ changes $v$ associated with
$k$ into a provided $v'$ and changes $r$ in $r'$, as well. Insertion and deletion 
also change $r$.
The interesting aspect is
that a trusted entity that intends to update the ADS should be able to autonomously
compute  $r'$ starting from the proof of the elements that are changing.

\begin{figure}
	\centering
	\includegraphics[width=\linewidth]{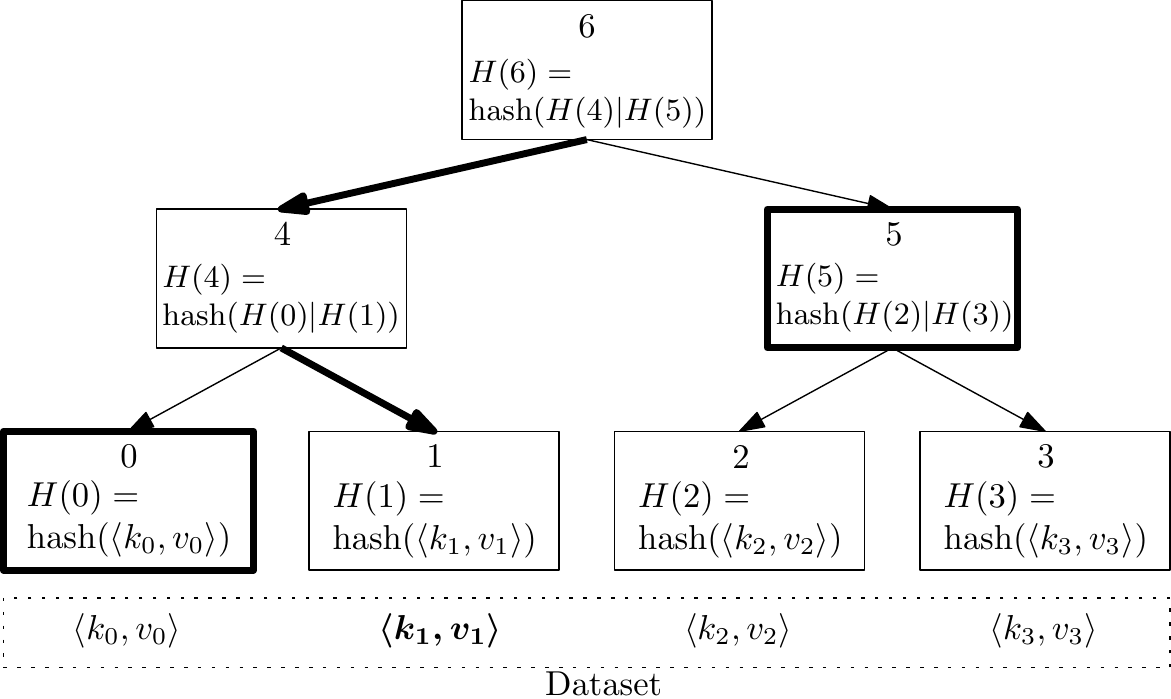}
	\caption{An example of Merkle Hash Tree with four leaves 
    and a binary structure. We evidenced the elements regarding the proof of $\langle k_1, v_1\rangle$ with thick contours.}
	\label{fig:MHT_example}
\end{figure}

A simple ADS is the Merkle Hash Tree~\cite{merkle1987digital} (\emph{MHT}). An
example of MHT is shown in Figure~\ref{fig:MHT_example}. In our example, the MHT
is a binary search tree in which every leaf is associated with a
key-value pair $\left<k,v\right>$. Function $\hash{\cdot}$ is a cryptographic
hash function. Every node $n$ is labeled by a cryptographic hash $H(n)$. If $n$ 
is a leaf, we define $ H(n) = \hash{\left<k,v\right>}$. If $n$ is an internal
node, with $n'$ and $n''$ its children, $ H(n) = \hash{H(n')|H(n'')}$. If $n$ is
the root, $r=H(n)$ is the root-hash of the MHT.

A typical application of MHT is to allow a client with limited amount of
resources to outsource the storage of a large amount of data to an untrusted
server. The client keeps only a trusted version of the root-hash while the
server keeps the MHT. The server provides proofs for each query performed by the
client. A proof for a leaf is obtained by considering the path from the leaf to
the root and, for each node of the path, putting the hash of the sibling into
the proof (see Figure~\ref{fig:MHT_example} for an example). This allows the receiver of the proof to be able to compute the
root-hash from the content of the leaf, i.e., from the result of the query. The
client can compare the resulting root-hash against its trusted copy to verify
authenticity of the reply. The length of the proofs is $O(\log n)$ for balanced
trees (where $n$ is the number of leaves).

Suppose that the client intends to change the value associated with a certain key
$k$. Client query the server for the $k$ and get the current associated $v$ with
its proof. After checking the proof against its root-hash, client can compute
the new root-hash for a new value $v'$ just performing the same computation as
for proof checking but pretending the new value $v'$ is associated with $k$. The
obtained root-hash should be the root-hash of the updated MHT, and be used as
trusted root-hash for subsequent queries.

A DB-tree can be adapted to serve as a persisted and efficient ADS. In this
case, we call it an \emph{authenticated DB-tree}. From the point of view of the
algorithms, we use $\hash{a|b}$ as aggregation function on $a$ and $b$, and we consider
cryptographic hash values as a sort of aggregate values (details are provided
below). Note that, $\hash{a|b}$ is usually non-associative, since standard hash
functions (like, e.g., SHA-2) are not associative\footnote{An associative
	cryptographic hash function was proposed by Tillich and
	Z\'emor~\cite{tillich1994hashing} in 1994. It resisted cryptanalysis attempts
	untill 2011~\cite{grassl2011cryptanalysis}, but now it is considered insecure.
	However, since these functions have many interesting properties, the research is
	still active in this area (see, for example,~\cite{bromberg2017navigating}) and
	in the future we might have cryptographic hash functions both associative and
	secure, to use with regular DB-trees.}. In this case, the concept of range query is not correctly defined, and
Algorithm~\ref{algo:aggregate-range-query} is not useful. On the contrary, the
new \emph{authenticated query} operation is of interest in this context, which
allows the user not only to get the value of a key but also to get the
corresponding proof. Further, without associativity, the structure of the tree
impacts the resulting root-hash. This is not desirable, since users would like
to uniquely associate a root-hash to a given content of the ADS. To eliminate
this dependency, we can ``de-randomize'' on the basis of the content as follows.
In running Algorithm~\ref{algo:random-level}, to choose the level associated with
a key $k$, we suppose to adopt a random generator that is initialized by a seed
that we set to $k$ itself\footnote{A more practical approach is
	to compute $\hash{k}$ and take, from the result, the position of the first bit
	that is zero.}. In this way, given a key $k$, a level is deterministically
associated with $k$, while keeping unchanged the statistics of the levels and
the properties described in Section~\ref{ssec:db-tree-formal}.

In an authenticated DB-tree, the hash of a node $n$ with $n.\mbox{aseq}=a_0,
p_1, a_1, \dots, a_{m-1}, p_m, a_m$ is $\hash{n}=\hash{a_0 | p_1 |a_1 | \dots | a_{m-1}|
	p_m| a_m}$, where each $a_i=\hash{n_i}$ and $n_i$ is the corresponding child. As
we did in the rest of the paper, we intend that missed $a_i$'s are simply
omitted in the formulas. If $n$ is a leaf, it has no children, hence, its hash is
computed on the concatenation of the pairs it contains, unnoticed.
The root of an authenticated DB-tree contains only one
hash, that is its \emph{root-hash}. A \emph{proof}, for a given pair $p=\langle
k, v\rangle$, is the sequence $N$ of nodes that have $k$ in their range, ordered
by ascending level, up to the root. This set of nodes can be retrieved in one query round and has size logarithmic in the number of elements contained in the DB-tree (see Section~\ref{sec:algorithms}).

 The client checks that the response to an
authenticated query operation is genuine, by the following procedure.
It computes $\hash{n}$ for the first node $n\in N$.
For each node $n\in N$ after the first, let $n'$ be the node that precedes $n$ in $N$.
Node $n$ contains in $n.$aseq 
one hash that is related to the node $n'$. This is the hash between the keys that are the closest to $k$. This hash is ignored and substituted by 
$\hash{n'}$. Then, $\hash{n}$ is computed and used in the next iteration. 
When all nodes are scanned, the hash of the root is compared with the trusted root-hash
the client should have.
A proof of non-existence for key $k$ has the very same structure, just the first node of $N$ does 
not contain a key-value pair of $k$ and the keys closest to $k$ in $n.$aseq have no aggregate value between them. The same checking procedure can be applied.

In an authenticated DB-tree, the algorithms to perform update, insertion and deletion of a key-value
pair are very similar to those of the corresponding operations for a regular DB-tree. We
should take care of the following aspects. When the server provides a set of
nodes they should always be considered proofs and checked against the trusted
root-hash of the client before proceeding. This is always possible since all
DB-tree-changing algorithms deal with paths up to the root of the DB-tree. After
the change, hashes, including the root-hash, are updated by the execution of
Algorithm~\ref{algo:update-aggrvalue-up}. After executing any changing
algorithm, the client should locally store the new root-hash as the trusted
root-hash to be used in the following operations.

Concerning security, we assume a threat model in which the DBMS can perform any
tampering on the DB-tree with the purpose to change the key-value pairs it
contains. We do not consider other kinds of attacks. The \emph{security} of an
authenticated DB-tree is its ability to always detect any misbehavior of the
DBMS, i.e., we would like to rule out any \emph{false negative}. As for Merkle
Hash Trees, this is a direct consequence of the inability of the untrusted DBMS
to find a collision on the adopted cryptographic hash function. In the context
of ADSes, \emph{correctness} means that any misbehavior detected by proof
checks is a real DBMS misbehavior, i.e., we would like to rule out any
\emph{false positive}. Correctness of DB-trees derives directly from the ability
of keeping invariants after all DB-tree-changing operations. This comprises correctly
computing the hashes of all nodes, which is responsibility of
Algorithm~\ref{algo:update-aggrvalue-up}.

\rnote{C2.3}%
Concerning performances of authenticated DB-trees in practice, we note that all
their operations interact with the DBMS using exactly the same queries as for
plain DB-trees. Only the local processing performed by the overlay-logic is
slightly different. In all experiments of Section~\ref{sec:exepriments}, the
time spent for the overlay-logic processing is negligible with respect to the
time spent to execute the queries. For authenticated DB-trees, the overlay-logic
may additionally need to check proofs performing all needed cryptographic
hashes. But this turns out to be negligible, too, as we show in the following.
For these reasons, in practice, authenticated DB-trees perform like regular
DB-trees. In particular, for insertion and deletion, experimental results are
essentially the same of those shown in Section~\ref{sec:exepriments}, hence they
are not reported here. An authenticated query for a key $k$ selects only the
nodes $n$ for which $n.\mbox{min} < k <n.\mbox{max}$. We measured the average
execution time of this query on the same DB-tree of 1 million elements that we
used in Section~\ref{sec:exepriments}. The time it takes on PostgreSQL (using
SSD) is about 20ms (average on 200 random queries). Our instance has 23 levels,
hence each query returns at most 23 nodes, which should be interpreted as the
proof returned by the authenticated query performed on the authenticated DB-tree. Now, we
show that the time spent to check a proof is negligible with respect to the time
taken to perform the query on the DBMS. The actual verification of the proof
should compute all the cryptographic hashes along the proof up to the root. In
our instance, each node contains 2 key-value pairs and 3 ``aggregates'', that is hashes of children, on average.
To verify the proof, for each node, we have to compute the hash of the pairs
plus one hash of the whole node. We performed the tests using SHA256. Each hash
takes about $1.5\mu s$ to be computed. Hence, the time taken by proof
verification turns out to be about $103.5\mu s$. This is four order of magnitude
less than the time taken to perform the query on the DBMS.

%


\section{Discussion of the Simplifying Assumptions}\label{sec:arch_discussion}

\rnote{C3.4}In Section~\ref{ssec:architectural-aspects}, we presented general architectural
problems related with overlay indexes and introduced some simplifying
assumptions. The objective of this section is to discuss these assumptions, also
considering the specific DB-tree approach, and show whether generalizations are
simple or require further scientific investigation.

\textbf{Middle layer.} In Section~\ref{ssec:architectural-aspects}, we mentioned
the possibility to develop a query rewriting middle layer to support
overlay-indexes. Query rewriting is a classical topic in database research (see,
for example,~\cite{park2018verdictdb,gupta1995aggregate,acharya1999aqua,
	goldstein2001optimizing,halevy2001answering}), but practical widely used
realizations are rare. Concerning the development of a middle layer targeted to DB-trees, we
identify some challenges. Firstly, query rewriting means dealing with a complex
SQL syntax and its many proprietary variations. We think that the scientific
interest of this aspect is marginal, but the development effort is large.
Further, a middle layer may address only the exact kinds of queries shown in
this paper (passing all others queries to the DBMS unchanged) or trying to
support the optimization of more complex aggregate (group-by) range queries, for
example comprising equality selections, \rnote{C4.3}two dimensional range selections, joins,
etc. Some of these objectives are direct extensions of what is described in this
paper, while some may require further scientific investigation. For example,
suppose to have an aggregate range query $Q$ supported by a DB-tree $T$. To
support $Q'$ derived from $Q$ by including an additional equality selection on a
column $c$, we can simply add $c$ as the most significant part of the key of $T$. \rnote{C4.3 C3.7}On the
contrary, the extension to  more-than-one-dimension range queries is not trivial
and, in our opinion, may deserve further scientific investigation (see also Section~\ref{sec:conclusions}). In any case, the
middle layer should allow the user to specify which DB-trees to build, specifying
the key, the value, and \rnote{C3.6}the aggregation function(s) to be supported. To do
that, an extension of the data definition language should be provided. The
design of this extension is a critical aspect from the point of view of the
usability and of the power of the resulting layer. Further, having a number of
DB-trees at disposal, the middle layer should be able to rewrite certain queries
taking advantage ot them, but only when this is deemed useful. 
This can be regarded as an optimization problem per se that may be independently studied. 
For example, the middle layer may choose not to rewrite an aggregated range
query whose range is small.


\textbf{Transactions, multiple clients, fault tolerance.} In this paper, 
for the sake of simplicity, we deliberately avoided to consider the use
of overlay-indexes and DB-trees in a context where transactions are needed.
Typical situations in which this occurs is when multiple clients access the same
DB-tree and when faults of the DBMS may occur. In principle, nothing
prevents to execute the algorithms presented in this paper within transactions
(supposing to adopt a DBMS that supports them).
Algorithms~\ref{algo:update},~\ref{algo:insert}, and~\ref{algo:delete} follow
the scheme showed at the beginning of Section~\ref{ssec:DBTree:upd-ins-del}.
That is, they wait for the reply of the read round, compute the changes and perform the update round. However, the best practice is to avoid the execution of schemes like this within a
transaction. In fact, a communication problem with the DBMS may keep the transaction
open (and involved tables locked) until a timeout expires. A possible workaround to this problem is to move the overlay-logic within
the DBMS by using stored procedures (if supported) or moving overlay-logic so
close to the DBMS so that network faults cannot independently occur (for example within the same
machine). 

The above considerations also apply to the case in which a DB-tree is
associated with a regular table. In this case,
Algorithms~\ref{algo:update},~\ref{algo:insert}, and~\ref{algo:delete} should be
executed in a transaction together with the change of the regular table.

When changing operations are rare, and the DB-tree is not associated with a
regular table, optimistic approaches may be adopted. For example, we could perform the
read round of a changing operation within a transaction and perform the
update round of that operation in a distinct transaction. In the second transaction, before applying
changes, it should be checked, preferably directly within the DBMS, that no change was
applied (e.g., by another client) in the meantime to the DB-tree. If the check fails, the
transaction should be aborted and the whole changing operation should be re-run, starting
from the query round. An interesting way to perform this check is to use the same DB-tree 
to realize an authenticated data structure (see Section~\ref{sec:ADS}). 
In this case, the root-hash can be checked to verify if any change to the DB-tree occurred from the last read.
It is
possible to support an aggregation function and a cryptographic hash within the
same DB-tree as explained in the following.

\textbf{Multiple aggregation functions.} \rnote{C3.5 C3.7} There are cases where more than one
aggregation function $\langle f_i, g_i, h_i\rangle$, for $i=1,\dots,q$, (see
Section~\ref{ssec:db-tree-formal}) should be supported, where each $f_i(\cdot)$ aggregates
values in $A_i$, and each $g_i(\cdot)$ generates values in $A_i$ from one or more columns. Here, we refer to columns as if 
they were stored independently in the database in a regular data table, but this is not strictly needed (see
Section~\ref{ssec:architectural-aspects}).  
If range selections are always performed on
the same key for all aggregation functions,
we can build a single
DB-tree to support all of them. For this DB-tree, the set $A$ (see Section~\ref{ssec:db-tree-formal}) is
$A_1\times \dots \times A_q$, aggregation function is defined as $f(a_1, \dots, a_q) = (f_1(a_1),\dots,f_q(a_q))$, and functions $g(\cdot)$ and $h(\cdot)$ are consistently defined.
\rnote{C3.7}Note that, if the queries that we have to support perform range
selections on distinct columns, to support them, we have to construct one
DB-tree for each of that columns (having each of them as its key). 
%
%
%
Cryptographic hash functions can also be
supported along with aggregation functions, if we take care of the fact that they are not associative (see the details in Section~\ref{sec:ADS}). 
This allows us to support the
optimistic approach described at the end of the previous paragraph. In this case, checking if a
DB-tree has been changed since the execution of a previous query, boils down to 
checking that its root-hash is unchanged.

\textbf{Caching, consistency, batch insertion.} In
Section~\ref{ssec:architectural-aspects}, we mentioned the possibility to
implement caching in the overlay-logic. We also mentioned that consistency
problems may arise in the case of multiple clients. This occurs when the DB-tree
is changed by Algorithms~\ref{algo:update},~\ref{algo:insert},
and~\ref{algo:delete}. The most obvious workaround is to broadcast all changes
to all clients, so that they can update (or invalidate) their cache. This
approach is very demanding in terms of networking and CPU resources, especially
if changes and clients are many. We may obtain most of the
benefits of caching by moving the overlay-logic close to the DBMS or by using
stored procedures. In fact, in this case, the overlay-logic can quickly access
the DBMS, which we suppose to have its own cache. Finally, we note that the notable
case of batch insertion performed by one client described in Section~\ref{ssec:DBTree:upd-ins-del}
is essentially a special case of caching.


\section{Conclusions and Future Work}\label{sec:conclusions}

We showed how it is possible to support aggregate range queries efficiently on
conventional databases that do not implement specific optimizations. To do that,
we introduced the DB-tree: a new data structure that realizes a specific type of
index that is represented at the database level (an \emph{overlay-index}) 
and can be accessed by performing regular queries to the
DBMS. It can be also customized in many ways to support a wide class of
aggregation functions, authenticated data structures, and group-by range
queries. DB-trees can be queried downloading only $O(\log r)$ data (and hence taking $O(\log r)$ time), 
with $r$
the size of the data on which aggregation is performed, while common DBMSes
scan $O(r)$ data and thus take $O(r)$ time to answer the same queries. 
Experiments show that the improvement with respect to
the plain DBMS can be very large even for moderately large datasets. DB-trees
introduce some overhead for insertion and deletion, which was experimentally
measured to be a factor from 2 to 20 in our tests.

Regarding future research directions, from the theoretical
point of view, it would be interesting to investigate a \rnote{C4.3} bi/multi-dimensional
generalization of DB-trees to support speedup of aggregate range queries in GIS
systems. This generalization may be inspired to the work of Eppstein et
al.~\cite{eppstein2008skip} about skip quadtrees.

From the practical point of view, it would be desirable to have
a framework that streamlines the use of DB-trees in practical contexts, like
VerdictDB~\cite{park2018verdictdb} does for approximate query processing.
Further, experiments in a big-data context may be useful to better asses 
the spectrum of applicability of DB-trees.

\end{document}